\documentclass{article}
\usepackage{kr}

\usepackage[page,header]{appendix}
\usepackage{titletoc}

\usepackage{times}
\usepackage{soul}
\usepackage{url}
\usepackage[hidelinks]{hyperref}
\usepackage[utf8]{inputenc}
\usepackage[small]{caption}
\usepackage{graphicx}
\usepackage{amsmath}
\usepackage{amsthm}
\usepackage{booktabs}
\usepackage{algorithm}
\usepackage{algorithmic}
\newtheorem{example}{Example}

\newtheorem{lemma}{Lemma}
\newtheorem{proposition}{Proposition}
\newtheorem{property}{Property}
\newtheorem{definition}{Definition}

\usepackage{latexsym} 
\usepackage{amsfonts} 
\usepackage{ stmaryrd}
\usepackage{xspace}
\usepackage{multirow}
\usepackage{bbold}
\usepackage{tikz}
\usetikzlibrary{shapes.geometric}
\usepackage{enumitem}
\usepackage{caption}
\usepackage{subcaption}
\usepackage{thm-restate} 
\usepackage{nicematrix}
\usepackage{microtype}

\newcommand{\OMIT}[1]{}

\newcommand{\mn}[1]{\ensuremath{\mathsf{#1}}}
\newcommand{\mi}[1]{\ensuremath{\mathit{#1}}}

\newcommand{\Cmc}{\ensuremath{\mathcal{C}}}

\newcommand{\Imc}{\ensuremath{\mathcal{I}}}

\newcommand{\Smc}{\ensuremath{\mathcal{S}}}

\newcommand{\eg}{e.g.,~}

\newcommand{\ie}{i.e.,~}
\newcommand{\wrt}{w.r.t.~}
\newcommand{\cf}{cf.~}

\def\true{\ensuremath{\mathsf{true}}}
\def\false{\ensuremath{\mathsf{false}}}
\newcommand{\df}{:=}

\newcommand{\constSet}{\bf{C}}
\newcommand{\varSet}{\bf{V}}

\newcommand{\predSet}{\bf{P}}
\newcommand{\pred}{p}
\newcommand{\ontology}{\Sigma}
\newcommand{\database}{D}
\newcommand{\domain}{\mathcal{D}}
\newcommand{\inter}{I}
\newcommand{\ans}{\vec{a}}

\newcommand{\schema}{\mathcal{S}}

\newcommand{\drvtr}[1]{T^{\ontology}_{\database}(#1)}
\newcommand{\drvtrsig}{\drvtr{\fact}}

\newcommand{\dep}{\texttt{depth}}

\newcommand{\ann}{\Lambda}
\newcommand{\annt}[1]{\prov^{#1}}
\newcommand{\atann}{\annt{\texttt{AT}}}
\newcommand{\mtann}{\annt{\texttt{NRT}}}
\newcommand{\mdtann}{\annt{\texttt{MDT}}}
\newcommand{\rmdtann}{\annt{\texttt{HMDT}}}

\newcommand{\modannot}{\mu}
\newcommand{\provmodtot}{\annt{\texttt{AM}}}
\newcommand{\provmodmon}{\annt{\texttt{SAM}}}

\newcommand{\provnaive}{\annt{\texttt{NE}}}
\newcommand{\provsemi}{\annt{\texttt{SNE}}}
\newcommand{\provopti}{\annt{\texttt{OE}}}

\newcommand{\consop}{T_\ontology}
\newcommand{\diffop}{\Delta_\ontology}
\newcommand{\naive}{\inter_{\mn{n}}}
\newcommand{\semi}{\inter_{\mn{sn}}}
\newcommand{\opti}{\inter_{\mn{o,\fact}}}
\newcommand{\naiveannot}{\annot_{\mn{n}}}
\newcommand{\semiannot}{\annot_{\mn{sn}}}
\newcommand{\optiannot}{\annot_{\mn{o,\fact}}}

\newcommand{\zero}{0}
\newcommand{\one}{1}
\newcommand{\kzero}{0_\semiringshort}
\newcommand{\kone}{1_\semiringshort}
\newcommand{\semiringshort}{\mathbb{K}}
\newcommand{\semiring}{(K,\kplus,\ktimes,\zero_\semiringshort,\one_\semiringshort)}
\newcommand{\kplus}{+_\semiringshort}
\newcommand{\ktimes}{\times_\semiringshort}
\newcommand{\plus}{+}
\newcommand{\semiringVars}{\ensuremath{X}\xspace}

\newcommand{\annot}{\lambda}

\newcommand{\prov}{\mathcal{P}}

\newcommand{\monomial}{\ensuremath{m}\xspace}
\newcommand{\monom}{\ensuremath{\mn{mon}}\xspace}

\newcommand{\fact}{\alpha}
\newcommand{\datrule}{r}
\newcommand{\goal}{\mn{goal}}

\tikzset{
	itria/.style={
		draw,dashed,shape border uses incircle,
		isosceles triangle,shape border rotate=90,yshift=-1.45cm},
	rtria/.style={
		draw,dashed,shape border uses incircle,
		isosceles triangle,isosceles triangle apex angle=90,
		shape border rotate=-45,yshift=0.2cm,xshift=0.5cm},
	ritria/.style={
		draw,dashed,shape border uses incircle,
		isosceles triangle,isosceles triangle apex angle=110,
		shape border rotate=-55,yshift=0.1cm},
	letria/.style={
		draw,dashed,shape border uses incircle,
		isosceles triangle,isosceles triangle apex angle=110,
		shape border rotate=235,yshift=0.1cm}
}

\title{Revisiting Semiring Provenance for Datalog}

\author{%
Camille Bourgaux$^1$\and
Pierre Bourhis$^2$\and
Liat Peterfreund$^3$\and
Micha\"el Thomazo$^1$ \\
\affiliations
$^1$DI ENS, ENS, CNRS, PSL University \& Inria, Paris, France\\
$^2$CRIStAL, CNRS, University of Lille, Inria, Lille, France\\
$^3$LIGM, CNRS, Universit\'{e} Gustave Eiffel, ENPC, Paris, France\\
\emails
$^{1,2}$\{first.last\}@inria.fr, 
$^3$\{first.last\}@univ-eiffel.fr
}

\begin{document}

\maketitle

\begin{abstract}
Data provenance consists in bookkeeping meta information during query evaluation, in order to enrich query results with their trust level, likelihood, evaluation cost, and more. The framework of semiring provenance abstracts from the specific kind of meta information that annotates the data. While the definition of semiring provenance is uncontroversial for unions of conjunctive queries, the picture is less clear for Datalog. Indeed, the original definition might include infinite computations, and is not consistent with other proposals for Datalog semantics over annotated data. 
In this work, we propose and investigate several provenance semantics, based on different approaches for defining classical Datalog semantics. We study the relationship between these semantics, and introduce properties that allow us to analyze and compare them.
\end{abstract}

\section{Introduction}

Datalog is a rule language widely studied both in the database community, where it is seen as a query language, and in the 
KR community, as an ontology language. 

In relational databases, the framework of \emph{semiring provenance} was introduced to generalize computations over annotated databases, \eg the semantics of probabilistic databases \cite{DBLP:journals/sigmod/Senellart17}, the bag semantics, lineage or why-provenance \cite{DBLP:journals/ftdb/CheneyCT09}. 
In this framework, the semantics of positive relational algebra queries over databases annotated with elements of any commutative semiring is inductively defined on the structure of the query  \cite{greenpods2007,DBLP:conf/pods/GreenT17}. 
\emph{Provenance semirings} are  expressions (such as polynomials) built from variables associated to each tuple of the database \cite{containmentGreen09}. A provenance expression provides a general representation of how tuples have been used to derive a query result, and can be faithfully evaluated in any semiring in which the considered provenance semiring can be homomorphically embedded.

Semiring provenance has also been studied for Datalog queries, for which it was defined based on the set of all \emph{derivation trees} for the query \cite{greenpods2007,DeutchMRT14,DBLP:journals/vldb/DeutchGM18}. 
However, this definition 
seems less axiomatic than in the case of relational databases. 
Indeed, there may be infinitely many derivation trees, leading to infinite provenance expressions, while Datalog programs have finite models that can be computed efficiently~\cite{DBLP:books/aw/AbiteboulHV95}.  
A consequence is that this definition is valid only for a restricted class of semirings, namely $\omega$-continuous. 
Recently, \citeauthor{DannertGNT21} \shortcite{DannertGNT21} restrict the semiring even further by considering fully-chain complete semirings in order to extend provenance definition to logical languages featuring negation and fixed-point. 
Even if numerous useful semirings are $\omega$-continuous, or can be extended to a such semiring, infinite provenance expressions may be considered unintuitive in some cases. Consider, for example, the counting semiring (\ie natural numbers with standard operations) for which provenance of positive relational algebra queries corresponds to their bag semantics. This semiring can be extended to an $\omega$-continuous one by adding $\infty$ to the natural numbers, hence providing a way to capture the bag semantics for Datalog queries \cite{DBLP:conf/vldb/MumickPR90,greenpods2007}. However, query answers having infinite multiplicities may not seem very natural or informative. 
Moreover, alternative bag semantics for languages close to Datalog have been defined, and would not lead to such infinite multiplicities when applied to Datalog. This is in particular the case of the bag semantics for ontology-based data access \cite{DBLP:conf/ijcai/NikolaouKKKGH17,DBLP:journals/ai/NikolaouKKKGH19}, which corresponds to one of the two semantics proposed for source-to-target tuple generating dependencies in the context of data exchange \cite{DBLP:conf/lics/HernichK17}. 
Interestingly, these bag semantics are not based on derivation trees but are \emph{model-theoretic} semantics: they define annotated interpretations, and conditions for rules satisfaction over such interpretations. 
Such model-theoretic semantics have also been used in other contexts to evaluate Datalog and variants over annotated databases, such as fuzzy Datalog \cite{DBLP:journals/actaC/AchsK95} or description logic knowledge bases annotated with provenance tokens \cite{provdllite,provEL}. 
Finally, yet other semantics definitions have been proposed for some use cases. For instance, 
\citeauthor{DBLP:journals/toplas/ZhaoSS20} \shortcite{DBLP:journals/toplas/ZhaoSS20} consider \emph{minimal depth proof trees}, which correspond to a Datalog \emph{evaluation algorithm}, with the intended use of understanding the computation of the result, and guiding debugging. 

The fact that the above semantics are not encompassed by the definition of semiring provenance for Datalog, along with the need of handling infinite computations which are entailed by this definition, motivate us to 
investigate alternative natural semantics that
might be a better fit in different contexts. 

In this paper we introduce several natural provenance semantics for Datalog over annotated data.
Our definitions are based on different classical approaches: model-theoretic, execution-based and proof tree-based. 
They capture the semantics mentioned previously, and 
are inspired by practical needs. 
For instance, our semantics definition based on minimal depth derivation trees capture the behavior of Datalog engines, such as Souffl\'{e} \cite{souffle}, that store only minimal depth derivation trees instead of storing them all (which might be impossible in case there are infinitely many); Our semantics based on non-recursive derivation trees (in which a fact is not derived from itself) resembles the approach taken by some graph query languages, \eg SPARQL and Cypher, to handle queries with possibly infinite outputs by allowing to explicitly restrict the output to include only simple paths. 
In addition, some of the semantics suggested in the paper are closely related to paradigms for weighted reasoning in the context of words and trees \cite{10.1007/978-3-642-22944-2_2,STUBER2008221}.

After defining these different semantics, 
we study under which conditions they coincide and investigate their connections. 
We then provide a general framework for defining such provenance semantics, and present several properties 
relevant for provenance semantics 
that allow us to compare them. 
We briefly discuss some complexity issues in conclusion. 
Proofs and additional discussion are available in the appendix.

\section{Preliminaries}
\subsection{Datalog}
We use the standard Datalog settings (cf.~\cite{DBLP:books/aw/AbiteboulHV95} part D).
\subsubsection*{Syntax} Let $\predSet$, $\constSet$, and $\varSet$ be mutually disjoint, possibly infinite sets of \emph{predicates}, \emph{constants}, and \emph{variables} respectively. 
Elements of $\constSet \cup \varSet$ are called \emph{terms}. 
An \emph{atom} has the form $\pred(t_1,\dots,t_n)$ where $\pred\in\predSet$ is an $n$-ary predicate, and  $t_i$'s are terms. 
A \emph{fact} (or \emph{ground atom}) is a variable-free atom. 
A \emph{(Datalog) rule} is an expression:
$\forall \vec{x}\forall\vec{y}(\phi(\vec{x},\vec{y} ) \rightarrow \psi(\vec{x}))$
where $\vec{x}$ and $\vec{y}$ are tuples of variables and $\phi(\vec{x},\vec{y} )$ and $\psi(\vec{x})$ are conjunctions of atoms 
whose variables are 
$\vec{x}\cup\vec{y}$ and $\vec{x}$ respectively.  
We call $\phi(\vec{x},\vec{y} )$ and $\psi(\vec{x})$  the \emph{body} and \emph{head} of the rule, respectively. 
From now on, we assume that rules are in \emph{normalized form}, \ie 
the head 
consists of a single atom $H(\vec{x})$, and
quantifiers are implicit. 
The \emph{domain} $\domain(\mathcal A)$ of a set $\mathcal A$ of atoms is the set of terms that appear in its atoms. 

A \emph{database} $\database$ is a finite set of facts, and 
a \emph{Datalog program (or ontology)} $\ontology$ is a finite set of Datalog rules. 
The \emph{schema} of 
$\database$ (resp.\ $\ontology$) 
denoted $\schema(\database)$ (resp.\ $\schema(\ontology)$) is the set of predicates that appear in its atoms.\footnote{
	Note that 
	we do not require the set of predicates of atoms appearing in heads of rules to be disjoint from  $\schema(\database)$; naturally, all of our results are valid under this assumption as well.}

\subsubsection*{Semantics} 
The semantics of Datalog can classically be defined in three ways: through models, fixpoints or derivation trees. All three definitions rely on the notion of homomorphism: 
a \emph{homomorphism} from a set $\mathcal A$ of atoms to a set $\mathcal B$ of atoms is a  function $h: \domain(\mathcal A)\rightarrow \domain(\mathcal B)$ such that $h(t)=t$ for all $t\in\constSet$, and 
$\pred(t_1,\dots, t_n)\in\mathcal A$ implies 
$h(\pred(t_1,\cdots, t_n))\df\pred(h(t_1),\dots, h(t_n))\in \mathcal B.$ 
We denote by $h(\mathcal A)$ the set $\{h(\pred(t_1,\dots, t_n))\mid \pred(t_1,\dots, t_n)\in \mathcal A\}$. 
The homomorphism definition is extended to 
conjunctions of atoms 
by viewing 
them as the sets of atoms they contain. 

A 
set $\inter$ of facts is a \emph{model} of a rule $\datrule \df \phi(\vec{x},\vec{y} ) \rightarrow \psi(\vec{x})$, denoted by $\inter\models \datrule$, if every 
homomorphism $h$ from $\phi(\vec{x},\vec{y})$ to $\inter$ is also a homomorphism from $\psi(\vec{x})$ to $\inter$; 
it is a \emph{model} of a Datalog program $\ontology$ if $\inter\models\datrule$ for every $\datrule\in\ontology$; 
it is a \emph{model} of a database $\database$ if $\database \subseteq \inter$. 
A fact $\fact$ is \emph{entailed} by $\database$ and $\ontology$, denoted $\ontology,\database\models\fact$, if  $\fact\in I$ for every model $I$ of $\ontology$ and $\database$. 
\begin{example}\label{ex:running}
	Let
	$\ontology$ contain the rules $B(x)\rightarrow A(x)$,  $R(x,y)\wedge A(y)\rightarrow B(x)$, and $R(x,y)\rightarrow  R(y,x)$, and 
	$\database \df \{B(a), B(b), R(a,b),  R(b,a)\}$. 
	Each model of $\database$ and $\ontology$ contains all facts in $\database$ as well as $A(a)$ and $A(b)$, which are thus entailed by $\ontology,\database$. 
\end{example}

An equivalent way to define the entailment of a fact $\fact$ by $\database$ and $\ontology$ is to check if there is a homomorphism from $\fact$ to a specific model, defined as the \emph{least fixpoint} containing $\database$ of the immediate consequence operator: An immediate consequence for $\database$ and $\ontology$ is either $\fact \in \database$, or $\alpha$ such that there exists a rule $\datrule \df \phi(\vec{x},\vec{y} ) \rightarrow \psi(\vec{x})$ and a homomorphism $h$ from $\phi(\vec{x},\vec{y} )$ to $\database$ such that $h(\psi(\vec{x})) = \fact$.

Finally, a third definition relies on 
derivation trees.

\begin{definition}[Derivation Tree]
	A \emph{derivation tree} $t$ of a fact $\fact$ \wrt a database $\database$ and a program $\ontology$ is a finite tree whose leaves are labeled by facts from $\database$ and 
	non-leaf nodes are labeled by triples $(p(t_1,\ldots, t_m), \datrule, h)$ where 
	\begin{itemize}
		\item $p(t_1,\ldots, t_m)$ is a fact over the schema $\schema(\ontology)$;
		\item $\datrule$ is a rule from $\ontology$ of the form  $\phi(\vec{x},\vec{y} ) \rightarrow  p(\vec{x})$;
		\item $h$ is a homomorphism from $\phi(\vec{x},\vec{y} )$ to the facts of the labels of the node children, such that $h(p(\vec{x})) = p(t_1,\ldots,t_m)$;
		\item there is a bijection $f$ between the node children and the atoms of $\phi(\vec{x},\vec{y} )$, such that for every $q(\vec{z})\in \phi(\vec{x},\vec{y} )$, $f(q(\vec{z}))$ is of the form $(h(q(\vec{z})),\datrule',h')$ or is a leaf labeled by $h(q(\vec{z}))$. 
	\end{itemize}
	Moreover, if $(p(t_1,\cdots, t_m), \datrule, h)$ or $p(t_1,\cdots, t_m)$ is the root of $t$, then $p(t_1,\cdots, t_m)=\fact$. 
\end{definition}

\begin{example}\label{ex:trh}
	Let $\ontology$ contain $\datrule_1 \df R(x,y)\rightarrow H(x,x)$, $\datrule_2 \df R(x,y)\rightarrow H(x,y)$ and $\datrule_3\df S(x,y,z) \wedge S(x,z,y)\rightarrow H(x,x)$. If $D=\{R(a,a), S(a,b,c),S(a,c,b)\}$, then the fact $\fact\df H(a,a)$ has the following derivation trees
	\noindent
	\begin{tabular}{llll}
		\begin{tikzpicture}
			[level distance=0.75cm]
			\node {$(\fact,\datrule_1,h)$}
			child {node {$R(a,a)$}};
		\end{tikzpicture}
		&\hspace{-0.6cm}
		\begin{tikzpicture}
			[level distance=0.75cm]
			\node {$(\fact,\datrule_2,h)$}
			child {node {$R(a,a)$}};
		\end{tikzpicture}
		&\hspace{-0.8cm}
		\begin{tikzpicture}
			[level distance=0.75cm,
			level 1/.style={sibling distance=1.3cm}]
			\node {$(\fact,\datrule_3,h_3)$}
			child {node {$S(a,b,c)$}
			}
			child {node {$S(a,c,b)$}
			};
		\end{tikzpicture}
		&\hspace{-0.7cm}
		\begin{tikzpicture}
			[level distance=0.75cm,
			level 1/.style={sibling distance=1.3cm}]
			\node {$(\fact,\datrule_3,h_3')$}
			child {node {$S(a,c,b)$}
			}
			child {node {$S(a,b,c)$}
			};
		\end{tikzpicture}
	\end{tabular}
	where $h(x) = h(y) = a$, $h_3(x)=a$, $h_3(y)=b$, $h_3(z)=c$ and $h'_3(x)=a$, $h'_3(y)=c$, $h'_3(z)=b$.
\end{example}
Note that when the program at hand is recursive (i.e., the dependency graph of its predicates 
contains cycles) 
a fact may have infinitely many derivation trees. Figure~\ref{fig:derivtree} depicts some of the infinitely many derivation trees of $A(a)$ from Example~\ref{ex:running}.
In this example, and from this point on, we omit rules and homomorphisms from trees when there is no ambiguity.

\begin{figure}[t]	
	\centering
	\begin{tikzpicture}
		[level distance=2.1cm,
		level 1/.style={sibling distance=1.8cm},
		level 2/.style={sibling distance=0.9cm}]
		\node {$A(a)$}
		child {node {$B(a)$}};
	\end{tikzpicture}
	\hspace{0.7cm}
	\begin{tikzpicture}
		[level distance=0.7cm,
		level 1/.style={sibling distance=2.0cm},
		level 2/.style={sibling distance=1.2cm}]
		\node {$A(a)$}
		child {node {$B(a)$}
			child {node {$R(a,b)$}}
			child {node {$A(b)$}	
				child {node {$B(b)$}			
		}}};
	\end{tikzpicture}
	\hspace{0.7cm}
	\begin{tikzpicture}
		[level distance=0.7cm,
		level 1/.style={sibling distance=2.0cm},
		level 2/.style={sibling distance=1.2cm}]
		\node {$A(a)$}
		child {node {$B(a)$}
			child {node {$R(a,b)$} child {node {$R(b,a)$}}}
			child {node {$A(b)$}	
				child {node {$B(b)$}			
		}}};
	\end{tikzpicture}	
	\caption{
		Some derivation trees of $A(a)$ in Example~\ref{ex:running}. 
		\label{fig:derivtree}}
\end{figure}
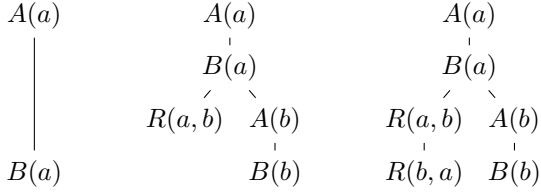

\subsubsection*{Queries} A \emph{conjunctive query} (CQ) is an existentially quantified formula $\exists \vec{y}\, \phi(\vec{x}, \vec{y})$ where $\phi(\vec{x}, \vec{y})$ is a conjunction of atoms with variables in $\vec{x}\cup\vec{y}$; 
a \emph{union of conjunctive queries} (UCQ) is a disjunction of CQs (over the same free variables). A query is \emph{Boolean} if it has no free-variables. 
A set of facts $\inter$ satisfies a Boolean CQ (BCQ) $q\df\exists \vec{y}\, \phi(\vec{y})$, written $\inter\models q$, if and only if there is a homomorphism from $\phi(\vec{y})$ to $\inter$. 
A BCQ $q$ is \emph{entailed} by a Datalog program $\ontology$ and database $\database$, written $\ontology,\database\models q$, if and only if $\inter\models q$ for every model $\inter$ of $\ontology$ and $\database$. 
Note that $\ontology,\database\models q$ if and only if $\ontology\cup\{\phi(\vec{y})\rightarrow \goal\},\database\models\goal$, where $\goal$ is a nullary predicate such that 
$\goal\notin\schema(\ontology)\cup\schema(\database)$. 
A tuple of constants $\ans$ is an \emph{answer} to a CQ $q(\vec{x})\df\exists \vec{y}\, \phi(\vec{x}, \vec{y})$ over $\ontology$ and $\database$ if $\ans$ and $\vec{x}$ have the same arity and $\ontology,\database\models q(\ans)$ where $q(\ans)$ is the BCQ obtained by replacing the variables from $\vec{x}$ with the corresponding constants from $\ans$. 
When $\ontology=\emptyset$, it amounts to the existence of a homomorphism from $q(\ans)$ 
to $\database$, which corresponds to the semantics of CQs over relational~databases.

\subsection{Annotated Databases}  
To equip databases with extra information, their facts might be annotated with, e.g., trust levels, clearance degree required to access them, or identifiers to track how they are used.

In the framework of semiring provenance, 
annotations are elements of 
algebraic structures known as commutative semirings. 
A \emph{semiring}\sloppy{ $\semiringshort=\semiring$ } is a set $K$ with distinguished elements $\kzero$ and $\kone$, equipped with two binary operators: $\kplus$, called the \emph{addition}, which is an associative and commutative operator with identity~$\kzero$, and $\ktimes$, called the \emph{multiplication}, which is an associative operator with identity~$\kone$. 
It also holds that $\ktimes$ distributes over $\kplus$, and $\kzero$ is annihilating for~$\ktimes$. 
When multiplication is commutative, the semiring is said to be 
\emph{commutative}. 
We use the convention according to which multiplication is applied before addition to omit parentheses. 
We omit the subscript of operators and distinguished elements when there is no~ambiguity.

\begin{definition}
	An \emph{annotated database} is a triple $(\database, \semiringshort, \annot)$ where $\database$ is a database, $\semiringshort = \semiring$ is a semiring, and $\annot : \database \mapsto K \setminus \{\kzero\}$ maps facts into semiring elements different from $\kzero$. 
\end{definition}

\begin{example}[Ex. \ref{ex:running} cont'd]\label{ex:running-annot}
	The semiring $\mathbb{N}=(\mathbb{N}, \plus, \times, 0,1)$ of the natural numbers equipped with the usual operations is used for bag semantics. The tropical semiring $\mathbb{T}=(\mathbb{R}^\infty_+, \mn{min}, +, \infty, 0)$ is used to compute minimal-cost paths.
	We define $\annot_\mathbb{N}:\database\mapsto \mathbb{N}\setminus\{0\}$ by
	$\annot_\mathbb{N}(B(a))=3$,  $\annot_\mathbb{N}(B(b))=1$,   $\annot_\mathbb{N}(R(a,b))=2$,  $\annot_\mathbb{N}(R(b,a))=1$;
	And $\annot_\mathbb{T}:\database\mapsto \mathbb{R}_+$ by
	$\annot_\mathbb{T}(B(a))=10$,  $\annot_\mathbb{T}(B(b))=1$, $\annot_\mathbb{T}(R(a,b))=5$,  $\annot_\mathbb{T}(R(b,a))=2$.
\end{example}

We next list some possible properties of semirings. 
A semiring is \emph{$\plus\,$-idempotent} (resp.\  \emph{$\times$-idempotent}) if for every $a\in K$, $a\plus a=a$ (resp.\ $a\times a=a$). 
It is \emph{absorptive} if for every $a,b\in K$, $a \times b \plus a= a$. 
It is \emph{positive} if for every $a,b\in K$, $a \times b =0$ if and only if ($a=0$ or $b=0$), and $a \plus b =0$ if and only if $a=b=0$. 
Finally, an important class is that of $\omega$-continuous commutative semirings in which infinite sums are well-defined.  
Given a semiring, we define the binary relation $\sqsubseteq$ such that $a\sqsubseteq b$ if and only if there exists $c\in K$ such that $a\plus c =b$. 
A commutative semiring is \emph{$\omega$-continuous} if $\sqsubseteq$ is a partial order, every (infinite) $\omega$-chain $a_0\sqsubseteq a_1\sqsubseteq a_2\dots$ has a least upper bound $\sup((a_i)_{i\in \mathbb{N}})$, and for every $a$, $a\plus \sup((a_i)_{i\in \mathbb{N}})=\sup((a\plus a_i)_{i\in \mathbb{N}})$ and $a\times \sup((a_i)_{i\in \mathbb{N}})=\sup((a\times a_i)_{i\in \mathbb{N}})$.

The semantics of queries from the positive relational algebra, and in particular of UCQs, over annotated databases is defined inductively on the structure of the query \cite{greenpods2007}. Intuitively, joint use of data (conjunction) corresponds to multiplication, and alternative use of data (union or projection) corresponds to addition. 

\begin{example}[Ex. \ref{ex:running-annot} cont'd]
	\sloppy{The BCQ $\exists xy \,(R(x,y)\wedge B(y))$ is entailed from $(\database, \mathbb{N}, \annot_\mathbb{N})$ with multiplicity $\annot_\mathbb{N}(R(a,b)) \times \annot_\mathbb{N}(B(b))\plus\annot_\mathbb{N}(R(b,a)) \times \annot_\mathbb{N}(B(a))=5$,
		and from $(\database, \mathbb{T}, \annot_\mathbb{T})$ with minimal cost $\mn{min}(\annot_\mathbb{T}(R(a,b)) + \annot_\mathbb{T}(B(b)), \annot_\mathbb{T}(R(b,a)) + \annot_\mathbb{T}(B(a))) = 6$. }
\end{example}

A semantics of Datalog over annotated databases has been defined by \citeauthor{greenpods2007} \shortcite{greenpods2007} using derivation trees, that we shall name the \emph{all-tree semantics}. 
It associates to each fact $\fact$ entailed by $\ontology$ and $\database$ the following sum, where $\drvtrsig$ is the set of all derivation trees for $\fact$ \wrt $\ontology$ and $\database$ and $\ann(t) := \prod_{v\text{ is a leaf of $t$}} \annot(v)$ is the $\semiringshort$-annotation of the derivation tree $t$ (since $\semiringshort$ is commutative, the result of the product is well-defined). 
\[
\atann(\ontology,\database,\semiringshort,\annot,\fact):= \sum_{ t\in  \drvtrsig}   \ann(t).
\]
Since $\drvtrsig$ may be infinite, $\atann$ is well-defined for all $\ontology$, $(\database,\semiringshort,\annot)$ and $\alpha$ only in the case where $\semiringshort$ is $\omega$-continuous.

\begin{example}[Ex. \ref{ex:running-annot} cont'd] The fact
	$\alpha\df A(a)$ is entailed 
	with minimal cost: $\atann(\ontology,\database, \mathbb{T}, \annot_\mathbb{T},\fact)=\mn{min}_{t \in \drvtrsig} \Sigma_{v\text{ is a leaf of $t$}} \annot_\mathbb{T}(v) = 3$. 
	Since $\mathbb{N}$ is not $\omega$-continuous, $\atann(\ontology,\database, \mathbb{N}, \annot_\mathbb{N},\fact)$ is not defined. 
\end{example}

\subsection{Provenance Semirings} 
Provenance semirings have been introduced to abstract from a particular semiring by associating a unique provenance token to each fact of the database, and building expressions that trace their use. 
Given a set $\semiringVars$ of \emph{variables} that annotate the database, a \emph{provenance semiring} $\mi{Prov}(\semiringVars)$ is a semiring over a space of provenance expressions with variables from~$\semiringVars$. 

Various such semirings were introduced in the context of relational databases \cite{containmentGreen09}: The most expressive annotations are provided by the \emph{provenance polynomials} semiring \sloppy{ $\mathbb{N}[\semiringVars] \df (\mathbb{N}[\semiringVars], + ,\times,0,1)$ } of polynomials with coefficients from $\mathbb{N}$ and variables from $\semiringVars$, and the usual operations. 
Less general provenance semirings include, for example, the semiring $\mathbb{B}[\semiringVars]\df (\mathbb{B}[\semiringVars], + ,\times,0,1)$ of polynomials with Boolean coefficients, and the semiring $\mi{PosBool}(\semiringVars)\df(\mi{PosBool}(\semiringVars), \vee ,\wedge,\false,\true)$ of positive Boolean expressions.

In the Datalog context, it is important to allow for infinite provenance expressions, as there can be infinitely many derivation trees.  
A \emph{formal power series} with variables from $\semiringVars$ and coefficients from $K$ is a mapping that associates to each monomial over $\semiringVars$ a coefficient in $K$. A formal power series $S$ can be written as a possibly infinite sum $S=\Sigma_{\monomial\in\monom(\semiringVars)}S(\monomial)\monomial$ where $\monom(\semiringVars)$ is the set of monomials over $\semiringVars$ and $S(\monomial)$ is the coefficient of the monomial $\monomial$. 
The set of formal power series with variables from $\semiringVars$ and coefficients from $K$ is denoted $K\llbracket\semiringVars\rrbracket$. \citeauthor{greenpods2007} \shortcite{greenpods2007} define the \emph{Datalog provenance semiring} as the semiring $\mathbb{N}^\infty\llbracket\semiringVars\rrbracket$ of formal power series with coefficients from $\mathbb{N}^\infty=\mathbb{N}\cup\{\infty\}$. 

A \emph{semiring homomorphism} from $\semiringshort = \semiring$ to $\mathbb{K'} = (K', +_\mathbb{K'} ,\times_\mathbb{K'}, 0_\mathbb{K'}, 1_\mathbb{K'})$ is a mapping $h : K \rightarrow K'$ such that $h(\kzero) = 0_\mathbb{K'}$, $h(\kone) = 1_\mathbb{K'}$, and for all $a,b\in K$, $h(a \kplus b) = h(a) +_\mathbb{K'} h(b)$ and $h(a \ktimes b) = h(a) \times_\mathbb{K'} h(b)$. 
A semiring homomorphism between $\omega$-continuous semirings is \emph{$\omega$-continuous} if it preserves least upper bounds: $h(\sup((a_i)_{i\in \mathbb{N}})) =\sup((h(a_i))_{i\in \mathbb{N}})$.

Following \citeauthor{DeutchMRT14} \shortcite{DeutchMRT14}, we say that a provenance semiring $\mi{Prov}(\semiringVars)$ \emph{specializes} correctly to a semiring $\semiringshort$, if any valuation $\nu : \semiringVars \rightarrow K$ extends uniquely to a ($\omega$-continuous if $\mi{Prov}(\semiringVars)$ and $\semiringshort$ are $\omega$-continuous) semiring homomorphism $h : \mi{Prov}(\semiringVars) \rightarrow K$, allowing the computations for $\semiringshort$ to factor through the computations for $\mi{Prov}(\semiringVars)$. 
A provenance semiring $\mi{Prov}(\semiringVars)$ is \emph{universal for a set of semirings} if it specializes correctly to each semiring of this set. 
\citeauthor{greenpods2007} \shortcite{greenpods2007} showed that $\mathbb{N}[\semiringVars]$ is universal for commutative semirings, and $\mathbb{N}^\infty\llbracket\semiringVars\rrbracket$ is universal for commutative $\omega$-continuous semirings.

\section{Alternative Semantics}\label{sec:semantics}

In this section we propose several natural ways of defining the semantics of Datalog over annotated databases, and investigate their connections. 
We have seen that the semantics of Datalog can equivalently be defined through models, fixpoints or derivation trees. 
The semantics we propose also fall into these three approaches. 
For presentation purposes, 
we see each semantics as a partial function $\prov$ that associates to a Datalog program $\ontology$, annotated database $(\database, \semiringshort, \annot)$, and fact $\fact$, a semiring element $\prov(\ontology, \database, \semiringshort, \annot,\fact)$.

\subsection{Model-Based Semantics}\label{subsec:model-based-sem}
We first investigate two provenance semantics based on Datalog's model-theoretic semantics. 
In both cases, we will define interpretations $(\inter,\modannot^I)$ where $I$ is a set of facts and $\modannot^I$ is a function that annotates facts of $I$, and formulate requirements for them to be models of $\ontology$ and $(\database,\semiringshort,\annot)$, extending standard models of $\ontology$ and $\database$ with fact annotations. 

\subsubsection*{Annotated Model-based} 
\citeauthor{DBLP:conf/lics/HernichK17} \shortcite{DBLP:conf/lics/HernichK17} define two bag semantics in the context of data exchange: 
the \emph{incognizant} and \emph{cognizant} semantics. 
The difference between them arise from the two different semantics of bag union: the incognizant semantics uses the maximum-based union, while the cognizant semantics uses the sum-based union. 

In more details, both semantics are based on the following semantics for source-to-target tuple generating dependencies (s-t tgds): a pair $(I,J)$ of source and target instances satisfies an s-t tgd $q_1(\vec{x})\rightarrow q_2(\vec{x})$ if 
for every answer $\ans$ to $q_1$ over $I$, $\ans$ is an answer to $q_2$ over $J$ with at least the same  multiplicity. 
Given a set of s-t tgds $\ontology$ and a source $I$, a target $J$ is an \emph{incognizant solution} for $I$ \wrt $\ontology$ if $(I,J)$ satisfies every s-t tgd in $\ontology$. 
It is a \emph{cognizant solution} if for every $\datrule\in\ontology$, there is a target instance $J_\datrule$ such that $(I,J_\datrule)$ satisfies $\datrule$ and $\uplus J_{\datrule}\subseteq J$, where $\uplus$ denotes the sum-union of bags (\ie the multiplicity of each element of the sum-union is equal to the sum of its multiplicities). 
The incognizant (resp.\ cognizant) \emph{certain answers} to a query $q$ \wrt $\ontology$ on $I$ are defined using bag intersection of the answers over the incognizant (resp.\ cognizant) solutions for $I$ \wrt $\ontology$, \ie the multiplicity of an answer is the minimum of its multiplicities over the solutions. Note that for BCQs, the only possible certain answer is the empty tuple.

For example,  consider $\ontology=\{B(x) \rightarrow A(x), C(x) \rightarrow A(x)\}$ and $\database=\{(B(a),1),(C(a),1)\}$. 
Under the incognizant semantics, the multiplicity of the certain answer of the Boolean query $A(a)$ \wrt $\ontology$ and $\database$ is $1$ while under the cognizant semantics it is $2$. 
Indeed, $J=\{(A(a),1)\}$ 
is an incognizant solution for $\database$ \wrt $\ontology$ as it satisfies both s-t tgds, but is not a cognizant solution as the sum of multiplicities that arise from the two rules is $2$. 

It is easy to show that the cognizant semantics is equivalent to $\atann$ on the counting semiring 
$\mathbb{N}=(\mathbb{N}, \plus, \times, 0,1)$, and thus coincides with the classical bag semantics for Datalog. 
However, we have seen 
that the incognizant and cognizant semantics differ. 
Moreover, note that in the field of ontology-based data access, the bag semantics defined by \citeauthor{DBLP:journals/ai/NikolaouKKKGH19} for DL-Lite$_R$ \shortcite{DBLP:conf/ijcai/NikolaouKKKGH17,DBLP:journals/ai/NikolaouKKKGH19} coincides with the \emph{incognizant} semantics, thus disagrees with the classical Datalog bag semantics \cite{DBLP:conf/vldb/MumickPR90,greenpods2007}.

We hence define a provenance semantics that coincides  with these semantics when used with the counting semiring. 
Since it is based on greatest lower bounds, it is defined on a restricted class of semirings. 

Let $\semiringshort=(K,\plus,\times,\zero,\one)$ be a commutative $\omega$-continuous semiring such that for every $K'\subseteq K$, the greatest lower bound $\inf(K')$ of $K'$ is well defined (\ie there exists a unique $z\in K$ such that $z\sqsubseteq x$ for every $x\in K'$ and every $z'$ such that $z'\sqsubseteq x$ for every $x\in K'$ is such that $z'\sqsubseteq z$), $\ontology$ be a Datalog program, and $(\database,\semiringshort, \annot)$ be an annotated database. We define 
\emph{$\semiringshort$-annotated interpretations} as pairs $(\inter,\modannot^I)$ where $\inter$ is a set of facts, and $\modannot^I$ is a function from $\inter$ to $K$. 
We say that a 
$\semiringshort$-annotated interpretation
$(\inter,\modannot^I)$ is a \emph{model} of $\ontology$ and $ (\database,\semiringshort, \annot)$, denoted by  $(\inter,\modannot^I)\models (\ontology,\database,\semiringshort, \annot)$, if 
\begin{enumerate}
	\item $\database\subseteq \inter$, and for every $\fact\in\database$, $\annot(\fact)\sqsubseteq \modannot^I(\fact)$;
	\item for every $\phi(\vec{x},\vec{y} ) \rightarrow H(\vec{x})$ in $\ontology$, whenever there is a homomorphism $h:\phi(\vec{x},\vec{y})\mapsto \inter$, 
	then \mbox{$h(H(\vec{x}))\in\inter$} and \mbox{$\displaystyle\sum_{h':\phi(\vec{x},\vec{y})\mapsto I, h'(\vec{x})=h(\vec{x})} \prod_{\beta\in h'(\phi(\vec{x},\vec{y}))} \modannot^I(\beta) \sqsubseteq \modannot^I(h(H(\vec{x})))$.}
\end{enumerate}

The \emph{annotated model-based provenance semantics} $\provmodtot$ is defined by 
$$\provmodtot(\ontology, \database, \semiringshort, \annot,\fact) \df \inf (\{\modannot^I(\fact){\mid}(\inter,\modannot^I){\models}(\ontology,\database,\semiringshort, \annot)\}).$$

\begin{restatable}{proposition}{propprovmodtotcaptureincognizantandDL}\label{prop:provmodtot-capture-incognizant-and-DL}
	If the Datalog rules in $\ontology$ are (1) s-t tgds, or (2) formulated in DL-Lite$_R$, then for every BCQ $q$, $\provmodtot(\ontology\cup\{q\rightarrow \mn{goal}\},\database,\mathbb{N},\annot_\mathbb{N},\mn{goal})$ is equal to the multiplicity of the empty tuple in (1) the incognizant certain answers or (2) the bag certain answers to $q$ \wrt $\ontology$ and $(\database,\mathbb{N},\annot_\mathbb{N})$.
\end{restatable}

\subsubsection*{Set-Annotated Model-based} 
We adapt the work on provenance for the description logics DL-Lite$_R$ and $\mathcal{ELH}^r$ \cite{provdllite,provEL}, where the semiring is assumed to be a $\times$-idempotent provenance semiring $\mi{Prov}(\semiringVars)$ and rules are also annotated. 
Annotated models of annotated knowledge bases are defined as set of facts annotated with sets of monomials from $\mi{Prov}(\semiringVars)$. 
Given a fact $\fact$ and a monomial $m$ over $\semiringVars$, $(\ontology,\database,\mi{Prov}(\semiringVars),\annot_\semiringVars)\models (\fact,m) $ holds when $m$ belongs to the annotation set of $\fact$ in every models of $\ontology$ and $(\database,\mi{Prov}(\semiringVars),\annot_\semiringVars)$. 

To obtain an analog provenance semantics for Datalog, we define interpretations which associate facts with (possibly infinite) \emph{sets of annotations}, and formulate the requirements for them to be models of $\ontology$ and $(\database,\semiringshort,\annot)$.

Let $\semiringshort=(K,\plus,\times,\zero,\one)$ be a commutative $\omega$-continuous semiring, $\ontology$ be a Datalog program, and $(\database,\semiringshort, \annot)$ be an annotated database. We define 
\emph{$\semiringshort$-set-annotated interpretations} as pairs $(\inter,\modannot^I)$ where $\inter$ is a set of facts, and $\modannot^I$ is a function from $\inter$ 
to the \emph{power-set of $K$}. 
We say that a $\semiringshort$-set-annotated interpretation
$(\inter,\modannot^I)$ is a model of $\ontology$ and $ (\database,\semiringshort, \annot)$, denoted by  $(\inter,\modannot^I)\models (\ontology,\database,\semiringshort, \annot)$, if 
\begin{enumerate}
	\item$\database\subseteq \inter$, and for every $\fact\in\database$, $\annot(\fact)\in\modannot^I(\fact)$;
	\item   for every $\phi(\vec{x},\vec{y} ) \rightarrow H(\vec{x})$ in $\ontology$, whenever there is a homomorphism $h:\phi(\vec{x},\vec{y})\mapsto \inter$, 
	then \mbox{$h(H(\vec{x}))\in\inter$} and if $h(\phi(\vec{x},\vec{y} ))=\beta_1\wedge\dots\wedge\beta_n$, $\{\Pi_{i=1}^n k_i \mid (k_1,\dots,k_n)\in\modannot^I(\beta_1)\times\dots\times\modannot^I(\beta_n)\}\subseteq \modannot^I(h(H(\vec{x})))$. 
\end{enumerate}

The \emph{set-annotated model-based provenance semantics} $\provmodmon$ is defined by
\[\provmodmon(\ontology, \database, \semiringshort, \annot,\fact):= \sum_{k\in \bigcap_{(\inter,\modannot^I)\models (\ontology,\database,\semiringshort, \annot)} \modannot^I(\fact)} k.\]

\subsubsection*{Connections between semantics}
Let $\sqsubseteq$ be the binary relation between provenance semantics such that $\prov\sqsubseteq\prov'$ if and only if $\prov(\ontology,\database,\semiringshort,\annot,\fact)\sqsubseteq\prov'(\ontology,\database,\semiringshort,\annot,\fact)$ for every $\ontology$, $(\database,\semiringshort,\annot)$ and $\alpha$ on which $\prov$ and $\prov'$ are well-defined. 

\begin{restatable}{proposition}{propmodelBasedRelationship}\label{prop:modelBasedRelationship}
	The following holds: $$\provmodtot\sqsubseteq \atann\text{ and }\provmodmon\sqsubseteq \atann.$$
\end{restatable}

Next examples show that $\provmodtot$ and $\provmodmon$ are incomparable.

\begin{example}
	Let $\ontology=\{A(x)\rightarrow \mn{goal}, B(x)\rightarrow \mn{goal}\}$, $\database=\{A(a), B(a)\}$, 
	$\annot_\mathbb{N}(A(a))=2$ and $\annot_\mathbb{N}(B(a))=3$. 
	
	Annotated models of $\ontology$ and $(\database,\mathbb{N},\annot_\mathbb{N})$ are such that $\modannot^I(\mn{goal})\geq 3$, so $\provmodtot(\ontology,\database,\mathbb{N},\annot_\mathbb{N},\mn{goal})=3$.

	Set-annotated models of $\ontology$ and $(\database,\mathbb{N},\annot_\mathbb{N})$ are such that $\{2,3\}\subseteq \modannot^I(\mn{goal})$, so $\provmodmon(\ontology,\database,\mathbb{N},\annot_\mathbb{N},\mn{goal})=5$.
	
	Hence $\provmodtot(\ontology,\database,\mathbb{N},\annot_\mathbb{N},\mn{goal}){<}\provmodmon(\ontology,\database,\mathbb{N},\annot_\mathbb{N},\mn{goal})$. 
\end{example}
\begin{example}
	Let $\ontology=\{R(x,y)\rightarrow \mn{goal}\}$, $\database=\{R(a,b), R(a,c)\}$, 
	$\annot_\mathbb{N}(R(a,b))=2$ and $\annot_\mathbb{N}(R(a,c))=2$. 
	
	Annotated models of $\ontology$ and $(\database,\mathbb{N},\annot_\mathbb{N})$ are such that $\modannot^I(\mn{goal})\geq 4$, so $\provmodtot(\ontology,\database,\mathbb{N},\annot_\mathbb{N},\mn{goal})=4$.

	Set-annotated models of $\ontology$ and $(\database,\mathbb{N},\annot_\mathbb{N})$ are such that $\{2\}\subseteq \modannot^I(\mn{goal})$, so $\provmodmon(\ontology,\database,\mathbb{N},\annot_\mathbb{N},\mn{goal})=2$.
	
	Hence $\provmodtot(\ontology,\database,\mathbb{N},\annot_\mathbb{N},\mn{goal}){>}\provmodmon(\ontology,\database,\mathbb{N},\annot_\mathbb{N},\mn{goal})$. 
\end{example}

Despite of their inherently different approaches, 
$\provmodtot$, $\provmodmon$ and $\atann$ coincide on a large class of semirings. 
\begin{restatable}{proposition}{propmodelBasedSameAsTreeBasedForIdempotent}\label{prop:modelBasedSameAsTreeBasedForIdempotent}
	If $\semiringshort$ is a commutative $\plus\,$-idempotent $\omega$-continuous semiring, then for every $\ontology$, $(\database,\semiringshort, \annot)$, and~$\fact$, 
	$\provmodtot(\ontology, \database, \semiringshort, \annot,\fact)=\provmodmon(\ontology, \database, \semiringshort, \annot,\fact)= \atann(\ontology, \database, \semiringshort, \annot,\fact).$
\end{restatable}

Additional insights on the connection between definitions can be gained by considering the provenance semiring  $\mathbb{N}^\infty\llbracket\semiringVars\rrbracket$: the monomials with non-zero coefficients are the same with all semantics but their coefficients may differ ($\atann$ leading to the highest coefficients by Proposition~\ref{prop:modelBasedRelationship}). 

\begin{restatable}{proposition}{propmodelBasedVSTreeBased}\label{prop:modelBasedVSTreeBased}
	Let $\annot_\semiringVars$ be an injective function from $\database$ to $\semiringVars$. 
	\begin{itemize}
		\item A monomial occurs in $\atann(\ontology, \database, \mathbb{N}^\infty\llbracket\semiringVars\rrbracket, \annot_\semiringVars,\fact)$ if and only if it occurs in $\provmodtot(\ontology, \database, \mathbb{N}^\infty\llbracket\semiringVars\rrbracket, \annot_\semiringVars,\fact)$. 
		\item $\provmodmon(\ontology, \database, \mathbb{N}^\infty\llbracket\semiringVars\rrbracket, \annot_\semiringVars,\fact)$ is obtained by setting all non-zero coefficients to $1$ in $\atann(\ontology, \database, \mathbb{N}^\infty\llbracket\semiringVars\rrbracket, \annot_\semiringVars,\fact)$. 
	\end{itemize}
\end{restatable}

An example where $\atann$ and $\provmodtot$ or $\provmodmon$ differ on $\mathbb{N}^\infty\llbracket\semiringVars\rrbracket$ is the following: 
Let $\ontology$ contain $A(x)\rightarrow B(x)$, $B(x)\rightarrow A(x)$, $\database=\{A(a)\}$ and \mbox{$\annot_\semiringVars(A(a))=x$.} 
Since there are infinitely many derivation trees for $A(a)$, \mbox{$\atann(\ontology,\database,\semiringshort,\annot_\semiringVars,A(a))=\infty x$} while for $\prov\in\{\provmodtot,\provmodmon\}$, $\prov(\ontology,\database,\mathbb{N}^\infty\llbracket\semiringVars\rrbracket,\annot_\semiringVars,A(a))=x$, as $\{A(a),B(a)\}$ with both facts annotated with $x$ (resp.\ $\{x\}$) is a (resp.\ set-)annotated model for $\ontology$ and $(\database,\mathbb{N}^\infty\llbracket\semiringVars\rrbracket,\annot_\semiringVars)$. 

Note that $\provmodtot$ and $\provmodmon$ can still lead to infinite provenance expressions: Let $\ontology=\{A(x)\wedge B(x)\rightarrow A(x)\}$, $\database=\{A(a), B(a)\}$, $\annot_\semiringVars(A(a))=x$ and $\annot_\semiringVars(B(a))=y$. For $\prov\in\{\provmodtot,\provmodmon\}$, $\prov(\ontology,\database,\mathbb{N}^\infty\llbracket\semiringVars\rrbracket,\annot_\semiringVars,A(a))=x + xy+xy^2+xy^3+\dots$.

\subsection{Execution- and Tree-Based Semantics}

We saw that when annotations are present there is more than one way to define a model-based semantics for Datalog and that it differs from the all-tree semantics. 
We now investigate definitions based on classical Datalog evaluation algorithms. 

We extend the notion of \emph{immediate consequence operator}
describing the application of rules onto facts, with the computation of annotation.  To this end, we introduce the
\emph{annotation aware {immediate} consequence operator} $\consop$. 
Applying $\consop$ on a set of annotated facts $(\inter,\semiringshort,\annot)$ results in $(\inter_{\consop},\semiringshort,\annot_{\consop})$ 
where $\inter_{\consop}$ is the result of applying the immediate consequence operator to $\ontology$ and $\inter$, and $\annot_{\consop}$ annotates facts in $\inter_{\consop}$ with the relational provenance (over $(\inter,\semiringshort,\annot)$) of the UCQ formed by the bodies of the rules that create them.  
Formally,
\begin{align*}
	&	{\inter_{\consop} }
	\df\{H(\vec{a}) \mid \inter\models \exists \vec{y}\,\phi(\vec{a},\vec{y})\,,\, \phi(\vec{x},\vec{y})\rightarrow H(\vec{x})\in\ontology  \}\\
	&\annot_{\consop}(H(\vec{a}))\df
	\sum_{\substack{h(\vec{x})=\vec{a},\,
			\inter\models h( \phi(\vec{x},\vec{y})) \\
			\phi(\vec{x},\vec{y})\rightarrow H(\vec{x})\in\ontology}} \,\,
	\prod_{\beta\in h( \phi(\vec{x},\vec{y}))} \annot(\beta)
\end{align*}

We define 
a union operator 
for annotated databases (over the same semiring): $(\inter, \semiringshort,\annot)\cup (\inter',\semiringshort,\annot') \df (\inter\cup\inter', \semiringshort, \annot'')$ where $\annot''(\alpha) \df \annot(\alpha) \plus \annot'(\alpha)$ 
where we slightly abuse notation by setting $\annot(\alpha)=0$ if $\alpha\notin\inter$, {and $\annot'(\alpha)=0$ if $\alpha\notin\inter'$}.

\subsubsection*{Naive Evaluation /\ All Trees} 

In the naive evaluation algorithm, all rules are applied in parallel until a fixpoint is reached. 
The `annotation aware' version of it is as follows: 
We set $\naive^{0}(\ontology,\database,\semiringshort,\annot) \df (\database,\semiringshort,\annot)$, and define inductively $\naive^{i+1}(\ontology,\database,\semiringshort,\annot) \df \consop(\naive^{i}(\ontology,\database,\semiringshort,\annot))\cup (\database,\semiringshort,\annot)$.
Note that the subscript $\mn{n}$ of $\naive$ is an abbreviation for `naive', and the superscript $i$ indicates how many times $\consop$ was applied.

Let $(\naive^{i},\semiringshort,\naiveannot^i) $ denote $  \naive^{i}(\ontology,\database,\semiringshort,\annot)$. We say that $\naive^{i}(\ontology,\database,\semiringshort,\annot)$ \emph{converges} if there is some $k$ such that $\naive^{\ell} = \naive^{k}$ for every $\ell \ge k$, and $\sup (\naiveannot^i (\fact)) $ exists for every $\fact \in \naive^{k}$.
\begin{restatable}{proposition}{propnaive}\label{prop:naive}
	For every $\ontology,\database,\semiringshort,\annot$,
	if ~ $\semiringshort$ is $\omega$-continuous then 
	$\naive^{i}(\ontology,\database,\semiringshort,\annot) $ converges.
\end{restatable}
\noindent
In this case, we define $\naive^\infty \df \naive^k$ and $\naiveannot^\infty \df \sup_{i\rightarrow \infty} \naiveannot^i$. 
\noindent
The \emph{na{i}ve execution provenance semantics} $\provnaive$ is defined by
\[\provnaive(\ontology, \database, \semiringshort, \annot,\fact):=
\left\{\begin{matrix}
	\naiveannot^\infty(\alpha) 	&\alpha\in\naive^\infty \\ 
	0	& \text{otherwise}
\end{matrix}\right.
\]
and is equivalent to the all-tree semantics.
\begin{restatable}{proposition}{propexecutiontreeconnection}\label{prop:executiontreeconnection} It holds that 
	$\provnaive=\atann $.
\end{restatable}

\subsubsection*{Optimized Naive Evaluation /\  Minimal Depth Trees}
We consider an optimized version of the naive algorithm 
that stops as soon as the desired fact
is derived. 
We define the `annotation aware' version of this algorithm by 
$\opti^{0}(\ontology,\database,\semiringshort,\annot)\df(\database,\semiringshort,\annot)$, and{
	\begin{multline*}
		\opti^{i+1}(\ontology,\database,\semiringshort,\annot) \df\\
		\left\{\begin{matrix}
			\consop(\opti^{i}(\ontology,\database,\semiringshort,\annot))\cup (\database,\semiringshort,\annot) &\fact\notin \opti^{i} \\ 
			\opti^{i}(\ontology,\database,\semiringshort,\annot)	& \text{otherwise}
		\end{matrix}\right.
	\end{multline*}
	where $\opti^{i}$ is such that   $ \opti^{i}(\ontology,\database,\semiringshort,\annot) \df(\opti^{i},\semiringshort,\optiannot^i)$.}

\begin{restatable}{proposition}{propopti}\label{prop:opti}
	For every $\ontology,\database,\semiringshort,\annot$,
	{and $\fact$ such that $\ontology, \database\models \fact$	,}	there exists 
	$k\geq 0$ such that $\opti^{k}(\ontology,\database,\semiringshort,\annot)  
	= \opti^{\ell}(\ontology,\database,\semiringshort,\annot) $ for every 
	$\ell\ge k$. 
\end{restatable}

With $k$ as provided by Proposition \ref{prop:opti}, we define the \emph{optimized execution provenance semantics} $\provopti$ by:
\[\provopti(\ontology, \database, \semiringshort, \annot,\fact):=
\left\{\begin{matrix}
	\optiannot^k(\alpha) 	&\alpha\in\opti^k \\ 
	0	& \text{otherwise}
\end{matrix}\right.\] 
We show that an equivalent tree-based semantics can be obtained 
by considering only minimal depth trees for the desired fact. 
This approach 
has been considered useful, for example to present a `small proof' 
for debugging \cite{DBLP:journals/toplas/ZhaoSS20}. 
Formally, let $\dep(t)$ denote the depth of tree $t$. We say that $t\in \drvtrsig$ is \emph{of minimal depth} if for every $t'\in  \drvtrsig$ it holds that $\dep(t)\le \dep(t')$. 
The \emph{minimal depth tree provenance semantics} $\mdtann$ is defined by
\[
\mdtann(\ontology,\database,\semiringshort,\annot,\fact):= 
\sum_{\substack{ t\in  \drvtrsig \\\text{ is of minimal depth}}}   \ann(t)
\]
and is equivalent to the optimized naive execution.
\begin{restatable}{proposition}{propexecutiontreeconnectionmdt}\label{prop:executiontreeconnection_mdt}
	It holds that
	$
	\provopti=\mdtann.
	$
\end{restatable}

\subsubsection*{Seminaive Evaluation /\ Hereditary Minimal Depth Trees} In the seminaive evaluation  algorithm, 
facts are derived only once.  
We introduce a new consequence operator $\diffop$ that derives only new facts and is defined as follows: 
$\diffop(\inter,\semiringshort,\annot)\df(\inter_{\diffop},\semiringshort,\annot_{\diffop})$ where 
$\consop(\inter,\semiringshort,\annot)\df(\inter_{\consop},\semiringshort,\annot_{\consop})$, 
$\inter_{\diffop}\df\inter_{\consop}\setminus\inter$, and $\annot_{\diffop}$ is the restriction of $\annot_{\consop}$ to $\inter_{\diffop}$.
We can now define the annotation aware version of the seminaive evaluation:
\sloppy{ $\semi^{0}(\ontology,\database,\semiringshort,\annot)\df(\database,\semiringshort,\annot)$ and $\semi^{i+1}(\ontology,\database,\semiringshort,\annot)\df\semi^{i}(\ontology,\database,\semiringshort,\annot) \cup \diffop(\semi^{i}(\ontology,\database,\semiringshort,\annot))$.}

\begin{restatable}{proposition}{propsemi}\label{prop:semi}
	{For every $\ontology,\database,\semiringshort,\annot$,}	there exists 
	$k\geq 0$ such that $\semi^{k}(\ontology,\database,\semiringshort,\annot)  
	= \semi^{\ell}(\ontology,\database,\semiringshort,\annot) $ for every 
	$\ell\ge k$. 
\end{restatable}
Note that, unlike in Proposition~\ref{prop:naive}, we do not require $\semiringshort$ to be $\omega$-continuous. 
With $k$ provided by Proposition \ref{prop:semi}, the \emph{seminaive execution provenance semantics} $\provsemi$ is defined~by
\[\provsemi(\ontology, \database, \semiringshort, \annot,\fact):=
\left\{\begin{matrix}
	\semiannot^k	&\alpha\in\semi^k \\ 
	0	& \text{otherwise}
\end{matrix}\right.\] 

To capture this with the tree-based approach we need to further restrict all subtrees to be of minimal depth. 
Formally, a derivation tree $t\in \drvtrsig$ is a  \emph{hereditary minimal-depth (derivation) tree} if for every node $n$ of $t$ labeled by $(\beta,\datrule,h)$, the subtree $t_\beta$ with root $n$ is a minimal-depth derivation tree for $\beta$. 
The \emph{hereditary minimal depth tree provenance semantics} $\rmdtann$ is defined by
\[
\rmdtann(\ontology,\database,\semiringshort,\annot,\fact):= 
\sum_{\substack{ t\in  \drvtrsig\\ \text{ is hereditary minimal-depth}}}   \ann(t)
\]
and is equivalent to the seminaive execution.
\begin{restatable}{proposition}{propexecutiontreeconnectionhmdt}\label{prop:executiontreeconnection_hmdt}
	It holds that
	$
	\provsemi=\rmdtann.
	$
\end{restatable}

\subsection{Non-Recursive Tree-Based Semantics}

Both execution-based semantics $\provopti$ and $\provsemi$ take into account finite subsets of derivation trees (and hence converge). Is there a more informative tree-based semantics (\ie one that takes into account a bigger subset of derivation trees) that still converges? 
We present such a semantics based on the intuition that deriving a fact from itself is redundant.

Formally, a \emph{non-recursive (derivation) tree} is a derivation tree that does not contain two nodes labeled with the same fact and such that one is the descendant of the other. 	
The \emph{non-recursive tree provenance semantics} $\mtann$ is defined by
\[
\mtann(\ontology,\database,\semiringshort,\annot,\fact):= \sum_{\substack{ t\in  \drvtrsig\\ \text{ is non-recursive}}}   \ann(t).
\]

\subsubsection*{Connections between semantics}
Next proposition follows from 
the fact that hereditary minimal-depth trees are of minimal-depth and non recursive. The sets of minimal depth trees and non-recursive trees are incomparable, so that $\mtann \not\sqsubseteq \mdtann$ and $\mdtann \not\sqsubseteq \mtann$.
\begin{restatable}{proposition}{proptreeconnections}\label{prop:treeconnections}
	The following hold:
	$$\rmdtann \sqsubseteq \mtann \sqsubseteq \atann \ \text{ and }\quad \rmdtann \sqsubseteq \mdtann \sqsubseteq \atann$$
\end{restatable}

Moreover $\mtann$ and $\atann$ coincide on specific semirings.
\begin{restatable}{proposition}{propalltreenonrecabsorptive}\label{prop:alltree-nonrec-absorptive}
	For every $\ontology,\database,\semiringshort,\annot$ and $\fact$,
	if $\semiringshort$ is a commutative absorptive $\omega$-continuous semiring, then 
	$\mtann(\ontology,\database,\semiringshort,\annot,\fact)=\atann(\ontology,\database,\semiringshort,\annot,\fact)$.
\end{restatable}
If $\semiringshort$ is not absorptive, there exists $\ontology$, $(\database,\semiringshort,\annot)$ and $\fact$ such that $\mtann(\ontology,\database,\semiringshort,\annot,\fact)\neq\atann(\ontology,\database,\semiringshort,\annot,\fact)$, even in the case where $\semiringshort$ is $\plus\,$-idempotent and $\times$-idempotent: Let $\ontology$ consist of the rule $A(x)\wedge B(x)\rightarrow A(x)$ and $\database=\{A(a), B(a)\}$. Then $\mtann(\ontology,\database,\semiringshort,\annot, A(a))=\annot(A(a))$ while $\atann(\ontology,\database,\semiringshort,\annot, A(a))=\annot(A(a))\plus \annot(A(a))\times \annot(B(a))$.

The other semantics differ even under strong restrictions. 

\begin{example}
	This example shows that $\mtann$, $\mdtann$ and $\rmdtann$ differ even if $\semiringshort$ is $\plus$ and $\times$-idempotent and absorptive. 
	\begin{align*}
		\text{Let  }\ontology=&\{B(x) \wedge C(x)\rightarrow A(x),\ 
		D(x)\rightarrow B(x), \\&
		E(x)\rightarrow C(x), \ F(x)\rightarrow E(x)
		\}\\
		\database=&\{C(a),\ D(a),\ E(a),\ F(a) \}
	\end{align*}
	The three derivation trees of $A(a)$ \wrt $\ontology$ and $\database$ are non-recursive, but only the first two are of minimal depth and only the first one is a hereditary minimal-depth tree.
	
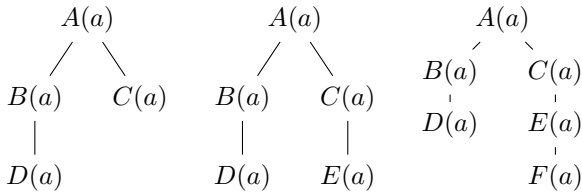
\begin{center}
		\begin{tikzpicture}
			[level distance=1.05cm,
			level 1/.style={sibling distance=1.4cm},
			level 2/.style={sibling distance=0.7cm}]
			\node {{$A(a)$}}
			child {node {{$B(a)$}}
				child {node {{$D(a)$}}}}
			child {node {{$C(a)$}}};
		\end{tikzpicture}
		\hspace{0.2cm}
		\begin{tikzpicture}
			[level distance=1.05cm,
			level 1/.style={sibling distance=1.4cm},
			level 2/.style={sibling distance=0.7cm}]
			\node {{$A(a)$}}
			child {node {{$B(a)$}}
				child {node {{$D(a)$}}}}
			child {node {{$C(a)$}}
				child { node{{$E(a)$}}}};
		\end{tikzpicture}
		\hspace{0.2cm}
		\begin{tikzpicture}
			[level distance=0.7cm,
			level 1/.style={sibling distance=1.4cm},
			level 2/.style={sibling distance=0.7cm}]
			\node {{$A(a)$}}
			child {node {{$B(a)$}}
				child {node {{$D(a)$}}}}
			child {node {{$C(a)$}}
				child { node{{$E(a)$}} 
					child { node{{$F(a)$}}}}};
		\end{tikzpicture}
	\end{center}
	
	\noindent
	Thus, if $(\database,\semiringshort,\annot)$ is such that $\annot(C(a))=c$, $\annot(D(a))=d$, $\annot(E(a))=e$, and $\annot(F(a))=f$ then
	\begin{align*}
		\mtann(\ontology,\database,\semiringshort,\annot,A(a)) =& c\times d \plus d\times e \plus d\times f\\
		\mdtann(\ontology,\database,\semiringshort,\annot,A(a)) = &c\times d \plus d\times e\\
		\rmdtann(\ontology,\database,\semiringshort,\annot,A(a)) = &c\times d
	\end{align*}	
\end{example}

\section{Basics Properties}\label{sec:properties}

In this section, we provide a framework allowing to compare the provenance semantics 
presented in the previous section. It is clear that they all fulfill the following definition.

\begin{definition}[Provenance semantics]\label{def:provenance}
	A \emph{provenance semantics} is a partial function that assigns to a Datalog program $\ontology$, annotated database $(\database, \semiringshort, \annot)$ 
	and fact $\fact$, an element $\prov(\ontology, \database, \semiringshort, \annot,\fact)$ in $K$ such that:
	\begin{enumerate}
		\item $\ontology, \database\not\models \fact$ implies $\prov(\ontology, \database, \semiringshort, \annot,\fact)=\kzero$. 
		\label{defprov:null}
		\item If $\semiringshort$ is positive, $\prov(\ontology, \database, \semiringshort, \annot,\fact)=\kzero$ implies $\ontology, \database\not\models \fact$. 
		\label{defprov:null_in_positive}
	\end{enumerate} 
	We call the \emph{semiring domain} of $\prov$ {the maximal}  set $S$ of semirings such that $\prov(\ontology, \database, \semiringshort, \annot,\fact)$ is defined for every $\semiringshort\in S$, and every 
	$\ontology$, 
	$(\database, \semiringshort, \annot)$ and 
	$\fact$.
\end{definition}

Intuitively, 
Definition \ref{def:provenance} means that the semantics reflects fact (non)-entailment.  
It is extremely permissive: We could define 
such a semantics that associates to each 
entailed fact a random semiring element different from zero, and 
does not bring any information beyond facts entailment. 
In the sequel, we state and discuss a number of properties that may be expected to be satisfied by a provenance semantics.  

Throughout this section, when not stated otherwise, $\prov$, $\ontology$, $\database$, $\semiringshort$, $\annot$ and $\fact$ denote respectively an arbitrary provenance semantics, Datalog program, database, commutative semiring $(K, \plus, \times, 0, 1)$, function from $\database$ to $K\setminus\{0\}$, and fact. 
{We phrase properties as conditions, and say that 
	$\prov$ \emph{satisfies} a property if it satisfies the condition.
}
We also denote by $\annot_\semiringVars$ an injective function $\annot_\semiringVars: \database\mapsto\semiringVars$.

\subsection{Compatibility with Classical Notions}
Property \ref{prop-algebra-consistency} is a sanity check: 
if a Datalog program amounts to a UCQ, the provenance should be the same as the one defined for relational databases \cite{greenpods2007}.
A Datalog program $\ontology$ is \emph{UCQ-defined} if its rules are of the form $\phi(\vec{x},\vec{y} ) \rightarrow H(\vec{x})$ where $H$ is a predicate 
that does not occur in the body of any rule. 
\sloppy{
	In this case,
	the equivalent UCQ $Q^{\ontology}$ of $\ontology$ is $ \bigcup_{\phi(\vec{x},\vec{y} ) \rightarrow H(\vec{x}) \in\Sigma} \exists  \vec{y}\phi(\vec{x},\vec{y} )$. 
}

\newcommand{\propAlgebraConsistency}{
	If $\ontology$ is UCQ-defined with rule head $ H(\vec{x})$ 
	and $H\notin\schema(\database)$,
	then for every tuple $\vec a $ of same arity as $\vec{x}$, 
	the relational provenance of $Q^{\ontology}(\vec{a})$
	is equal to	$\prov(\ontology, \database, \semiringshort, \annot, H(\vec{a}))$.	
}
\begin{property}[Algebra Consistency]
	\label{prop-algebra-consistency}
	\propAlgebraConsistency
\end{property}

While Property 1 considers the behavior of a provenance semantics on a restricted class of queries, we can alternatively consider its behavior on a specific semiring. 
Boolean provenance has a very natural definition, based on the database subsets that entail the query, and  
is widely used, notably for probabilistic databases \cite{DBLP:journals/sigmod/Senellart17}, but also for ontology-mediated query explanation (\eg  in Datalog$^{+/-}$ 
or description logics~\cite{DBLP:conf/ijcai/CeylanLMV19,DBLP:conf/ecai/CeylanLMV20}). 
It is formalized with the semiring~$\mi{PosBool}(\semiringVars)$. 

\begin{property}[Boolean Compatibility]
	\label{prop-boolean}
	\[\prov(\ontology, \database,\mi{PosBool}(\semiringVars), \annot_X, \fact)=\bigvee_
	{\substack{
			D'\subseteq \database \\ \ontology ,D' \models \fact 	
	}}
	\bigwedge_{\beta\in D'}\annot_X(\beta)
	\] 
\end{property}	

Property \ref{prop-boolean} expresses `insensibility' to syntax, that is, every provenance semantics that satisfies Property \ref{prop-boolean} agrees on equivalent programs (\ie those that have the same models) for the semiring $\mi{PosBool}(\semiringVars)$. 
This is related to ideas from~\cite{containmentGreen09} on the provenance of equivalent UCQs.

\subsection{Compatibility with Specialization}
Semiring provenance has been introduced to abstract from the particular semiring at hand, and factor the computations in some provenance semiring which specializes correctly to any semiring of interest. 
The next property allows one to do so (Appendix \ref{app:commutation-universal}),
and is thus highly desirable. 

\begin{property}[Commutation with Homomorphisms]
	\label{prop-commutation-homomorphism}
	If there is a semiring homomorphism $h$ from $\semiringshort_1$ to $\semiringshort_2$, then $h(\prov(\ontology, \database,\semiringshort_1,\annot,\fact))=\prov(\ontology, \database,\semiringshort_2, h\circ\annot, \fact)$.
\end{property}

We call Property~\ref{prop-commutation-homomorphism} restricted to $\omega$-continuous homomorphisms \emph{Commutation with $\omega$-Continuous Homomorphisms}.

Specializing correctly is all the more useful when $\prov$ is well-defined for a lot of semirings, in particular on all commutative or at least all commutative $\omega$-continuous semirings.

\begin{property}[Any ($\omega$-Continuous) Semiring]
	$\prov$ satisfies the \emph{Any Semiring Property} (resp. \emph{Any $\omega$-Continuous Semiring Property}) if the semiring domain of $\prov$ contains the set of all commutative (resp. commutative $\omega$-continuous) semirings. 
\end{property}

\subsection{Joint and Alternative Use of the Data}
How is the actual usage of the data reflected in the provenance semantics?
The next property formalizes that multiplication reflects joint use of the data, and addition alternative use. For the rest of this section, we set $\goal$ to be a nullary predicate not in $\schema(\ontology)\cup\schema(\database)$.

\begin{property}[Joint and Alternative Use]
	\label{prop:joint-alternative} 
	For all tuples of facts $(\fact^1_1,\cdots, \fact^1_{n_1})$, $\dots$, $(\fact^m_1,\cdots, \fact^m_{n_m})$, it holds that  
	\[\prov(\ontology', \database, \semiringshort, \annot,\goal)= \Sigma_{i=1}^m \Pi_{j=1}^{n_i}
	\prov(\ontology, \database, \semiringshort, \annot,\fact^i_j)\] 
	where $\ontology'=\ontology\cup\{\bigwedge_{j=1}^{n_i}\fact^i_j\rightarrow\goal\mid 1\le i \le m \}$. 
\end{property}

We weaken the  above by  
referring to each mode 
separately:

\begin{property}[Joint Use]
	\label{prop-joint-use}
	\sloppy{
		For all facts $\fact_1,\cdots, \fact_n$,
		\[\prov(\ontology', \database, \semiringshort, \annot,\goal)= \Pi_{j=1}^n \prov(\ontology, \database, \semiringshort, \annot,\fact_j)\] 
		where  $\ontology'=\ontology\cup\{\bigwedge_{j=1}^n\fact_j \rightarrow\goal\}$. 
	}
\end{property}

\begin{property}[Alternative Use]\label{prop:alternative}
	\label{prop-alternative-use}
	For all facts $\fact_1,\cdots, \fact_m$, \[
	\prov(\ontology', \database, \semiringshort, \annot,\goal)= \Sigma_{i=1}^m
	\prov(\ontology, \database, \semiringshort, \annot,\fact_i)
	\] 
	where 
	$\ontology'=\ontology\cup\{\fact_i\rightarrow\goal\mid 1\le i \le m \}$.
\end{property}

\subsection{Fact Roles in Entailment.}
After considering 
how facts can be 
combined or used alternatively 
to entail a result, we ponder their possible 
roles \wrt the entailment. 
Property~\ref{prop-self} asserts that the original annotation of a 
fact 
takes part in 
the provenance of its~entailment. 

\newcommand{\propself}{
	If $\fact\in \database$, then $\annot(\fact)\sqsubseteq \prov(\ontology, \database, \semiringshort, \annot,\fact)$.
}
\begin{property}[Self]
	\label{prop-self}
	\propself
\end{property}

Moreover, if a database fact cannot be alternatively derived using the 
rules, then its provenance should be exactly its original annotation. To phrase this property we use the \emph{grounding} $\ontology_\database$ of $\ontology$ \wrt $\database$, defined by $\ontology_\database=\{h(\phi(\vec{x},\vec{y})) \rightarrow h(H(\vec{x})) \mid \phi(\vec{x},\vec{y} ) \rightarrow H(\vec{x})\in\ontology, h: \vec{x}\cup\vec{y} \mapsto \domain(\database)\}$. 
It holds that $\ontology,\database\models \fact$ if and only if 
$\ontology_\database,\database\models \fact$. 
\begin{property}[Parsimony]
	\label{prop-parsimony}
	If $\fact\in \database$ does not occur in any rule head in 
	$\ontology_\database$  
	then \(\prov(\ontology, \database, \semiringshort, \annot,\fact)= \annot(\fact)\).
\end{property}

Parsimony Property together with other constraints guarantee Algebra Consistency Property. 
\begin{restatable}{proposition}{propgroundingjointalttoalgebra}\label{prop:grounding-joint-alt-to-algebra}
	If $\prov$ satisfies Properties \ref{prop:joint-alternative} (Joint and Alternative Use) and \ref{prop-parsimony} (Parsimony), and is such that for every $\ontology, \database, \semiringshort, \annot, \fact$, $\prov(\ontology, \database, \semiringshort, \annot, \fact)=\prov(\ontology_\database, \database, \semiringshort, \annot,\fact)$, 
	then it satisfies Property \ref{prop-algebra-consistency} (Algebra Consistency).
\end{restatable}

Property \ref{prop-necessary-fact} states that $\prov$ reflects the necessity of a fact for the entailment. 
We say that $\beta\in\database$ is \emph{necessary} to $\ontology, \database\models\fact$ if $\ontology, \database\setminus\{\beta\} \not \models \fact$, 
and denote by \mi{Nec} the set of such facts. 

\begin{property}[Necessary Facts]
	\label{prop-necessary-fact}
	There exists $e\in K$ such that
	\(
	\prov(\ontology, \database, \semiringshort, \annot,\fact) =\Pi_{\beta\in \mi{Nec}}\annot({\beta})\times e\).
\end{property}

A fact is \emph{usable} to $\ontology, \database\models\fact$ if it occurs in some derivation tree in $\drvtr{\fact}$. 
Usable facts are related to the notion of \emph{lineage} \cite{DBLP:journals/tods/CuiWW00} 
and can be defined without resorting to derivation trees (\cf Appendix \ref{app:def-usable}). 
Intuitively, if a fact is not usable to derive another fact, it should not have any influence on its~provenance.

\begin{property}[Non-Usable Facts]
	\label{prop-non-usable-fact}
	For every $\annot'$ that differs from $\annot$ only on facts that are not usable to $\ontology,\database\models\fact$, it holds that 
	\(\prov(\ontology, \database, \semiringshort, \annot,\fact)= \prov(\ontology, \database, \semiringshort, \annot',\fact)\).
\end{property}

\subsection{Data Modification}
The last two properties indicate how provenance is impacted when facts are inserted or deleted.
\begin{property}[Insertion]
	\label{prop:monotonicity}
	For every $(\database',\semiringshort,\annot')$ such that $\database\cap\database'=\emptyset$, 
	
	$\prov(\ontology, \database,\semiringshort, \annot, \fact) \plus \prov(\ontology, \database',\semiringshort, \annot', \fact)$\flushright$\sqsubseteq \prov(\ontology, \database\cup\database',\semiringshort, \annot\cup\annot', \fact)$.
\end{property}

Maintaining provenance upon fact deletion is very useful in practice. We formalize this using a provenance semiring, which allows us to keep track of the facts. 
A partial evaluation of a provenance expression $p(\semiringVars)$ over variables $\semiringVars$ is an expression obtained from $p(\semiringVars)$ by replacing some of the variables by a given value. 
\begin{property}[Deletion]
	\label{prop:deletion}
	For every provenance semiring $\mi{Prov}(\semiringVars)$ and $\database'\subseteq\database$, if $\annot'$ is the restriction of $\annot_X$ to $\database'$ and $\Delta=\database\setminus\database'$,  
	then 
	$\prov(\ontology, \database',\mi{Prov}(\semiringVars), \annot', \fact)$ is equal to the partial evaluation of $\prov(\ontology, \database, \mi{Prov}(\semiringVars), \annot_X, \fact)$ obtained by setting the annotations of facts in $\Delta$ to $\zero$: $\prov(\ontology, \database, \mi{Prov}(\semiringVars), \annot_X, \fact)[\{\annot_X(x) = \zero\}_{x \in \Delta}]$.
\end{property}

\section{Semantics Analysis \wrt Properties}\label{sec:analysis}

In this section, we analyze the semantics proposed in Section~\ref{sec:semantics} \wrt the properties introduced in Section \ref{sec:properties}. The properties each semantics satisfies are summarized in Table~\ref{fig:prop-def}. 
Proofs of the positive cases are given in Appendices \ref{app:prooftable} and~\ref{app:proofTableMod} 
and we discuss the negative cases, which may be more characteristic, in the sequel.

\begin{table}[h]
	\setlength{\tabcolsep}{2pt}
	\begin{tabular}{lcccccc}
		\toprule
		&
		$\atann$&
		$\mtann$&
		$\mdtann$&
		$\rmdtann$&
		$\provmodtot$&
		$\provmodmon$
		\\
		\midrule
		Algebra Consistency
		& \checkmark & \checkmark & \checkmark & \checkmark &&\\
		Boolean Compat. & \checkmark&\checkmark &&&\checkmark&\checkmark\\ 
		\midrule
		Com. with Hom.
		&&\checkmark&\checkmark&\checkmark &&\\
		Com. with $\omega$-Cont.
		&\checkmark&\checkmark&\checkmark&\checkmark &&\\
		Any Semiring &  & \checkmark & \checkmark & \checkmark &&  \\
		Any $\omega$-Cont. Sem.&\checkmark &\checkmark & \checkmark&\checkmark& &\checkmark\\
		\midrule
		Joint and Alt. Use
		&\checkmark&\checkmark& & & &\\
		Joint Use
		&\checkmark&\checkmark&&\checkmark&\checkmark &\\
		Alternative Use
		&\checkmark&\checkmark&&& &\\
		\midrule
		Self
		&\checkmark&\checkmark&\checkmark&\checkmark&\checkmark&\checkmark\\
		Parsimony 
		&\checkmark&\checkmark&\checkmark&\checkmark&\checkmark&\checkmark\\
		Necessary Facts
		&\checkmark&\checkmark&\checkmark&\checkmark&&\checkmark\\
		Non-Usable Facts
		&\checkmark&\checkmark&\checkmark&\checkmark&\checkmark&\checkmark\\
		\midrule
		Insertion 
		&\checkmark & \checkmark&  &  & & \\
		Deletion 
		&\checkmark&\checkmark&&&\checkmark&\checkmark\\
		\bottomrule
	\end{tabular}
	\caption{\label{fig:prop-def}
		Does a property hold for a provenance semantics?}
\end{table}
\subsection{Tree- and Execution-Based Semantics Cases}

We first discuss 
$\mdtann$ and $\rmdtann$, 
which have not been much investigated and stand out compared to $\atann$ and $\mtann$. 
The next example shows that they do not satisfy the Boolean Compatibility, Joint and Alternative Use, Alternative Use, Insertion and Deletion Properties.
\begin{example}
	Consider $\ontology$ and $(\database,\mi{Prov}(X),\annot_X)$ as follows.
	\begin{align*}
		\ontology=&\{A(x)\rightarrow \goal,\ B(x)\rightarrow \goal,\ C(x)\rightarrow B(x) \}\\
		\database=&\{A(a),\ C(a)\}\text{ with }\annot_X(A(a))=a, \ \annot_X(C(a))=c
	\end{align*}
	It holds that both $\mdtann(\ontology,\database,\mi{Prov}(X),\annot_X,\goal)$ and $\rmdtann(\ontology,\database,\mi{Prov}(X),\annot_X,\goal)$ are equal to~$a$. 
	For $\prov\in\{\mdtann,\rmdtann\}$ we then have the following:
	
	\noindent $(i)$ The Boolean provenance of $\goal$ is $a\vee c $, hence $\prov$ does not satisfy the Boolean Compatibility Property.
	
	\noindent$(ii)$ Since $\prov(\emptyset,\database,\mi{Prov}(X),\annot_X,A(a))=a$ and $\prov(\emptyset,\database,\mi{Prov}(X),\annot_X,C(a))=c$, 
	$\prov$ does not satisfy the Alternative Use, nor the Joint and Alternative Use Property. 
	
	\noindent$(iii)$ Let $D'=\{\goal\}$ and $\annot'(\goal)=g$. It holds that 
	$\prov(\ontology, \database\cup\database',\mi{Prov}(X),\annot_X\cup\annot_X', \goal)= g$, which is different from 
	$\prov(\ontology, \database,\mi{Prov}(X),\annot_X, \goal) \plus \prov(\ontology, \database',\mi{Prov}(X),\annot_X', \goal) \plus e$ for every $e\in \mi{Prov}(X)$. Hence $\prov$ does not satisfy the Insertion Property.
	
	\noindent$(iv)$ Let $D'=\database\setminus\{A(a)\}=\{C(a)\}$. The partial evaluation of $\prov(\ontology,\database,\mi{Prov}(X),\annot_X,\goal)$ where $a$ is set to $0$ is equal to $0$ while $\prov(\ontology,\database',\mi{Prov}(X),\annot_X,\goal)=c$. Hence $\prov$ does not satisfy the Deletion Property. 
\end{example}

We now illustrate the difference between $\rmdtann$ and $\mdtann$: $\rmdtann$ satisfies the Joint Use Property while $\mdtann$ does~not.
\begin{example}
	Let $\ontology=\{C(x)\rightarrow B(x),\ D(x)\rightarrow A(x)\}$, $\database=\{B(a),C(a),D(a)\}$ and $\annot(B(a))=b$, $\annot(C(a))=c$, $\annot(D(a))=d$, and consider  
	$\ontology'=\ontology\cup\{A(a)\wedge B(a)\rightarrow \goal\}$. 
	$\mdtann(\ontology',\database,\semiringshort,\annot,\goal)= b\times d \plus c\times d$
	while $\mdtann(\ontology,\database,\semiringshort,\annot,A(a))=  d $ and $\mdtann(\ontology,\database,\semiringshort,\annot,B(a))=  b $. 
	Hence $\mdtann$ does not satisfy the Joint Use Property. 
\end{example}

We conclude this discussion with the remark that $\atann$ satisfies the Commutation with $\omega$-Continuous Homomorphisms but not the Commutation with Homomorphisms Property.
\begin{example}
	Consider the semiring $\mathbb{N}^\infty$ with the classical operations, and define $\mathbb{N}^{\infty,\infty'}$ as its extension  by an element $\infty'$ such that for every $n\in \mathbb{N}\cup\{\infty\}$, $n\plus\infty'=\infty'$, and $n\times\infty'=\infty'$ if $n\neq 0$, and $0$ otherwise. 
	Both semirings are $\omega$-continuous and $h:\mathbb{N}^\infty\mapsto \mathbb{N}^{\infty,\infty'}$ defined by $h(n)=n$ for every $n\in\mathbb{N}$, $h(\infty)=\infty'$ is a semiring homomorphism (which is not $\omega$-continuous). 
	Assume that $\atann(\ontology,\database,\mathbb{N}^\infty,\annot,\goal)=\Sigma_{i\in\mathbb{N}} 1 =\infty$. 
	Then $h(\atann(\ontology,\database,\mathbb{N}^\infty,\annot,\goal))=\infty'$ is different from $\atann(\ontology,\database,\mathbb{N}^{\infty,\infty'},h\circ\annot,\goal)=\Sigma_{i\in\mathbb{N}} h(1) =\infty$. 
\end{example}

\subsection{Model-Based Semantics Cases}
On $\plus\,$-idempotent semirings, 
$\provmodtot$ and $\provmodmon$ coincide with $\atann$ so verify the same properties, and the semiring $\mathbb{B}\llbracket\semiringVars\rrbracket$ of formal power series with Boolean coefficients can be used to compute them in any $\plus\,$-idempotent semiring (Appendix \ref{app:booleanseriesformodlbased}). 
However, on non-idempotent semirings, they do not satisfy several properties, and in particular the Commutation with ($\omega$-Continuous) Homomorphisms Properties.

\begin{example}
	Let $\ontology=\{A(x)\rightarrow \mn{goal}, B(x)\rightarrow \mn{goal}\}$ and $\database=\{A(a), B(a)\}$ and consider the provenance semiring $\mathbb{N}^\infty\llbracket\semiringVars\rrbracket$ with $\annot_\semiringVars(A(a))=x$ and $\annot_\semiringVars(B(a))=y$. 
	
	It holds that both $\provmodtot(\ontology,\database,\mathbb{N}^\infty\llbracket\semiringVars\rrbracket,\annot_\semiringVars,\mn{goal})$ and $\provmodmon(\ontology,\database,\mathbb{N}^\infty\llbracket\semiringVars\rrbracket,\annot_\semiringVars,\mn{goal})$ 
	are equal to $x+y$. 
	
	Consider now the semiring 
	\mbox{$\mathbb{N}^\infty$, $\annot_\mathbb{N}(A(a))=2$} and $\annot_\mathbb{N}(B(a))=2$. 
	Both $\provmodtot(\ontology,\database,\mathbb{N}^\infty,\annot_\mathbb{N},\mn{goal})$ and $\provmodmon(\ontology,\database,\mathbb{N}^\infty,\annot_\mathbb{N},\mn{goal})$ are equal to $2$. 
	
	For $\prov\in\{\provmodtot,\provmodmon\}$ we then have the following:

	\noindent $(i)$  Let $h$ be a $\omega$-continuous homomorphism from $\mathbb{N}^\infty\llbracket\semiringVars\rrbracket$ to $\mathbb{N}^\infty$ such that $h(x)=2$ and $h(y)=2$. 
	Since $h(x+y)=h(x)+h(y)=4$, $\prov$ does not satisfy the Commutation with $\omega$-Continuous Homomorphisms Property. 
	
	\noindent $(ii)$  The relational provenance of $Q^\ontology()$ \wrt $\mathbb{N}^\infty$ and $\annot_\mathbb{N}$ is $4$ so $\prov$ does not satisfy the Algebra Consistency Property. 
	
	\noindent $(iii)$ $\prov(\emptyset,\database, \mathbb{N}^\infty,\annot_\mathbb{N}, A(a))+ \prov(\emptyset,\database, \mathbb{N}^\infty,\annot_\mathbb{N}, B(a))=4$ 
	so $\prov$ does not satisfy the Alternative Use nor the Joint and Alternative Use Property. 
	
	\noindent $(iv)$ Since $\prov(\ontology,\{A(a)\}, \mathbb{N}^\infty,\annot_\mathbb{N},\mn{goal}) + \prov(\ontology,\{B(a)\}, \mathbb{N}^\infty,\annot_\mathbb{N},\mn{goal})=4$ is strictly greater than $\prov(\ontology,\database, \mathbb{N}^\infty,\annot_\mathbb{N},\mn{goal})$, $\prov$ does not satisfy the Insertion Property. 
\end{example}

Moreover, $\provmodmon$ does not satisfy the Joint Use Property.

\begin{example}
	Let 
	\begin{align*}
		\ontology=\{&A(x)\rightarrow \mn{g}_1, A(x)\rightarrow \mn{g}_2, B(x)\rightarrow \mn{g}_1, B(x)\rightarrow \mn{g}_2\}\\
		\database=\{&A(a), B(a)\}\text{ with }
		\annot_\semiringVars(A(a))=x,\annot_\semiringVars(B(a))=y. 
	\end{align*}
	Both $\provmodmon(\ontology,\database,\mathbb{N}^\infty\llbracket\semiringVars\rrbracket,\annot_\semiringVars,\mn{g}_1)$ and $\provmodmon(\ontology,\database,\mathbb{N}^\infty\llbracket\semiringVars\rrbracket,\annot_\semiringVars,\mn{g}_2)$ are equal to $x+y$ 
	but 
	$\provmodmon(\ontology\cup\{\mn{g}_1 \wedge \mn{g}_2\rightarrow \mn{goal}\},\database,\mathbb{N}^\infty\llbracket\semiringVars\rrbracket,\annot_\semiringVars,\mn{goal})=x^2+y^2+xy\neq (x+y)^2$. 
\end{example}

We show that $\provmodtot$ does not satisfy the Necessary Facts Property in Appendix \ref{app:proofs-for-provmodtot-properties} 
because we needed to craft a specific semiring to get a counter-example.

\section{Complexity Considerations and Conclusion}
In this paper, we present alternative provenance semantics for Datalog based on models, execution algorithms and derivation trees, and compare them through the lens of different properties. $\mtann$ is the only one that satisfies all the studied properties but does not coincide with an execution based semantics contrary to the other tree-based semantics $\atann$, $\mdtann$, and $\rmdtann$. 
The equivalence between the tree-based $\atann$, $\mtann$ and model-based $\provmodtot$ and $\provmodmon$ definitions  on absorptive semirings may also indicates a robust provenance on this restricted setting. 

One of the main complexity sources of Datalog provenance stems from  
its infinite representation. 
\citeauthor{DeutchMRT14} \shortcite{DeutchMRT14} studied semirings for which the provenance expressions given by $\atann$ are finite, and showed that they can be represented by polynomial size circuits. 
We show (Appendix~\ref{appendix:computation}) 
that the annotations produced at each iteration of our execution algorithms can be represented by arithmetic circuits of polynomial size in the data. 
Consequently, both $\mdtann$ and $\rmdtann$ can be represented by polynomial size circuits 
regardless of the semiring. 
On the contrary, we show  
that (assuming $\mathsf{P} \ne \mathsf{NP}$) there is no polynomially computable circuit  that computes $\mtann$ on $\mathbb{N}^\infty\llbracket\semiringVars\rrbracket$, by a reduction from a result by \citeauthor{DBLP:conf/www/ArenasCP12} \shortcite{DBLP:conf/www/ArenasCP12}. 
Whether it is possible to polynomially compute circuits for $\mtann$ on provenance semirings less expressive than  $\mathbb{N}^\infty\llbracket\semiringVars\rrbracket$ but non-absorptive remains open. 

\section*{Acknowledgements}
This work is supported by the ANR project CQFD (ANR-18-CE23-0003). 

\bibliographystyle{kr}
\bibliography{provenance}

\begin{thebibliography}{}

\bibitem[\protect\citeauthoryear{Abiteboul, Hull, and
  Vianu}{1995}]{DBLP:books/aw/AbiteboulHV95}
Abiteboul, S.; Hull, R.; and Vianu, V.
\newblock 1995.
\newblock {\em Foundations of Databases}.
\newblock Addison-Wesley.

\bibitem[\protect\citeauthoryear{Achs and
  Kiss}{1995}]{DBLP:journals/actaC/AchsK95}
Achs, {\'{A}}., and Kiss, A.
\newblock 1995.
\newblock Fuzzy extension of datalog.
\newblock {\em Acta Cybern.} 12(2):153--166.

\bibitem[\protect\citeauthoryear{Arenas, Conca, and
  P{\'{e}}rez}{2012}]{DBLP:conf/www/ArenasCP12}
Arenas, M.; Conca, S.; and P{\'{e}}rez, J.
\newblock 2012.
\newblock Counting beyond a yottabyte, or how {SPARQL} 1.1 property paths will
  prevent adoption of the standard.
\newblock In Mille, A.; Gandon, F.; Misselis, J.; Rabinovich, M.; and Staab,
  S., eds., {\em Proceedings of the 21st World Wide Web Conference 2012, {WWW}
  2012, Lyon, France, April 16-20, 2012},  629--638.
\newblock {ACM}.

\bibitem[\protect\citeauthoryear{Bourgaux \bgroup et al\mbox.\egroup
  }{2020}]{provEL}
Bourgaux, C.; Ozaki, A.; Pe{\~{n}}aloza, R.; and Predoiu, L.
\newblock 2020.
\newblock Provenance for the description logic elhr.
\newblock In Bessiere, C., ed., {\em Proceedings of the Twenty-Ninth
  International Joint Conference on Artificial Intelligence, {IJCAI} 2020},
  1862--1869.
\newblock ijcai.org.

\bibitem[\protect\citeauthoryear{Calvanese \bgroup et al\mbox.\egroup
  }{2019}]{provdllite}
Calvanese, D.; Lanti, D.; Ozaki, A.; Pe{\~{n}}aloza, R.; and Xiao, G.
\newblock 2019.
\newblock Enriching ontology-based data access with provenance.
\newblock In Kraus, S., ed., {\em Proceedings of the Twenty-Eighth
  International Joint Conference on Artificial Intelligence, {IJCAI} 2019,
  Macao, China, August 10-16, 2019},  1616--1623.
\newblock ijcai.org.

\bibitem[\protect\citeauthoryear{Ceylan \bgroup et al\mbox.\egroup
  }{2019}]{DBLP:conf/ijcai/CeylanLMV19}
Ceylan, {\.I}.~{\.I}.; Lukasiewicz, T.; Malizia, E.; and Vaicenavicius, A.
\newblock 2019.
\newblock Explanations for query answers under existential rules.
\newblock In Kraus, S., ed., {\em Proceedings of the Twenty-Eighth
  International Joint Conference on Artificial Intelligence, {IJCAI} 2019,
  Macao, China, August 10-16, 2019},  1639--1646.
\newblock ijcai.org.

\bibitem[\protect\citeauthoryear{Ceylan \bgroup et al\mbox.\egroup
  }{2020}]{DBLP:conf/ecai/CeylanLMV20}
Ceylan, {\.I}.~{\.I}.; Lukasiewicz, T.; Malizia, E.; and Vaicenavicius, A.
\newblock 2020.
\newblock Explanations for ontology-mediated query answering in description
  logics.
\newblock In Giacomo, G.~D.; Catal{\'{a}}, A.; Dilkina, B.; Milano, M.; Barro,
  S.; Bugar{\'{\i}}n, A.; and Lang, J., eds., {\em {ECAI} 2020 - 24th European
  Conference on Artificial Intelligence, 29 August-8 September 2020, Santiago
  de Compostela, Spain}, volume 325 of {\em Frontiers in Artificial
  Intelligence and Applications},  672--679.
\newblock {IOS} Press.

\bibitem[\protect\citeauthoryear{Cheney, Chiticariu, and
  Tan}{2009}]{DBLP:journals/ftdb/CheneyCT09}
Cheney, J.; Chiticariu, L.; and Tan, W.~C.
\newblock 2009.
\newblock Provenance in databases: Why, how, and where.
\newblock {\em Found. Trends Databases} 1(4):379--474.

\bibitem[\protect\citeauthoryear{Cui, Widom, and
  Wiener}{2000}]{DBLP:journals/tods/CuiWW00}
Cui, Y.; Widom, J.; and Wiener, J.~L.
\newblock 2000.
\newblock Tracing the lineage of view data in a warehousing environment.
\newblock {\em {ACM} Trans. Database Syst.} 25(2):179--227.

\bibitem[\protect\citeauthoryear{Dannert \bgroup et al\mbox.\egroup
  }{2021}]{DannertGNT21}
Dannert, K.~M.; Gr{\"{a}}del, E.; Naaf, M.; and Tannen, V.
\newblock 2021.
\newblock Semiring provenance for fixed-point logic.
\newblock In Baier, C., and Goubault{-}Larrecq, J., eds., {\em 29th {EACSL}
  Annual Conference on Computer Science Logic, {CSL} 2021, January 25-28, 2021,
  Ljubljana, Slovenia (Virtual Conference)}, volume 183 of {\em LIPIcs},
  17:1--17:22.
\newblock Schloss Dagstuhl - Leibniz-Zentrum f{\"{u}}r Informatik.

\bibitem[\protect\citeauthoryear{Deutch \bgroup et al\mbox.\egroup
  }{2014}]{DeutchMRT14}
Deutch, D.; Milo, T.; Roy, S.; and Tannen, V.
\newblock 2014.
\newblock Circuits for datalog provenance.
\newblock In Schweikardt, N.; Christophides, V.; and Leroy, V., eds., {\em
  Proc. 17th International Conference on Database Theory (ICDT), Athens,
  Greece, March 24-28, 2014},  201--212.
\newblock OpenProceedings.org.

\bibitem[\protect\citeauthoryear{Deutch, Gilad, and
  Moskovitch}{2018}]{DBLP:journals/vldb/DeutchGM18}
Deutch, D.; Gilad, A.; and Moskovitch, Y.
\newblock 2018.
\newblock Efficient provenance tracking for datalog using top-k queries.
\newblock {\em {VLDB} J.} 27(2):245--269.

\bibitem[\protect\citeauthoryear{Esparza and
  Luttenberger}{2011}]{10.1007/978-3-642-22944-2_2}
Esparza, J., and Luttenberger, M.
\newblock 2011.
\newblock Solving fixed-point equations by derivation tree analysis.
\newblock In Corradini, A.; Klin, B.; and C{\^i}rstea, C., eds., {\em Algebra
  and Coalgebra in Computer Science},  19--35.
\newblock Berlin, Heidelberg: Springer Berlin Heidelberg.

\bibitem[\protect\citeauthoryear{Green and
  Tannen}{2017}]{DBLP:conf/pods/GreenT17}
Green, T.~J., and Tannen, V.
\newblock 2017.
\newblock The semiring framework for database provenance.
\newblock In Sallinger, E.; den Bussche, J.~V.; and Geerts, F., eds., {\em
  Proceedings of the 36th {ACM} {SIGMOD-SIGACT-SIGAI} Symposium on Principles
  of Database Systems, {PODS} 2017, Chicago, IL, USA, May 14-19, 2017},
  93--99.
\newblock {ACM}.

\bibitem[\protect\citeauthoryear{Green, Karvounarakis, and
  Tannen}{2007}]{greenpods2007}
Green, T.~J.; Karvounarakis, G.; and Tannen, V.
\newblock 2007.
\newblock Provenance semirings.
\newblock In Libkin, L., ed., {\em Proceedings of the Twenty-Sixth {ACM}
  {SIGACT-SIGMOD-SIGART} Symposium on Principles of Database Systems, June
  11-13, 2007, Beijing, China},  31--40.
\newblock {ACM}.

\bibitem[\protect\citeauthoryear{Green}{2009}]{containmentGreen09}
Green, T.~J.
\newblock 2009.
\newblock Containment of conjunctive queries on annotated relations.
\newblock In Fagin, R., ed., {\em Database Theory - {ICDT} 2009, 12th
  International Conference, St. Petersburg, Russia, March 23-25, 2009,
  Proceedings}, volume 361 of {\em {ACM} International Conference Proceeding
  Series},  296--309.
\newblock {ACM}.

\bibitem[\protect\citeauthoryear{Hernich and
  Kolaitis}{2017}]{DBLP:conf/lics/HernichK17}
Hernich, A., and Kolaitis, P.~G.
\newblock 2017.
\newblock Foundations of information integration under bag semantics.
\newblock In {\em 32nd Annual {ACM/IEEE} Symposium on Logic in Computer
  Science, {LICS} 2017, Reykjavik, Iceland, June 20-23, 2017},  1--12.
\newblock {IEEE} Computer Society.

\bibitem[\protect\citeauthoryear{Mumick, Pirahesh, and
  Ramakrishnan}{1990}]{DBLP:conf/vldb/MumickPR90}
Mumick, I.~S.; Pirahesh, H.; and Ramakrishnan, R.
\newblock 1990.
\newblock The magic of duplicates and aggregates.
\newblock In McLeod, D.; Sacks{-}Davis, R.; and Schek, H., eds., {\em 16th
  International Conference on Very Large Data Bases, August 13-16, 1990,
  Brisbane, Queensland, Australia, Proceedings},  264--277.
\newblock Morgan Kaufmann.

\bibitem[\protect\citeauthoryear{Nikolaou \bgroup et al\mbox.\egroup
  }{2017}]{DBLP:conf/ijcai/NikolaouKKKGH17}
Nikolaou, C.; Kostylev, E.~V.; Konstantinidis, G.; Kaminski, M.; Grau, B.~C.;
  and Horrocks, I.
\newblock 2017.
\newblock The bag semantics of ontology-based data access.
\newblock In Sierra, C., ed., {\em Proceedings of the Twenty-Sixth
  International Joint Conference on Artificial Intelligence, {IJCAI} 2017,
  Melbourne, Australia, August 19-25, 2017},  1224--1230.
\newblock ijcai.org.

\bibitem[\protect\citeauthoryear{Nikolaou \bgroup et al\mbox.\egroup
  }{2019}]{DBLP:journals/ai/NikolaouKKKGH19}
Nikolaou, C.; Kostylev, E.~V.; Konstantinidis, G.; Kaminski, M.; Grau, B.~C.;
  and Horrocks, I.
\newblock 2019.
\newblock Foundations of ontology-based data access under bag semantics.
\newblock {\em Artif. Intell.} 274:91--132.

\bibitem[\protect\citeauthoryear{Senellart}{2017}]{DBLP:journals/sigmod/Senellart17}
Senellart, P.
\newblock 2017.
\newblock Provenance and probabilities in relational databases.
\newblock {\em {SIGMOD} Rec.} 46(4):5--15.

\bibitem[\protect\citeauthoryear{{Soufflé}}{2020}]{souffle}
{Soufflé}.
\newblock 2020.
\newblock \url{https://souffle-lang.github.io/index.html}.

\bibitem[\protect\citeauthoryear{Stüber and Vogler}{2008}]{STUBER2008221}
Stüber, T., and Vogler, H.
\newblock 2008.
\newblock Weighted monadic datalog.
\newblock {\em Theoretical Computer Science} 403(2):221--238.

\bibitem[\protect\citeauthoryear{Zhao, Subotic, and
  Scholz}{2020}]{DBLP:journals/toplas/ZhaoSS20}
Zhao, D.; Subotic, P.; and Scholz, B.
\newblock 2020.
\newblock Debugging large-scale datalog: {A} scalable provenance evaluation
  strategy.
\newblock {\em {ACM} Trans. Program. Lang. Syst.} 42(2):7:1--7:35.

\end{thebibliography}

\newpage
\onecolumn
\appendix
\appendixpage

\startcontents[sections]
\printcontents[sections]{l}{1}{\setcounter{tocdepth}{2}}


\section{Discussion and Proofs for Section \ref{sec:semantics}}

\subsection{Relationships between Provenance Semantics and Different Bag Semantics}

\subsubsection{Connection between Incognizant and Cognizant Bag Semantics and Provenance Semantics}
To compare the bag semantics defined by \citeauthor{DBLP:conf/lics/HernichK17} \shortcite{DBLP:conf/lics/HernichK17}  in the context of information integration and a provenance semantics $\prov$ with counting semiring $\mathbb{N}$, we consider the case where $\ontology$ is a set of Datalog rules (in normalized form with a single atom in the head) which are also s-t tgds, \ie such that the predicates used in the rule heads and bodies are disjoint. 
For a provenance semantics $\prov$, we want to know whether for every such $\ontology$, $\database$ over the source schema (\ie predicates that occur in the Datalog rule bodies), $\annot_\mathbb{N}:\database\mapsto \mathbb{N}\setminus\{0\}$ and BCQ $q$ over the target schema (\ie predicates that occur in Datalog rule heads), $\prov(\ontology\cup\{q\rightarrow\mn{goal}\},\database,\mathbb{N},\annot_\mathbb{N},\mn{goal})$ is equal to the multiplicity of the empty tuple in the the incognizant or cognizant certain answers of $q$ \wrt $\ontology$ and $\database$. 
Note that in this context, all derivation trees of $\mn{goal}$ \wrt $\database$ and $\ontology$ are non-recursive and of depth 2. 
Hence, all derivation tree-based (and execution-based) semantics coincide. We show below that they are in line with the cognizant bag semantics. We have shown in Section \ref{subsec:model-based-sem} that this is not the case for the incognizant semantics, \ie only the cognizant bag semantics for information integration agrees with the traditional bag semantics for Datalog. 

\begin{proposition}
	For $\prov\in\{\atann,\mtann, \rmdtann, \mdtann, \provnaive, \provsemi, \provopti\}$, $\ontology$, $\database$, $\annot_\mathbb{N}$, and $q$ as required, $\prov(\ontology\cup\{q\rightarrow\mn{goal}\},\database,\mathbb{N},\annot_\mathbb{N},\mn{goal})$ is equal to the multiplicity of the empty tuple in the cognizant certain answers of $q$ \wrt $\ontology$ and $\database$.
\end{proposition}
\begin{proof}
	We first show that for every fact $\fact=p(\vec{a})$ over the target schema such that $\ontology,\database\models \fact$, $\prov(\ontology,\database,\mathbb{N},\annot_\mathbb{N},\fact)$ is equal to the multiplicity $n_\fact$ of the empty tuple in the cognizant certain answers of the Boolean query $\fact$ \wrt $\ontology$ and $\database$. 
	\begin{itemize}
		\item The set of derivation trees for $\fact$ correspond precisely to the set of pairs $(\datrule,h)$ such that $\datrule\in\ontology$ is of the form $\phi(\vec{x},\vec{y} ) \rightarrow  p(\vec{x})$ and $h$ is an homomorphism from $\phi(\vec{x},\vec{y} )$ to $\database$ such that $h(\vec{x})=\vec{a}$. Hence $\prov(\ontology,\database,\mathbb{N},\annot_\mathbb{N},\fact)$ is equal to the sum over all s-t tgd $\datrule:=\phi(\vec{x},\vec{y} ) \rightarrow  p(\vec{x})$ in $\ontology$ of the multiplicity of $\vec{a}$ in the answer to $\phi(\vec{x},\vec{y} )$ on $\database$. 
		
		\item Let $J$ be a cognizant solution for $\database$ \wrt $\ontology$. 
		The multiplicity of $\fact$ in $J$ is at least the sum over the rules $\datrule:=\phi(\vec{x},\vec{y} ) \rightarrow  p(\vec{x})$ in $\ontology$ of the multiplicities of $\ans$ in the answers of $\phi(\vec{x},\vec{y} )$ on $\database$. Hence the multiplicity of $\fact$ in $J$ is greater or equal to $\prov(\ontology,\database,\mathbb{N},\annot_\mathbb{N},\fact)$. 
		Since this is true for every cognizant solution $J$, it follows that $\prov(\ontology,\database,\mathbb{N},\annot_\mathbb{N},\fact)\leq n_\fact$.
		
		\item If $J$ is a cognizant solution for $\database$ \wrt $\ontology$, then $J'$ obtained from $J$ by setting the multiplicity of $\fact$ to $\prov(\ontology,\database,\mathbb{N},\annot_\mathbb{N},\fact)$ is also a cognizant solution. Indeed, for each 
		$\datrule:=\phi(\vec{x},\vec{y} ) \rightarrow  p(\vec{x})$ in $\ontology$, $J_\datrule=\{(p(\vec{b}), n_{\datrule,\vec{b}} )\mid h:\phi(\vec{x},\vec{y} )\mapsto D, h(\vec{x})=\vec{b}, n_{\datrule,\vec{b}} \text{ multiplicity of }\vec{b}\text{ in answers to }\phi(\vec{x},\vec{y} )\text{ on }\database\}$ satisfies $\datrule$, and $\sum_{\datrule:=\phi(\vec{x},\vec{y} ) \rightarrow  p(\vec{x})\in\ontology}n_{\datrule,\vec{a}}=\prov(\ontology,\database,\mathbb{N},\annot_\mathbb{N},\fact)$. 
		It follows that $\prov(\ontology,\database,\mathbb{N},\annot_\mathbb{N},\fact)\geq n_\fact$.
	\end{itemize}
	Hence, the multiplicity $n_\fact$ of the empty tuple in the cognizant certain answers of the Boolean query $\fact$ \wrt $\ontology$ and $\database$ is equal to $\prov(\ontology,\database,\mathbb{N},\annot_\mathbb{N},\fact)$. 
	We now show that $\prov(\ontology\cup\{q\rightarrow\mn{goal}\},\database,\mathbb{N},\annot_\mathbb{N},\mn{goal})$ is equal to the multiplicity $n_q$ of the empty tuple in the cognizant certain answers of $q$ \wrt $\ontology$ and $\database$.
	Let $q=\exists \vec{x}\bigwedge_{i=1}^n p_i(\vec{x})$. 
	\begin{itemize}
		\item Let $J=\{\fact\mid\ontology,\database\models\fact\}$. 
		For each homomorphism $h:\bigwedge_{i=1}^n p_i(\vec{x})\mapsto J$, the derivation trees for $\mn{goal}$ with root labeled by $(\mn{goal}, q\rightarrow\mn{goal}, h)$ have children $h(p_1(\vec{x}))=\fact_1,\dots,h(p_n(\vec{x}))=\fact_n \in J$ and correspond to the choice of $t_1,\dots,t_n$ where each $t_i$ is a derivation tree for $\fact_i$. 
		Hence, 
		\begin{align*}
			\prov(\ontology\cup\{q\rightarrow\mn{goal}\},\database,\mathbb{N},\annot_\mathbb{N},\mn{goal})=&\sum_{h:\bigwedge_{i=1}^n p_i(\vec{x})\mapsto J} \sum_{(t_1,\dots,t_n)\in\drvtr{\fact_1}\times\dots\times\drvtr{\fact_n}} \ann(t_i)\\
			=&\sum_{h:\bigwedge_{i=1}^n p_i(\vec{x})\mapsto J} \prod_{i=1}^n \prov(\ontology,\database,\mathbb{N},\annot_\mathbb{N},\fact_i).
		\end{align*}
		\item Let $J$ be a cognizant solution for $\database$ \wrt $\ontology$. The multiplicity of the empty tuple in the answers of $q$ over $J$ is the sum over the homomorphisms $h:q\mapsto J$ of the product of the multiplicities of $\fact_i=h(p_i(\vec{x}))$ in $J$. Moreover, every such $\fact_i$ is such that $\ontology,\database\models\fact_i$ so its multiplicity in $J$ is at least $\prov(\ontology,\database,\mathbb{N},\annot_\mathbb{N},\fact_i)$. 
		Since this holds for any cognizant solution $J$, it follows that $n_q\geq \prov(\ontology\cup\{q\rightarrow\mn{goal}\},\database,\mathbb{N},\annot_\mathbb{N},\mn{goal})$.
		
		\item A cognizant solution $J$ for $\database$ \wrt $\ontology$ can be obtained by setting the multiplicity of each fact on the target schema $\alpha$ to $\prov(\ontology,\database,\mathbb{N},\annot_\mathbb{N},\fact)$. 
		It follows that $n_q\leq \sum_{h:\bigwedge_{i=1}^n p_i(\vec{x})\mapsto J} \prod_{i=1}^n \prov(\ontology,\database,\mathbb{N},\annot_\mathbb{N},\fact_i)=\prov(\ontology\cup\{q\rightarrow\mn{goal}\},\database,\mathbb{N},\annot_\mathbb{N},\mn{goal})$.
	\end{itemize}
	Hence the multiplicity $n_q$ of the empty tuple in the cognizant certain answers of $q$ \wrt $\ontology$ and $\database$ is equal to $\prov(\ontology\cup\{q\rightarrow\mn{goal}\},\database,\mathbb{N},\annot_\mathbb{N},\mn{goal})$.
\end{proof}

We show below that the incognizant semantics coincides with $\provmodtot$ with the counting semiring. The examples we gave in Section \ref{subsec:model-based-sem} to show that $\provmodtot$ and $\provmodmon$ differ show that this is not the case of $\provmodmon$.
\begin{proposition}\label{prop:provmodtot-capture-incognizant}
	For every $\ontology$, $\database$, $\annot_\mathbb{N}$, and $q$ as required, $\provmodtot(\ontology\cup\{q\rightarrow\mn{goal} \},\database,\mathbb{N},\annot_\mathbb{N},\mn{goal})$ is equal to the multiplicity  of the empty tuple in the incognizant certain answers of $q$ \wrt $\ontology$ and~$\database$.
\end{proposition}
\begin{proof}
	We first show that for every fact $\fact=p(\vec{a})$ over the target schema such that $\ontology,\database\models \fact$, $\provmodtot(\ontology,\database,\mathbb{N},\annot_\mathbb{N},\fact)$ is equal to the multiplicity $n_\fact$ of the empty tuple in the incognizant certain answers of the Boolean query $\fact$ \wrt $\ontology$ and $\database$. 
	\begin{itemize}
		\item We show that $\provmodtot(\ontology,\database,\mathbb{N},\annot_\mathbb{N},\fact)$ is equal to the maximum over all s-t tgd $\phi(\vec{x},\vec{y} ) \rightarrow  p(\vec{x})$ in $\ontology$ of the multiplicity of $\vec{a}$ in the answer to $\phi(\vec{x},\vec{y} )$ on $\database$. 
		Let $\phi(\vec{x},\vec{y} ) \rightarrow  p(\vec{x})$ in $\ontology$. 
		For every model $(\inter,\modannot^I)$ of $(\ontology,\database,\mathbb{N},\annot_\mathbb{N})$, 
		$\Sigma_{h':\phi(\vec{x},\vec{y})\mapsto I, h'(\vec{x})=\vec{a}} \Pi_{\beta\in h'(\phi(\vec{x},\vec{y}))} \modannot^I(\beta) \leq \modannot^I(\fact)$ and for every $\beta\in h'(\phi(\vec{x},\vec{y}))$, $\beta\in\database$ because $\phi(\vec{x},\vec{y} ) \rightarrow  p(\vec{x})$ is a s-t tgd, so $\modannot^I(\beta)\geq \annot_\mathbb{N}(\beta)$. Hence 
		$\modannot^I(\fact)$ is greater or equal to the multiplicity of $\vec{a}$ in the answer to $\phi(\vec{x},\vec{y} )$ on $\database$. 
		Moreover, the interpretation that annotates each $\beta\in\database$ by $\annot_\mathbb{N}(\beta)$ and each $\gamma$ produced by applying some s-t tgd by such maximal multiplicity is a model of $(\ontology,\database,\mathbb{N},\annot_\mathbb{N})$. 
		
		\item We now show that $n_\fact$ is precisely the maximum over all s-t tgd $\phi(\vec{x},\vec{y} ) \rightarrow  p(\vec{x})$ in $\ontology$ of the multiplicity of $\vec{a}$ in the answer to $\phi(\vec{x},\vec{y} )$ on $\database$. 
		Every incognizant solution $J$ for $\database$ \wrt $\ontology$ is such that the multiplicity of $\fact$ is at least the multiplicity of $\ans$ in the answer of $\phi(\vec{x},\vec{y} )$ on $\database$ for every $\phi(\vec{x},\vec{y} ) \rightarrow  p(\vec{x})$ in $\ontology$. 
		Moreover, if $J$ is a incognizant solution for $\database$ \wrt $\ontology$, then $J'$ obtained from $J$ by setting the multiplicity of $\fact$ to this maximal multiplicity is also a incognizant solution, as it still satisfies all s-t tgds. 
	\end{itemize}
	Hence, $n_\fact=\provmodtot(\ontology,\database,\mathbb{N},\annot_\mathbb{N},\fact)$.\smallskip

	$\provmodtot(\ontology\cup\{q\rightarrow\mn{goal}\},\database,\mathbb{N},\annot_\mathbb{N},\mn{goal})$ is the minimal $\modannot^I(\mn{goal})$ over the models $(\inter,\modannot^I)$ of $\ontology\cup\{q\rightarrow\mn{goal}\}$ and $(\database,\mathbb{N},\annot_\mathbb{N})$. Since such models are models of $\ontology$ and $(\database,\mathbb{N},\annot_\mathbb{N})$, $\provmodtot(\ontology\cup\{q\rightarrow\mn{goal}\},\database,\mathbb{N},\annot_\mathbb{N},\mn{goal})$ is equal to the sum of $\prod_{\fact_i\in h(q)}\provmodtot(\ontology,\database,\mathbb{N},\annot_\mathbb{N},\fact_i)$ where $h$ ranges over the homomorphisms from $q$ to the set $J$ of facts entailed by $\ontology$ and $\database$.
	
	The multiplicity of the empty tuple in the incognizant certain answers of $q$ \wrt $\ontology$ and $\database$ is the minimum over the incognizant solutions of the multiplicity of the empty tuple in the answers of $q$. 
	Given a incognizant solution $J$, the multiplicity of the empty tuple in the answers of $q$ over $J$ is the sum over the homomorphisms $h:q\mapsto J$ of the product of the multiplicities of $\fact_i=h(p_i(\vec{x}))$. Moreover, such $\fact_i$ are such that $\ontology,\database\models\fact_i$ so their minimal multiplicities in some incognizant solution are $n_{\fact_i}=\provmodtot(\ontology,\database,\mathbb{N},\annot_\mathbb{N},\fact_i)$. 
	
	It follows that the multiplicity of the empty tuple in the incognizant certain answers of the BCQ $\fact$ \wrt $\ontology$ and $\database$ is equal to $\provmodtot(\ontology\cup\{q\rightarrow\mn{goal}\},\database,\mathbb{N},\annot_\mathbb{N},\mn{goal})$.
\end{proof}

\subsubsection{Connection between Description Logics Bag Semantics and Provenance Semantics}

\citeauthor{DBLP:conf/ijcai/NikolaouKKKGH17} \shortcite{DBLP:conf/ijcai/NikolaouKKKGH17} defined a bag semantics for the description logic DL-Lite then extended it to the full ontology-based data access setting \shortcite{DBLP:journals/ai/NikolaouKKKGH19}. 
We focus on the DL-Lite case with unique name assumption considered in \cite{DBLP:conf/ijcai/NikolaouKKKGH17} for simplicity, but mapping rules compatible with our setting could be added. 

A bag ABox corresponds to a database annotated with integers $(\database,\mathbb{N}^\infty,\annot_\mathbb{N})$. 
A bag interpretation $\Imc=\langle \Delta^\Imc, \cdot^\Imc\rangle$ can also be seen as a (possibly infinite) set of facts annotated with elements from $\mathbb{N}^\infty$. The interpretation function extends to concepts and roles as follows: $(P^-)^\Imc(u,v)= P^\Imc(v,u)$ and $(\exists R)^\Imc(u)=\sum_{v\in \Delta^\Imc} R^\Imc(u,v)$. 
$\Imc$ is a model of a bag ABox if the multiplicity of every fact in $\Imc$ is at least its multiplicity in the ABox; it is a model of a TBox if it satisfies all its concept and role inclusions where $C\sqsubseteq D$ is satisfied iff $C^\Imc(\vec{x})\leq D^\Imc(\vec{x})$ for every $\vec{x}$. 

The bag answers of a CQ $q=\exists \vec{y}\phi(\vec{x},\vec{y})$ over a bag interpretation $\Imc$ are defined by $q^\Imc(\vec{a})=\sum_{\nu\in V}\prod_{S(\vec{t})\in\phi(\vec{x},\vec{y})}  S^\Imc(\nu(\vec{t}))$ where $V$ is the set of all valuations $\nu:\vec{x}\cup\vec{y}\mapsto \Delta^\Imc$ such that $\nu(\vec{x})=\vec{a}^\Imc$ and $\nu(a)=a^\Imc$ for every constant $a$. 
Finally the bag certain answers to $q$ is the bag-intersection of $q^\Imc$ over all models $\Imc$ of the ABox and TBox, \ie the multiplicity of a certain answer is the minimum of its multiplicities over the models.

To compare the bag semantics defined by \citeauthor{DBLP:conf/ijcai/NikolaouKKKGH17} \shortcite{DBLP:conf/ijcai/NikolaouKKKGH17} for DL-Lite and a provenance $\prov$ with the extended counting semiring $\mathbb{N}^\infty$, we consider the case where $\ontology$ is a set of Datalog rules which are formulated in DL-Lite, \ie use only unary and binary predicates and contains a single atom in body and head.

We show below that the DL-Lite bag semantics coincides with $\provmodtot$ with the (extended) counting semiring. The examples we gave in Section \ref{subsec:model-based-sem} to show that $\provmodtot$ and $\provmodmon$ differ can be easily adapted to show that $\provmodmon$ does not coincide with the DL-Lite bag semantics (just replace $\mn{goal}$ by $H(x)$ in $\ontology$ and consider BCQ $q=\exists x H(x)$), and the example we gave in Section \ref{subsec:model-based-sem} to show that the cognizant and incognizant semantics differ shows that the execution/derivation-tree based semantics do not coincide with the DL-Lite bag semantics either. 
\begin{proposition}\label{prop:provmodtot-capture-DL}
	For every set $\ontology$ of Datalog rules which are formulated in DL-Lite, $\database$, $\annot_\mathbb{N}$, and BCQ $q$, $\provmodtot(\ontology\cup\{q\rightarrow \mn{goal}\},\database,\mathbb{N}^\infty,\annot_\mathbb{N},\mn{goal})$ is equal to the multiplicity of the empty tuple in the bag certain answers to $q$ over $(\ontology,\database,\mathbb{N}^\infty,\annot_\mathbb{N})$.
\end{proposition}
\begin{proof}
	Let $n$ be the multiplicity of the empty tuple in the bag certain answers to $q$ over $(\ontology,\database,\mathbb{N}^\infty,\annot_\mathbb{N})$. 
	
	Let $\Imc=\langle \Delta^\Imc, \cdot^\Imc\rangle$ be a model of $\ontology$ and $(\database,\mathbb{N}^\infty,\annot_\mathbb{N})$ seen as a DL-Lite TBox and bag ABox. 
	Assume w.l.o.g. that $\domain(\database)\subseteq \Delta$ and that for every $a\in\domain(\database)$, $a^\Imc=a$. 
	Let $(\inter,\modannot^I)$ be the annotated interpretation defined by $\inter=\{ S(\vec{a}) \mid S^\Imc(\vec{a})>0\}\cup\{\mn{goal}\}$ and $\modannot^I(S(\vec{a}))=S^\Imc(\vec{a})$ for every $S(\vec{a})\in \inter$, $\modannot^I(\mn{goal})=q^\Imc() $. 
	We show that $(\inter,\modannot^I)$ is an annotated model of $\ontology$ and $(\database,\mathbb{N}^\infty,\annot_\mathbb{N})$. 
	\begin{itemize}
		\item Since $\Imc$ is a model of $(\database,\mathbb{N}^\infty,\annot_\mathbb{N})$, then for every $S(\vec{a})\in\database$, $S^\Imc(\vec{a}^\Imc)\geq \annot_\mathbb{N}(S(\vec{a}))$. Hence $\inter$ fulfills point (1) of the definition of annotated models. 
		\item Let $\datrule\df\phi(\vec{x},\vec{y})\rightarrow H(\vec{x})$ be in $\ontology$ and $h:\phi(\vec{x},\vec{y})\mapsto \inter$. We are in one of the following cases.
		\begin{itemize}
			\item $\datrule=A(x)\rightarrow B(x)$, $h(x)=a$, $A(a)\in\inter$: since $\Imc$ satisfies $A\sqsubseteq B$, $A^\Imc(a)\leq B^\Imc(a)$, so $B(a)\in\inter$ and $\modannot^I(A(a))\leq \modannot^I(B(a))$.
			\item $\datrule=R(x,y)\rightarrow S(x,y)$, $h(x)=a$, $h(y)=b$, $R(a,b)\in\inter$: since $\Imc$ satisfies $R\sqsubseteq S$, $R^\Imc(a,b)\leq S^\Imc(a,b)$, so $S(a,b)\in\inter$ and $\modannot^I(R(a,b))\leq \modannot^I(S(a,b))$.
			\item $\datrule=R(x,y)\rightarrow S(y,x)$, $h(x)=a$, $h(y)=b$, $R(a,b)\in\inter$: since $\Imc$ satisfies $R\sqsubseteq S^-$, $R^\Imc(a,b)\leq (S^-)^\Imc(a,b)=S^\Imc(b,a)$, so $S(b,a)\in\inter$ and $\modannot^I(R(a,b))\leq \modannot^I(S(b,a))$.
			\item $\datrule=R(x,y) \rightarrow B(x)$, $h(x)=a$, $h(y)=b$, $R(a,b)\in\inter$: since $\Imc$ satisfies $\exists R\sqsubseteq B$, $(\exists R)^\Imc(a)\leq B^\Imc(a)$ with $(\exists R)^\Imc(a)=\sum_{v\in \Delta^\Imc} R^\Imc(a,v)$, so $B(a)\in\inter$ and $\sum_{R(a,v)\in\inter} \modannot^I(R(a,v))\leq \modannot^I(B(a))$.
			\item $\datrule=R(y,x) \rightarrow B(x)$, $h(x)=a$, $h(y)=b$, $R(b,a)\in\inter$: since $\Imc$ satisfies $\exists R^-\sqsubseteq B$, $(\exists R^-)^\Imc(a)\leq B^\Imc(a)$ with $(\exists R^-)^\Imc(a)=\sum_{v\in \Delta^\Imc} R^\Imc(v,a)$, so $B(a)\in\inter$ and $\sum_{R(v,a)\in\inter} \modannot^I(R(v,a))\leq \modannot^I(B(a))$. 
		\end{itemize} 
		In all cases, we have shown that $h(H(\vec{x}))\in\inter$ and $\sum_{h':\phi(\vec{x},\vec{y})\mapsto I, h'(\vec{x})=h(\vec{x})} \prod_{\beta\in h'(\phi(\vec{x},\vec{y}))} \modannot^I(\beta) \sqsubseteq \modannot^I(h(H(\vec{x})))$. 
		
		\item Finally, consider $q\rightarrow \mn{goal}$ and assume that there is a homomorphism $h$ from $q$ to $\inter$. Since $\modannot^I(\mn{goal})=q^\Imc() = \sum_{\nu\in V}\prod_{S(\vec{t})\in q}  S^\Imc(\nu(\vec{t}))$ where $V$ is the set of all valuations $\nu:\vec{x}\mapsto \Delta^\Imc$ such that $\nu(a)=a^\Imc$ for every constant $a$, $\modannot^I(\mn{goal}) \geq \sum_{h':q\mapsto I}\prod_{S(\vec{t})\in q}  S^\Imc(h'(\vec{t}))= \sum_{h':q\mapsto I}\prod_{\beta\in h'(q)} \modannot^I(\beta)$. 
		
		Hence $\inter$ fulfills point (2) of the definition of annotated models. 
	\end{itemize} 
	By definition of $\provmodtot(\ontology\cup\{q\rightarrow \mn{goal}\},\database,\mathbb{N}^\infty,\annot_\mathbb{N},\mn{goal})$, it follows that $\provmodtot(\ontology\cup\{q\rightarrow \mn{goal}\},\database,\mathbb{N}^\infty,\annot_\mathbb{N},\mn{goal})\leq \modannot^I(\mn{goal})$. Hence $\provmodtot(\ontology\cup\{q\rightarrow \mn{goal}\},\database,\mathbb{N}^\infty,\annot_\mathbb{N},\mn{goal})\leq q^\Imc()$. Since this is true for any model $\Imc$ and $n$ is the minimal $q^\Imc()$ over all models, it follows that $\provmodtot(\ontology\cup\{q\rightarrow \mn{goal}\},\database,\mathbb{N}^\infty,\annot_\mathbb{N},\mn{goal})\leq n$. \smallskip
	
	In the other direction, let $(\inter,\modannot^I)$ be an annotated model of $\ontology\cup\{q\rightarrow \mn{goal}\}$ and $(\database,\mathbb{N}^\infty,\annot_\mathbb{N})$. 
	Let $\Imc=\langle \Delta^\Imc, \cdot^\Imc\rangle$ be defined by $\Delta^\Imc=\domain(\database)$ and $S^\Imc(\vec{x})=\modannot^I(S(\vec{x}))$ if $S(\vec{x})\in\inter$, $S^\Imc(\vec{x})=0$ otherwise. We show that $\Imc$ is a model of $\ontology$ and $(\database,\mathbb{N}^\infty,\annot_\mathbb{N})$ seen as a DL-Lite TBox and bag ABox. 
	\begin{itemize}
		\item Since for every $\fact\in\database$, $\annot_\mathbb{N}(\fact)\leq\modannot^I(\fact)$, the multiplicity of $\alpha$ in $\Imc$ is at least its multiplicity $\annot_\mathbb{N}(\fact)$ in the ABox, so $\Imc$ is a model of $(\database,\mathbb{N}^\infty,\annot_\mathbb{N})$ seen as a bag ABox.
		\item We show that $\Imc$ satisfies all concept and role inclusions corresponding to rules in $\ontology$.
		\begin{itemize}
			\item For simple concept inclusion $A\sqsubseteq B$, \ie $A(x)\rightarrow B(x)$, since for every $a\in\Delta^\Imc$, $A^\Imc(a)=\modannot^I(A(a))$, $B^\Imc(a)=\modannot^I(B(a))$ and $(\inter,\modannot^I)$ is a model of $\ontology$, $A^\Imc(a)\leq B^\Imc(a)$ so $\Imc$ satisfies $A\sqsubseteq B$.
			\item The role inclusion cases $R\sqsubseteq S$ and $R\sqsubseteq S^-$ are similar to $A\sqsubseteq B$.
			\item For the case $\exists R\sqsubseteq B$, \ie $R(x,y)\rightarrow B(x)$, for every $a\in\Delta^\Imc$, $(\exists R)^\Imc(a)=\sum_{v\in \Delta^\Imc} R^\Imc(a,v)=\sum_{R(a,v)\in\inter} \modannot^I(R(a,v))$ and since $(\inter,\modannot^I)$ is a model of $\ontology$, $\sum_{R(a,v)\in\inter} \modannot^I(R(a,v))\leq \modannot^I(B(a))$ so $(\exists R)^\Imc(a)\leq B^\Imc(a)$ and $\Imc$ satisfies $\exists R\sqsubseteq B$.
			\item The case $\exists R^-\sqsubseteq B$ is similar.
		\end{itemize}
	\end{itemize}
	$q^\Imc()=\sum_{\nu\in V}\prod_{S(\vec{t})\in q}  S^\Imc(\nu(\vec{t}))$ where $V$ is the set of all valuations $\nu:\vec{x}\mapsto \Delta^\Imc$ such that $\nu(a)=a^\Imc$ for every constant $a$, so $q^\Imc()= \sum_{h:q\mapsto \inter}\prod_{\beta\in h(q)}  \modannot^I(\beta)$. Since $(\inter,\modannot^I)$ is a model of $\ontology\cup\{q\rightarrow \mn{goal}\}$, it follows that $q^\Imc() \leq \modannot^I(\mn{goal})$. 
	Hence, by definition of $n$, $n\leq q^\Imc() \leq \modannot^I(\mn{goal})$. 
	Since this is true for any model $(\inter,\modannot^I)$, it follows that $n\leq \provmodtot(\ontology\cup\{q\rightarrow \mn{goal}\},\database,\mathbb{N}^\infty,\annot_\mathbb{N},\mn{goal})$, so that $n = \provmodtot(\ontology\cup\{q\rightarrow \mn{goal}\},\database,\mathbb{N}^\infty,\annot_\mathbb{N},\mn{goal})$.
\end{proof}

\subsubsection{Proof of Proposition \ref{prop:provmodtot-capture-incognizant-and-DL}} 
Proposition \ref{prop:provmodtot-capture-incognizant-and-DL} follows from Propositions \ref{prop:provmodtot-capture-incognizant} and \ref{prop:provmodtot-capture-DL}. In particular, note that in the DL-Lite case, if the database facts are annotated with elements from $\mathbb{N}$ rather than $\mathbb{N}^\infty$, then $\provmodtot(\ontology\cup\{q\rightarrow \mn{goal}\},\database,\mathbb{N}^\infty,\annot_\mathbb{N},\mn{goal})=\provmodtot(\ontology\cup\{q\rightarrow \mn{goal}\},\database,\mathbb{N},\annot_\mathbb{N},\mn{goal})$ since there exist finite models of $\ontology$ and $\database$ (because $\ontology$ contains only Datalog rules) such that all facts annotations are in $\mathbb{N}$.
\propprovmodtotcaptureincognizantandDL*

\subsection{Model-Based Semantics}\label{app:model-based-relations}

We start by showing a few lemmas that will be useful to prove results from Section \ref{subsec:model-based-sem} as well as later in the proofs of the results in Table \ref{fig:prop-def}.

\begin{lemma}\label{provmodtot-upperbound}
	The annotated interpretation $(\inter_0,\modannot^{I_0})$ defined by $\inter_0=\{\fact\mid \ontology,\database\models \fact\}$ and $\modannot^{I_0}(\fact)=\atann(\ontology, \database, \semiringshort, \annot,\fact)$ for every $\fact\in \inter_0$
	is a model of $\ontology$ and $(\database,\semiringshort, \annot)$. 
	It follows that $\provmodtot(\ontology, \database, \semiringshort, \annot,\alpha)\sqsubseteq \atann(\ontology, \database, \semiringshort, \annot,\fact)$. 
\end{lemma}
\begin{proof}
	It is easy to check point (1) of the definition: $\database\subseteq \inter_0$ and for every $\fact\in\database$, there exists a derivation tree in $\drvtr{\fact}$ that consists of a single root node labelled $\fact$, so $ \annot(\fact)\sqsubseteq \modannot^{I_0}(\fact)$. 
	
	For point (2), let $\phi(\vec{x},\vec{y} ) \rightarrow H(\vec{x})$ be a rule in $\ontology$ and $h$ be a homomorphism from $\phi(\vec{x},\vec{y})$ to $\inter_0$. 
	By construction of $\inter_0$, for every $\beta\in h(\phi(\vec{x},\vec{y}))$, $\ontology,\database\models \beta$.
	Hence $\ontology,\database\models h(\phi(\vec{x},\vec{y}))$ so $\ontology,\database\models h(H(\vec{x}))$. 
	It follows that $h(H(\vec{x}))\in \inter_0$ by definition of $\inter_0$. 
	
	Let $h(H(\vec{x}))=\alpha$, and let $\Smc$ be the set of all homomorphisms $h':\phi(\vec{x},\vec{y})\mapsto \inter_0$ such that $h'(\vec{x})=h(\vec{x})$. For each such $h'\in \Smc$, let $\gamma^{h'}_1,\dots,\gamma^{h'}_{k_{h'}}$ be the set of facts from $\inter_0$ such that $h'(\phi(\vec{x},\vec{y}))=\gamma^{h'}_1\wedge\dots\wedge \gamma^{h'}_{k_{h'}}$. 
	\begin{enumerate}
		\item By definition of $(\inter_0,\modannot^{I_0})$, $\modannot^{I_0}(\alpha)= \Sigma_{t\in \drvtr{\alpha}} \ann(t)$.
		
		\item For every $h'\in \Smc$, and every $(t_1,\dots,t_{k_{h'}}) \in \drvtr{\gamma^{h'}_1}\times\dots\times  \drvtr{\gamma^{h'}_{k_{h'}}}$, there is a derivation tree $t\in \drvtr{\alpha}$ whose root is $(\alpha,\datrule,h')$ and has subtrees $t_1,\dots,t_{k_{h'}}$. 
		
		\item It follows that $\Sigma_{h'\in \Smc} \Sigma_{(t_1,\dots,t_{k_{h'}}) \in \drvtr{\gamma^{h'}_1}\times\dots\times  \drvtr{\gamma^{h'}_{k_{h'}}}} \Pi_{i=1}^{k_{h'}}\ann(t_i) \sqsubseteq \modannot^{I_0}(\alpha)$. 
		
		\item  By definition of $(\inter_0,\modannot^{I_0})$, for every $h'\in \Smc$ and $1\leq i\leq k_{h'}$, $\modannot^{I_0}(\gamma^{h'}_i)= \Sigma_{t\in \drvtr{\gamma^{h'}_i}}\ann(t)$. 
		
		\item Hence for every $h'\in \Smc$, $\Pi_{\beta\in h'(\phi(\vec{x},\vec{y}))} \modannot^{\inter_0}(\beta) = \Pi_{i=1}^{k_{h'}} \Sigma_{t\in \drvtr{\gamma^{h'}_i}}\ann(t)= \Sigma_{(t_1,\dots,t_{k_{h'}}) \in \drvtr{\gamma^{h'}_1}\times\dots\times  \drvtr{\gamma^{h'}_{k_{h'}}}} \Pi_{i=1}^{k_{h'}}\ann(t_i)$.
	\end{enumerate}
	It follows from (3) and (5) that  $\Sigma_{h':\phi(\vec{x},\vec{y})\mapsto \inter_0, h'(\vec{x})=h(\vec{x})} \Pi_{\beta\in h'(\phi(\vec{x},\vec{y}))} \modannot^{\inter_0}(\beta) \sqsubseteq \modannot^{\inter_0}(h(H(\vec{x})))$.  
	Hence $(\inter_0,\modannot^{I_0})$ is a model of $\ontology$ and $(\database,\semiringshort,\annot)$.
\end{proof}

\begin{lemma}\label{provmodmon-upperbound}
	The set-annotated interpretation $(\inter_0,\modannot^{I_0})$ defined by $\inter_0=\{\fact\mid \ontology,\database\models \fact\}$ and $\modannot^{I_0}(\fact)=\{\ann(t) \mid t\in\drvtr{\fact}\}$ for every $\fact\in \inter_0$ 
	is a model of $\ontology$ and $(\database,\semiringshort, \annot)$. 
\end{lemma}
\begin{proof}
	Let $\inter_0=\{\fact\mid \ontology,\database\models \fact\}$ and $\modannot^{I_0}(\fact)=\{\ann(t) \mid t\in\drvtr{\fact}\}$ for every $\fact\in \inter_0$. 
	We show that $(\inter_0,\modannot^{I_0})$ is a model of $\ontology$ and $(\database,\semiringshort,\annot)$. 
	
	It is easy to check point (1) of the definition: $\database\subseteq \inter_0$ and for every $\fact\in\database$, there exists a derivation tree in $\drvtr{\fact}$ that consists of a single root node labelled $\fact$, so $ \annot(\fact)\in \modannot^{I_0}(\fact)$. 
	
	For point (2), let $\phi(\vec{x},\vec{y} ) \rightarrow H(\vec{x})$ be a rule in $\ontology$ and $h$ be a homomorphism from $\phi(\vec{x},\vec{y})$ to $\inter_0$. 
	By construction of $\inter_0$, for every $\beta\in h(\phi(\vec{x},\vec{y}))$, $\ontology,\database\models \beta$.
	Hence $\ontology,\database\models h(\phi(\vec{x},\vec{y}))$ so $\ontology,\database\models h(H(\vec{x}))$. 
	It follows that $h(H(\vec{x}))\in \inter_0$. 
	Let $h(H(\vec{x}))=\alpha$ and $h(\phi(\vec{x},\vec{y}))=\gamma_1\wedge\dots\wedge\gamma_k$.
	\begin{enumerate}
		\item By definition of $(\inter_0,\modannot^{I_0})$, $\modannot^{I_0}(\alpha)=\{\ann(t) \mid t\in\drvtr{\alpha}\}$.
		
		\item For every $(t_1,\dots,t_k) \in \drvtr{\gamma_1}\times\dots\times  \drvtr{\gamma_{k}}$, there is a derivation tree $t\in \drvtr{\alpha}$ whose root is $(\alpha,\datrule,h')$ and has subtrees $t_1,\dots,t_{k}$. 
		
		\item It follows that $\{ \Pi_{i=1}^{k}\ann(t_i) \mid (t_1,\dots,t_{k}) \in \drvtr{\gamma_1}\times\dots\times  \drvtr{\gamma_{k}}\}\subseteq \modannot^{I_0}(\alpha)$. 
		
		\item  By definition of $(\inter_0,\modannot^{I_0})$, for every $1\leq i\leq k$, $\modannot^{I_0}(\gamma_i)=\{\ann(t)\mid t\in \drvtr{\gamma_i}\}$. 
	\end{enumerate}
	It follows from (3) and (4) that $\{ \Pi_{i=1}^{k} m_i \mid (m_1,\dots,m_k)\in \modannot^{I_0}(\gamma_1)\times\dots\times\modannot^{I_0}(\gamma_k) \}\subseteq \modannot^{\inter_0}(h(H(\vec{x}))$.  
	Hence $(\inter_0,\modannot^{I_0})$ is a model of $\ontology$ and $(\database,\semiringshort,\annot)$ and  $\bigcap_{(\inter,\modannot^I)\models(\ontology,\database,\semiringshort,\annot)}\modannot^I(\fact)=\{\ann(t) \mid t\in \drvtr{\fact}\}$.
\end{proof}

\begin{lemma}\label{ann-model-entailment}
	For both annotated interpretations and set-annotated interpretations, $\ontology,\database\models \fact$ if and only if $\fact\in \inter$ for every model $(\inter,\modannot^I)$ of $\ontology$ and $(\database, \semiringshort, \annot)$. 
\end{lemma}
\begin{proof}
	This follows from the facts that
	\begin{enumerate}
		\item every (set-)annotated model $(\inter,\modannot^I)$ of $\ontology$ and $(\database, \semiringshort, \annot)$ is such that $\inter$ is a model of $\ontology$ and $\database$ by definition of annotated and set-annotated models; and 
		\item for every model $\inter$ of $\ontology$ and $\database$, there exists a (set-)annotated model $(\inter,\modannot^I)$ of $\ontology$ and $(\database, \semiringshort, \annot)$: 
		\begin{itemize}
			\item For $\provmodtot$ such a model can be obtained by setting  $\modannot^{I}(\fact)=\atann(\ontology, \database, \semiringshort, \annot,\fact)$ by Lemma \ref{provmodtot-upperbound}.
			\item For $\provmodmon$ such a model can be obtained by setting  $\modannot^{I}(\fact)=\{\ann(t) \mid t\in\drvtr{\fact}\}$ by Lemma \ref{provmodmon-upperbound}.\qedhere
		\end{itemize}
	\end{enumerate}
\end{proof}

\begin{lemma}\label{lem:greatest-lower-bound-implies-two-ways-ineq-is-eq}
	If $\semiringshort=(K,\plus,\times,\zero,\one)$ is a commutative $\omega$-continuous semiring such that for every $x,y\in K$, the greatest lower bound of $x$ and $y$ is well defined (\ie there exists a unique element $z\in K$ such that $z\sqsubseteq x$, $z\sqsubseteq y$ and every $z'$ such that $z'\sqsubseteq x$ and $z'\sqsubseteq y$ is such that $z'\sqsubseteq z$), then for every $a,b\in K$, $a\sqsubseteq b$ and $b\sqsubseteq a$ implies that $a=b$. 
	In particular, $a\sqsubseteq 0$ implies that $a=0$.
\end{lemma}
\begin{proof}
	Let $a,b\in K$ such $a\sqsubseteq b$ and $b\sqsubseteq a$. Since
	\begin{itemize}
		\item $a\sqsubseteq a$, $a\sqsubseteq b$ and every $z'$ such that $z'\sqsubseteq a$ and $z'\sqsubseteq b$ is such that $z'\sqsubseteq a$; and 
		\item $b\sqsubseteq a$, $b\sqsubseteq b$ and every $z'$ such that $z'\sqsubseteq a$ and $z'\sqsubseteq b$ is such that $z'\sqsubseteq b$,
	\end{itemize}
	then both $a$ and $b$ are the (unique) greatest lower bound of $a$ and $b$. 
\end{proof}

\begin{lemma}\label{provmodtot-lowerbound}
	For every $\alpha$ such that $\ontology,\database\models \alpha$ and $t\in \drvtr{\alpha}$, for every annotated model $(\inter,\modannot^I)$ of $\ontology$ and $(\database,\semiringshort, \annot)$, $\ann(t)\sqsubseteq \modannot^I(\alpha)$, where $\ann(t)=\Pi_{\beta\text{ is a leaf of $t$}} \annot(\beta)$. 
	It follows that for every $t\in \drvtr{\alpha}$, $\ann(t)\sqsubseteq\provmodtot(\ontology, \database, \semiringshort, \annot,\alpha)$. 
\end{lemma}
\begin{proof}
	We show by induction that for every $n\in\mathbb{N}$, for every $\alpha$ such that $\ontology,\database\models \alpha$ and $t\in \drvtr{\alpha}$ which contains at most $n$ inner nodes (\ie nodes that have children), for every model $(\inter,\modannot^I)$ of $\ontology$ and $(\database,\semiringshort, \annot)$, $\ann(t)\sqsubseteq \modannot^I(\alpha)$. 
	\begin{itemize}
		\item Base case: $n=0$. Let $\alpha$ be such that $\ontology,\database\models \alpha$ and let $t\in \drvtrsig$ do not contain any inner node. 
		In this case, $t$ consists of a single node labeled with $\alpha$. Thus 
		$\ann(t)=\annot(\alpha)$. 
		Moreover, $\alpha\in\database$ so for every model $(\inter,\modannot^I)$ of $\ontology$ and $(\database,\semiringshort, \annot)$, $\annot(\alpha)\sqsubseteq \modannot^I(\alpha)$. 
		
		\item Induction step: assume that the property is true for some $n$ and let $\alpha$ be such that $\ontology,\database\models \alpha$ and $t\in \drvtrsig$ contain at most $n+1$ inner nodes. 
		By definition of a derivation tree, the root of $t$ is of the form $(\alpha,\datrule,h)$  for some rule $\datrule=\phi(\vec{x},\vec{y})\rightarrow H(\vec{x})$ that belongs to $\ontology$ and its children are of the form $(\gamma_1,\datrule_1, h_1),\dots,(\gamma_k, \datrule,_k,h_k)$ where $\gamma_1\wedge\dots\wedge\gamma_k = h(\phi(\vec{x},\vec{y}))$  and $\alpha=h(H(\vec{x}))$. 
		For $1\leq i\leq k$, let $t_i$ be the subtree of $t$ rooted in $(\gamma_i,\datrule_i,h_i)$. 
		$t_i$ is a derivation tree of $\gamma_i$ \wrt $\ontology,\database$ so $\ontology,\database\models \gamma_i$. 
		Moreover $t_i$ contains at most $n$ inner nodes.  
		Let $(\inter,\modannot^I)$ be a model of $\ontology$ and $(\database,\semiringshort, \annot)$. For $1\leq i\leq k$, $\gamma_i\in\inter$ by Lemma \ref{ann-model-entailment} and 
		by induction hypothesis, $\ann(t_i) \sqsubseteq \modannot^I(\gamma_i)$. 
		Moreover, since $h$ is a homomorphism from the body of $\datrule$ to $\{\gamma_1,\dots,\gamma_k\}$ such that $h(H(\vec{x}))=\alpha$, then $\Pi_{i=1}^k  \modannot^I(\gamma_i) \sqsubseteq \modannot^I(\alpha)$. 
		It follows that $\Pi_{i=1}^k  \ann(t_i) \sqsubseteq \modannot^I(\alpha)$, hence $\ann(t)\sqsubseteq\modannot^I(\alpha)$. \qedhere
	\end{itemize}
\end{proof}

\begin{lemma}\label{provmodmon-lemma}
	$\provmodmon(\ontology, \database, \semiringshort, \annot,\fact)= \sum_{\{\ann(t) \mid t\in \drvtr{\fact}\}} \ann(t)$, where $\ann(t)=\Pi_{\beta\text{ is a leaf of $t$}} \annot(\beta)$. 
\end{lemma}
\begin{proof}
	We show by induction that for every $n\in\mathbb{N}$, for every $\alpha$ such that $\ontology,\database\models \alpha$ and $t\in \drvtr{\alpha}$ which contains at most $n$ inner nodes (\ie nodes that have children), for every set-annotated model $(\inter,\modannot^I)$ of $\ontology$ and $(\database,\semiringshort, \annot)$, $\ann(t)\in \modannot^I(\alpha)$. 
	\begin{itemize}
		\item Base case: $n=0$. Let $\alpha$ be such that $\ontology,\database\models \alpha$ and let $t\in \drvtrsig$ do not contain any inner node. 
		In this case, $t$ consists of a single node labeled with $\alpha$. Thus 
		$\ann(t)=\annot(\alpha)$. 
		Moreover, $\alpha\in\database$ so for every model $(\inter,\modannot^I)$ of $\ontology$ and $(\database,\semiringshort, \annot)$, $\annot(\alpha)\in\modannot^I(\alpha)$. 
		
		\item Induction step: assume that the property is true for some $n$ and let $\alpha$ be such that $\ontology,\database\models \alpha$ and $t\in \drvtrsig$ contain at most $n+1$ inner nodes. 
		By definition of a derivation tree, the root of $t$ is of the form $(\alpha,\datrule,h)$  for some rule $\datrule=\phi(\vec{x},\vec{y})\rightarrow H(\vec{x})$ that belongs to $\ontology$ and its children are of the form $(\gamma_1,\datrule_1, h_1),\dots,(\gamma_k, \datrule,_k,h_k)$ where $\gamma_1\wedge\dots\wedge\gamma_k = h(\phi(\vec{x},\vec{y}))$  and $\alpha=h(H(\vec{x}))$. 
		For $1\leq i\leq k$, let $t_i$ be the subtree of $t$ rooted in $(\gamma_i,\datrule_i,h_i)$. 
		$t_i$ is a derivation tree of $\gamma_i$ \wrt $\ontology,\database$ so $\ontology,\database\models \gamma_i$. 
		Moreover $t_i$ contains at most $n$ inner nodes.  
		Let $(\inter,\modannot^I)$ be a model of $\ontology$ and $(\database,\semiringshort, \annot)$. For $1\leq i\leq k$, $\gamma_i\in\inter$ by Lemma \ref{ann-model-entailment} and 
		by induction hypothesis, $\ann(t_i) \in \modannot^I(\gamma_i)$. 
		Moreover, since $h$ is a homomorphism from the body of $\datrule$ to $\{\gamma_1,\dots,\gamma_k\}$ such that $h(H(\vec{x}))=\alpha$, then $\{\Pi_{i=1}^k  m_i\mid (m_1,\dots, m_k)\in \modannot^I(\gamma_1)\times\dots\times\modannot^I(\gamma_k) \} \subseteq \modannot^I(\alpha)$. 
		It follows that $\Pi_{i=1}^k  \ann(t_i) \in \modannot^I(\alpha)$, hence $\ann(t)\in\modannot^I(\alpha)$. 
	\end{itemize}
	Hence for every $\alpha$ such that $\ontology,\database\models \alpha$ and $t\in \drvtr{\alpha}$, for every model $(\inter,\modannot^I)$ of $\ontology$ and $(\database,\semiringshort, \annot)$, $\ann(t)\in \modannot^I(\alpha)$. Therefore $\{\ann(t) \mid t\in \drvtr{\fact}\}\subseteq \bigcap_{(\inter,\modannot^I)\models(\ontology,\database,\semiringshort,\annot)}\modannot^I(\fact)$. 
	
	Since by Lemma \ref{provmodmon-upperbound} the set-annotated interpretation $(\inter_0,\modannot^{I_0})$ defined by $\inter_0=\{\fact\mid \ontology,\database\models \fact\}$ and $\modannot^{I_0}(\fact)=\{\ann(t) \mid t\in\drvtr{\fact}\}$ for every $\fact\in \inter_0$ is a model of $\ontology$ and $(\database,\semiringshort, \annot)$, it follows that $ \bigcap_{(\inter,\modannot^I)\models(\ontology,\database,\semiringshort,\annot)}\modannot^I(\fact)=\{\ann(t) \mid t\in \drvtr{\fact}\}$.
\end{proof}

\propmodelBasedRelationship*
\begin{proof}
	$\provmodtot\sqsubseteq \atann$ follows from Lemma \ref{provmodtot-upperbound} and $\provmodtot\sqsubseteq \atann$ follows from Lemma \ref{provmodmon-lemma}.
\end{proof}

\propmodelBasedSameAsTreeBasedForIdempotent*
\begin{proof}
	Let $\semiringshort$ be a commutative $\plus\,$-idempotent $\omega$-continuous semiring. In this case, 
	\begin{align*}
		\atann(\ontology, \database, \semiringshort, \annot,\alpha)=&\Sigma_{t\in \drvtrsig}\ann(t)\\
		= &\sum_{\{\ann(t) \mid t\in \drvtr{\fact}\}} \ann(t)
	\end{align*}
	By Lemma~\ref{provmodmon-lemma}, it follows that $\provmodmon(\ontology, \database, \semiringshort, \annot,\fact)=\atann(\ontology, \database, \semiringshort, \annot,\alpha)$.

	Moreover, by Lemma~\ref{provmodtot-upperbound}, it follows that 
	$\provmodtot(\ontology, \database, \semiringshort, \annot,\fact)\sqsubseteq \sum_{\{\ann(t) \mid t\in \drvtr{\fact}\}} \ann(t)$, 
	and by Lemma \ref{provmodtot-lowerbound}, for every $t\in \drvtr{\fact}$, $\ann(t)\sqsubseteq \provmodtot(\ontology, \database, \semiringshort, \annot,\fact)$, \ie $\provmodtot(\ontology, \database, \semiringshort, \annot,\fact)=\ann(t)+S_t$ for some $S_t\in\semiringshort$. 
	Since $\semiringshort$ is $\plus\,$-idempotent, 
	\begin{align*}
		\provmodtot(\ontology, \database, \semiringshort, \annot,\fact)=&\sum_{\{\ann(t) \mid t\in \drvtr{\fact}\}} \provmodtot(\ontology, \database, \semiringshort, \annot,\fact)\\
		=& \sum_{\{\ann(t) \mid t\in \drvtr{\fact}\}} (\ann(t)+S_t)\\
		=& \sum_{\{\ann(t) \mid t\in \drvtr{\fact}\}} \ann(t)+\sum_{\{\ann(t) \mid t\in \drvtr{\fact}\}}S_t
	\end{align*}
	Hence $\sum_{\{\ann(t) \mid t\in \drvtr{\fact}\}} \ann(t)\sqsubseteq \provmodtot(\ontology, \database, \semiringshort, \annot,\fact)$. 
	Since $\provmodtot$ is defined for $\semiringshort$ such that for every $x,y\in K$, the greatest lower bound of $x$ and $y$ is well defined, by Lemma \ref{lem:greatest-lower-bound-implies-two-ways-ineq-is-eq}, 
	$\sum_{\{\ann(t) \mid t\in \drvtr{\fact}\}} \ann(t)\sqsubseteq \provmodtot(\ontology, \database, \semiringshort, \annot,\fact)$ and $\provmodtot(\ontology, \database, \semiringshort, \annot,\fact)\sqsubseteq \sum_{\{\ann(t) \mid t\in \drvtr{\fact}\}} \ann(t)$ implies that $\provmodtot(\ontology, \database, \semiringshort, \annot,\fact) =\sum_{\{\ann(t) \mid t\in \drvtr{\fact}\}} \ann(t)=\atann(\ontology, \database, \semiringshort, \annot,\alpha)$.
\end{proof}

\propmodelBasedVSTreeBased*
\begin{proof}
	By Lemma~\ref{provmodtot-upperbound}, $\provmodtot(\ontology, \database, \mathbb{N}^\infty\llbracket\semiringVars\rrbracket, \annot_\semiringVars,\fact)\sqsubseteq \atann(\ontology, \database, \mathbb{N}^\infty\llbracket\semiringVars\rrbracket, \annot_\semiringVars,\fact)$, \ie $\atann(\ontology, \database, \mathbb{N}^\infty\llbracket\semiringVars\rrbracket, \annot_\semiringVars,\fact)=\provmodtot(\ontology, \database, \mathbb{N}^\infty\llbracket\semiringVars\rrbracket, \annot_\semiringVars,\fact)+S$ for some $S\in \mathbb{N}^\infty\llbracket\semiringVars\rrbracket$. 
	Moreover, for every monomial $m$ in $S$, $m$ occurs in $\atann(\ontology, \database, \mathbb{N}^\infty\llbracket\semiringVars\rrbracket, \annot_\semiringVars,\fact)$ so there exists $t\in \drvtr{\fact}$ such that $m=\ann(t)$. By Lemma \ref{provmodtot-lowerbound}, $m\sqsubseteq \provmodtot(\ontology, \database, \mathbb{N}^\infty\llbracket\semiringVars\rrbracket, \annot_\semiringVars,\fact)$. 
	
	By Lemma~\ref{provmodmon-lemma}, $\provmodmon(\ontology, \database,  \mathbb{N}^\infty\llbracket\semiringVars\rrbracket, \annot_\semiringVars,\fact)= \sum_{\{\ann(t) \mid t\in \drvtr{\fact}\}} \ann(t)$ where $\ann(t)=\Pi_{\beta\text{ is a leaf of $t$}} \annot_\semiringVars(\beta)$. 
	It is easy to see that this is exactly the sum of monomials that occur in $\atann(\ontology, \database, \mathbb{N}^\infty\llbracket\semiringVars\rrbracket, \annot_\semiringVars,\fact)$.
\end{proof}

This final lemma shows that $\provmodtot$ and $\provmodmon$ fulfill the conditions of Definition \ref{def:provenance}.

\begin{lemma}\label{ann-model-null-prov}
	If $\ontology,\database\not\models \fact$ then $\provmodtot(\ontology,\database, \semiringshort, \annot,\fact)=\provmodmon(\ontology,\database, \semiringshort, \annot,\fact)=0$. 
	
	If $\semiringshort$ is positive, for $\prov=\provmodtot$ and $\prov=\provmodmon$, $\prov(\ontology, \database, \semiringshort, \annot,\fact)=0$ implies $\ontology, \database\not\models \fact$. 
\end{lemma}
\begin{proof}
	In $\provmodtot$ case,  by Lemma \ref{provmodtot-upperbound}, $\provmodtot(\ontology, \database, \semiringshort, \annot,\alpha)\sqsubseteq \atann(\ontology, \database, \semiringshort, \annot,\fact)$. 
	If $\ontology,\database\not\models \fact$, $\atann(\ontology, \database, \semiringshort, \annot,\fact)=0$ so by Lemma \ref{lem:greatest-lower-bound-implies-two-ways-ineq-is-eq}, $\provmodtot(\ontology, \database, \semiringshort, \annot,\alpha)=0$. 
	By Lemma \ref{provmodtot-lowerbound}, 
	for every $t\in \drvtr{\alpha}$, $\ann(t)\sqsubseteq\provmodtot(\ontology, \database, \semiringshort, \annot,\alpha)$. 
	Hence, if $\semiringshort$ is positive, $\provmodtot(\ontology, \database, \semiringshort, \annot,\fact)=0$ implies that $\ann(t)=0$ for every $t\in \drvtr{\alpha}$. Since databse facts cannot be annotated with $0$, this means that $\drvtr{\alpha}=\emptyset$ and $\ontology, \database\not\models \fact$. 
	
	In $\provmodmon$ case, by Lemma \ref{ann-model-entailment}, $\ontology,\database\not\models \fact$ if and only if there exists a set-annotated model $(\inter,\modannot^I)$ of $\ontology$ and $(\database, \semiringshort, \annot)$ such that $\fact\notin \inter$. Hence, 
	\begin{itemize}
		\item $\ontology,\database\not\models \fact$ implies that $\bigcap_{(\inter,\modannot^I)\models (\ontology,\database,\semiringshort, \annot)} \modannot^I(\fact)=\emptyset$ and $\provmodmon(\ontology, \database, \semiringshort, \annot,\fact)= \Sigma_{k\in \bigcap_{(\inter,\modannot^I)\models (\ontology,\database,\semiringshort, \annot)} \modannot^I(\fact)} k=0$; and 
		\item if $\semiringshort$ is positive, $\provmodmon(\ontology, \database, \semiringshort, \annot,\fact)=0$ implies that $\bigcap_{(\inter,\modannot^I)\models (\ontology,\database,\semiringshort, \annot)} \modannot^I(\fact)=\emptyset$ and $\ontology,\database\not\models \fact$.\qedhere
	\end{itemize}
\end{proof}

\subsection{Execution- and Tree-Based Semantics}

\subsubsection{Naive Evaluation /\ All Trees}

We denote $(D,\semiringshort,\naiveannot)$ by $(\naive^{0},\semiringshort,\naiveannot^0)$, and for $i\ge 1$ we denote $\naive^{i}(\ontology,\database,\semiringshort,\naiveannot)$ by $(\naive^{i},\semiringshort,\naiveannot^i)$.
\begin{lemma}
	\label{lem:deductiontree-iteration}
	For every fact $\alpha \in \naive^{i}$, it holds that 
	\[
	\naiveannot^i(\alpha) =	\sum_{\substack{ t\in  \drvtrsig\\ \text{ is of depth}\le i}}   \ann(t)
	\]
\end{lemma}

\begin{proof}
	We prove the claim by induction on $i$.
	\paragraph{Induction Basis}
	If $i=0$ then $	\naiveannot^0 = D$ and the claim holds since
	$\naiveannot^0(\alpha) = \lambda(\alpha)$ if $\alpha\in D$, or $0$ otherwise. And, by definition, 	$\sum_{\substack{ t\in  \drvtrsig\\ \text{ is of depth}\le 0}}   \ann(t) = \lambda(\alpha)$ if 
	$\alpha\in D$, or $0$ otherwise.
	
	\paragraph{Induction Step}
	By definition we have	\[
	\naiveannot^i(\alpha) =	\sum_{\substack{ t\in  \drvtrsig\\ \text{ is of depth}\le i}}   \ann(t) = 
	\sum_{\substack{ t\in  \drvtrsig\\ \text{ is of depth }\le i-1}}   \ann(t) \oplus
	\sum_{\substack{ t\in  \drvtrsig\\ \text{ is of depth } i}}   \ann(t)
	\]
	By definition, 
	\[
	\sum_{\substack{ t\in  \drvtrsig\\ \text{ is of depth } \le i}}   \ann(t) = 
	\sum_{\substack{ t\in  \drvtrsig\\ \text{ is of depth } \le i-1}}   \ann(t) \oplus
	\sum_{\substack{ t\in  \drvtrsig\\ \text{ is of depth } i}}   \ann(t)
	\]
	By induction hypothesis,
	\[
	\sum_{\substack{ t\in  \drvtrsig\\ \text{ is of depth } \le i}}   \ann(t) = 
	\naiveannot^{i-1}(\alpha) \oplus
	\sum_{\substack{ t\in  \drvtrsig\\ \text{ is of depth } i}}   \ann(t)
	\]
	
	Note that from the definition $\naive^{i}(\ontology,\database,\semiringshort,\annot) \df \consop(\naive^{i-1}(\ontology,\database,\semiringshort,\annot))\cup (\database,\semiringshort,\annot)$
	we can conclude (using a simple induction) that 
	\[
	\naive^{i}(\ontology,\database,\semiringshort,\annot) = \underbrace{\consop( \cdots \consop( }_{i\text{ times}} \naive^{0}(\ontology,\database,\semiringshort,\annot) \,)\cdots ) \oplus 
	\underbrace{\consop( \cdots \consop( }_{i-1\text{ times}} \naive^{0}(\ontology,\database,\semiringshort,\annot) \,)\cdots ) \oplus 
	\cdots \oplus \naive^{0}(\ontology,\database,\semiringshort,\annot)
	\]
	By the definition of $\consop$ and the above we conclude the desired equivalence.
\end{proof}

\propnaive*
\begin{proof}
	Let us denote $\naive^{i}(\ontology,\database,\semiringshort,\annot)$ by $(\naive^{i},\semiringshort,\naiveannot^i)$.
	From the convergence of the naive Datalog evalutaion algoritm~\cite{DBLP:books/aw/AbiteboulHV95}, we can conclude that there is a $k$ such that  $\naive^{\ell} =  \naive^{k}$ for every $\ell\ge k$. 
	Let $\alpha \in  \naive^{k}$. It suffices to show that $\naiveannot^i(\alpha) $ converges. 
	By Lemma~\ref{lem:deductiontree-iteration}, 
	\[
	\naiveannot^i(\alpha) = \sum_{\substack{ t\in  \drvtrsig\\ \text{ is of depth } 0}}   \ann(t) \oplus
	\cdots \oplus 
	\sum_{\substack{ t\in  \drvtrsig\\ \text{ is of depth } i}}   \ann(t).
	\]
	Since $\semiringshort$ is $\omega$-continuous this sum converges, and therefore $\sup_{i} \naiveannot^i(\alpha)$ exists, which concludes the proof.
\end{proof}

\propexecutiontreeconnection*
\begin{proof}
	By Lemma \ref{lem:deductiontree-iteration}, the sequence of provenance of a fact $\alpha$ over the $\naive^{i}(\ontology,\database,\semiringshort,\annot)$ converge to the sum of the derivation trees of $\alpha$. 
	Thus, we can conclude that the sequence of  $\naive^{i}(\ontology,\database,\semiringshort,\annot)$ converge to an annotated database that the  $\consop(\naive^{\infty}(\ontology,\database,\semiringshort,\annot)) \cup (\database,\semiringshort,\annot)$ is equal to  $\naive^{\infty}(\ontology,\database,\semiringshort,\annot)$ and therefore, 
	$\provnaive$ is equal to $\atann$.
\end{proof}

\subsubsection{Optimized Naive Evaluation /\ Minimal Depth Trees}
\propopti*
\begin{proof}
	The existence of $k$ such that $\opti^{\ell}= \opti^k$ for every $\ell \ge k$ is a consequence of the convergence of the original seminaive Datalog algorithm~\cite{DBLP:books/aw/AbiteboulHV95}.
	It suffices to show that for the same $k$ and for every fact $\beta \in \opti^{k}$  if holds that $\semiannot^k(\beta) =\semiannot^{\ell}(\alpha)$ whenever $\ell \ge k$.
	This, indeed, follows directly from the definition of the operator $\diffop$. 
\end{proof}

\begin{lemma}
	\label{lem:oh}
	For every fact $\beta \in \opti^{i}$ where $(\opti^{i},\semiringshort,\opti^i) \df \opti^{i}(\ontology,\database,\semiringshort,\annot)$, the following holds
	\[
	\opti^i(\beta) = 	\sum_{\substack{ t\in  \drvtrsig\\ \text{ is of minimal depth with depth } \le i}}   \ann(t)
	\]
\end{lemma}
\begin{proof}
	The proof is a direct proof by induction on $i$.
\end{proof}

\propexecutiontreeconnectionmdt*
\begin{proof}
	The proof is straightforward from Lemma~\ref{lem:oh}.
\end{proof}

\subsubsection{Seminaive Evaluation /\ Hereditary Minimal Depth Trees}
We denote $(D,\semiringshort,\semiannot)$ by $(\semi^{0},\semiringshort,\semiannot^0)$, and for $i\ge 1$ we denote $\naive^{i}(\ontology,\database,\semiringshort,\semiannot)$ by $(\semi^{i},\semiringshort,\semiannot^i)$.

\propsemi*
\begin{proof}
	The existence of $k$ such that $\semi^{\ell}= \semi^k$ for every $\ell \ge k$ is a consequence of the convergence of the original seminaive Datalog algorithm~\cite{DBLP:books/aw/AbiteboulHV95}.
	It suffices to show that for the same $k$ and for every fact $\alpha \in \semi^{k}$  if holds that $\semiannot^k(\alpha) =\semiannot^{\ell}(\alpha)$ whenever $\ell \ge k$.
	This, indeed, follows directly from the definition of the operator $\diffop$. 
\end{proof}

\begin{lemma}
	\label{lem:sh}
	For every fact $\alpha \in \semi^{i}$ where $(\semi^{i},\semiringshort,\annot^i) \df \semi^{i}(\ontology,\database,\semiringshort,\annot)$, the following holds
	\[
	\annot^i(\alpha) = 	\sum_{\substack{ t\in  \drvtrsig\\ \text{ is of hereditary minimal depth with depth } \le i}}   \ann(t)
	\]
\end{lemma}
\begin{proof}
	The proof is a direct proof by induction on $i$.
\end{proof}

\propexecutiontreeconnectionhmdt*
\begin{proof}
	The proof is straightforward from Lemma~\ref{lem:sh}.
\end{proof}

\proptreeconnections*
\begin{proof}
	It is straightforward that 
	$ \mtann \sqsubseteq \atann \ \text{ and }\quad  \mdtann \sqsubseteq \atann$.
	Since every hereditary minimal depth tree is also non-recursive, we have $\rmdtann \sqsubseteq \mtann $. Since every hereditary minimal depth tree is also a minimal depth tree, we have $ \mdtann \sqsubseteq \atann$.
\end{proof}

\subsection{Non-Recursive Tree-Based Semantics}
\propalltreenonrecabsorptive*
\begin{proof}
	Assume that $\semiringshort$ is a commutative absorptive $\omega$-continuous semiring. 
	The sum $\atann(\ontology,\database,\semiringshort,\annot,\fact)$ of the annotations of the (possibly infinitely many) derivations trees in $\drvtrsig$ is defined as the supremum of the set of the sums of the annotations of any finite subset of $\drvtrsig$. We show that this supremum coincides with $\mtann(\ontology,\database,\semiringshort,\annot,\fact)$.
	
	It is clear that $\mtann(\ontology,\database,\semiringshort,\annot,\fact)$ is a lower bound of the supremum, since the set of non-recursive trees is a subset of $\drvtrsig$. 
	
	Conversely, let $\{t_1,\ldots,t_n\}$ be a finite subset of $\drvtrsig$. Let us assume that $\{t_1,\ldots,t_n\}$ contains all the non-recursive trees of $\drvtrsig$. Let us define a non-recursive version $t'_i$ of $t_i$ as follows. A simplification of $t_i$ is obtained from $t_i$ by picking a node $(\beta,\datrule,h)$ which has a descendant $n$ of the form $(\beta,\datrule', h')$ or $\beta$, and by replacing the subtree rooted in $(\beta,\datrule,h)$ by the subtree rooted in $n$. $t'_i$ is any tree obtained from $t_i$ on which no simplification is performable. It holds that $\ann(t_i) = \ann(t'_i)\times e$ for some $e\in K$, which is the product of the labels of the leaves that have been removed from $t_i$ through successive steps of simplification. Hence, by absorptivity of $\semiring$, $\Sigma_{i=1}^n \ann(t_i)=\Sigma_{i=1}^n \ann(t'_i)=\mtann(\ontology,\database,\semiringshort,\annot,\fact)$. 
	Since this is true for any finite subset $\{t_1,\ldots,t_n\}$ of $\drvtrsig$ that contains the non-recursive trees of $\drvtrsig$, this concludes the proof.
\end{proof}
\section{Discussion and Proofs for Section \ref{sec:properties}}

\subsection{Commutation with Homomorphisms and Universal Semirings}\label{app:commutation-universal}

The following proposition explicits the connection between the satisfaction of the Commutation with Homomorphisms Property or Commutation with $\omega$-Continuous Homomorphisms Property by a provenance semantics $\prov$ and the ability to use a provenance semiring universal for its semiring domain to factor the computations.

\begin{proposition}
	If $\prov$ satisfies the Commutation with  (resp.\ $\omega$-Continuous) Homomorphisms Property and $\mi{Prov}(\semiringVars)$ is universal for its semiring domain (resp.\ which contains only $\omega$-continuous semirings), $\prov(\ontology, \database,\semiringshort, \annot, \fact)=h(\prov(\ontology, \database,\mi{Prov}(\semiringVars),\annot_\semiringVars,\fact))$ where  $h$ is the unique (resp.\ $\omega$-continuous) semiring homomorphism that extends $\nu:\semiringVars\rightarrow K$ where $\nu(x)=\annot(\annot_\semiringVars^-(x))$ for every $x\in \semiringVars$.
	
	Conversely, if $S$ is a set of (resp.\ $\omega$-continuous) semirings such that $\mi{Prov}(\semiringVars)$ is universal for $S$ and for every $\semiringshort\in S$, it holds that $\prov(\ontology, \database,\semiringshort, \annot, \fact)=h(\prov(\ontology, \database,\mi{Prov}(\semiringVars),\annot_\semiringVars,\fact))$ where  $h$ is the unique (resp.\ $\omega$-continuous) semiring homomorphism that extends $\nu:\semiringVars\rightarrow K$ where $\nu(x)=\annot(\annot_\semiringVars^-(x))$ for every $x\in \semiringVars$, then the restriction $\prov_S$ of $\prov$ to the semiring domain $S$ satisfies the Commutation with (resp.\ $\omega$-Continuous) Homomorphisms Property.
\end{proposition}

\begin{proof}
	Assume that $\prov$ satisfies the Commutation with Homomorphisms Property and let $\mi{Prov}(\semiringVars)$ be a universal semiring for the semiring domain of $\prov$.
	Let $\ontology$ be a Datalog program, $(\database, \semiringshort, \annot)$ be an annotated database and $\alpha$ be a fact. 
	Let $\annot_\semiringVars$ associate a distinct variable from $\semiringVars$ to each fact of $\database$ and $h: \mi{Prov}(\semiringVars) \rightarrow K$ be the unique semiring homomorphism that extends $\nu:\semiringVars\rightarrow K$ where $\nu(x)=\annot(\annot_\semiringVars^-(x))$ for every $x\in \semiringVars$ (the existence of $h$ is guaranteed by the fact that $\mi{Prov}(\semiringVars)$ specializes correctly to $\semiringshort$ by definition of a universal semiring). 
	By the Commutation with Homomorphisms Property, we have 
	$h(\prov(\ontology, \database,\mi{Prov}(\semiringVars),\annot_\semiringVars,\alpha))=\prov(\ontology, \database,\semiringshort, h\circ\annot_\semiringVars, \alpha)$. 
	Moreover, for every $\beta\in \database$, $h\circ\annot_\semiringVars(\beta) =\nu(\annot_\semiringVars(\beta))= \annot(\annot_\semiringVars^-(\annot_\semiringVars(\beta)))=\annot(\beta)$, so $h\circ\annot_\semiringVars=\annot$. 
	Hence $h(\prov(\ontology, \database,\mi{Prov}(\semiringVars),\annot_\semiringVars,\fact))=\prov(\ontology, \database,\semiringshort, \annot, \fact)$.
	
	Assume that $S$ is a set of semirings such that $\mi{Prov}(\semiringVars)$ is universal for $S$ and for every $\semiringshort\in S$, it holds that $\prov(\ontology, \database,\semiringshort, \annot, \fact)=h(\prov(\ontology, \database,\mi{Prov}(\semiringVars),\annot_\semiringVars,\fact))$ where  $h$ is the unique semiring homomorphism that extends $\nu:\semiringVars\rightarrow K$ where $\nu(x)=\annot(\annot_\semiringVars^-(x))$ for every $x\in \semiringVars$, and let $\prov_S$ be the restriction of $\prov$ to the semiring domain~$S$. 
	Let $\semiringshort_1$ and $\semiringshort_2$ be two commutative semirings in $S$ such that there is a semiring homomorphism $h$ from $\semiringshort_1$ to $\semiringshort_2$.
	Let $\ontology$ be a Datalog program, $(\database, \semiringshort_1, \annot)$ be an annotated database and $\alpha$ be a fact. 
	Let $h_1: \mi{Prov}(\semiringVars) \rightarrow K_1$ (resp.\ $h_2: \mi{Prov}(\semiringVars) \rightarrow K_2$) be the unique semiring homomorphism that extends $\nu_1:\semiringVars\rightarrow K_1$ (resp.\ $\nu_2:\semiringVars\rightarrow K_2$) where $\nu_1(x)=\annot(\annot_\semiringVars^-(x))$ for every $x\in \semiringVars$ and $\nu_2(x)=h(\nu_1(x))=h(\annot(\annot_\semiringVars^-(x)))$. 
	Applying the hypothesis with $\semiringshort_1$ gives 
	$\prov_S(\ontology, \database,\semiringshort_1,\annot,\alpha)=h_1(\prov_S(\ontology, \database,\mi{Prov}(\semiringVars),\annot_\semiringVars,\alpha))$.   
	Hence $h(\prov_S(\ontology, \database,\semiringshort_1,\annot,\alpha))=h(h_1(\prov_S(\ontology, \database,\mi{Prov}(\semiringVars),\annot_\semiringVars,\alpha)))$. 
	Since for every $x\in\semiringVars$, $h(h_1(x))=h(\nu_1(x))=\nu_2(x)$, and $h_2$ is the unique semiring homomorphism that extends $\nu_2$, it follows that $h\circ h_1=h_2$. 
	Thus $h(\prov_S(\ontology, \database,\semiringshort_1,\annot,\alpha))=h_2(\prov_S(\ontology, \database,\mi{Prov}(\semiringVars),\annot_\semiringVars,\alpha))$. 
	Moreover, applying the hypothesis with $\semiringshort_2$ gives $h_2(\prov_S(\ontology, \database,\mi{Prov}(\semiringVars),\annot_\semiringVars,\alpha))=\prov_S(\ontology, \database,\semiringshort_2,h\circ \annot,\alpha)$. 
	Hence $h(\prov_S(\ontology, \database,\semiringshort_1,\annot,\alpha))=\prov_S(\ontology, \database,\semiringshort_2,h\circ \annot,\alpha)$ and 
	$\prov_S$ satisfies the Commutation with Homomorphisms Property.
	
	The proof for the $\omega$-continuous case is similar, but all semirings are assumed to be $\omega$-continuous.
\end{proof}

\subsection{Proof of Proposition \ref{prop:grounding-joint-alt-to-algebra}}
\propgroundingjointalttoalgebra*
\begin{proof}
	Assume that $\prov$ satisfies Properties \ref{prop:joint-alternative} and \ref{prop-parsimony}, and is such that $\prov(\ontology, \database, \semiringshort, \annot, \fact)=\prov(\ontology_\database, \database, \semiringshort, \annot,\fact)$ where $\ontology_\database$ is the grounding of $\ontology$ \wrt $\database$. 
	Let $\ontology$ be a UCQ defined Datalog program with nullary predicate $\mn{goal}$ in rule heads, and $(\database, \semiringshort, \annot)$ be an annotated database that does not contain $\mn{goal}$.  
	By assumption, 
	$\prov(\ontology, \database, \semiringshort, \annot,\mn{goal})=\prov(\ontology_\database, \database, \semiringshort, \annot,\mn{goal})$. 
	Since $\ontology$ is UCQ defined, 
	$\ontology_\database=\{\bigwedge_{j=1}^{n_i} \alpha^i_j\rightarrow \mn{goal} \mid 1\leq i\leq m\}$. 
	By Property~\ref{prop:joint-alternative}, 
	$\prov(\ontology_\database, \database, \semiringshort, \annot,\mn{goal})= \Sigma_{i=1}^m \Pi_{j=1}^{n_i} \prov(\emptyset, \database, \semiringshort, \annot, \alpha^i_j)$. 
	Either (i) $\alpha^i_j\in\database$ and by Property \ref{prop-parsimony} $\prov(\emptyset, \database, \semiringshort, \annot, \alpha^i_j)=\annot(\alpha^i_j)$, or (ii) $\emptyset,\database\not\models\alpha^i_j$ and $\prov(\emptyset, \database, \semiringshort, \annot, \alpha^i_j)=0$ by point (\ref{defprov:null}) of Definition \ref{def:provenance}. 
	Each product $\Pi_{j=1}^{n_i} \prov(\emptyset, \database, \semiringshort, \annot, \alpha^i_j)$ is then either equal to $0$ is some of the $\alpha^i_j$ does not belong to $\database$, or equal to $\Pi_{j=1}^{n_i} \annot(\alpha^i_j)$. 
	Note that since $\bigwedge_{j=1}^{n_i} \alpha^i_j\rightarrow \mn{goal}$ is an instantiation of some rule $\phi(\vec{y} ) \rightarrow \mn{goal}\in\Sigma$, it is the case that all $\alpha^i_j$ in such a product belong to $\database$ exactly when there exists a homomorphism $h$ from $\phi(\vec{y} )$ to $\database$ such that $h(\phi(\vec{y} ))= \bigwedge_{j=1}^{n_i} \alpha^i_j$. 
	It follows that $\prov(\ontology, \database, \semiringshort, \annot,\mn{goal}) = \Sigma_{\phi(\vec{y} ) \rightarrow \mn{goal}\in\Sigma, h : \phi(\vec{y} ) \rightarrow \database} \Pi_{p(\vec{y})\in \phi(\vec{y} )} \annot(h(p(\vec{y})))$ where $h : \phi(\vec{y} ) \rightarrow \database$ denotes that $h$ is a homomorphism from $\phi(\vec{y} )$ to $\database$. 
	This is precisely the relational database provenance of the equivalent UCQ $Q(\Sigma)= \bigvee_{\phi(\vec{y} ) \rightarrow \mn{goal}\in\Sigma} \exists \vec{y} \phi(\vec{y} )$ over $(\database, \semiringshort, \annot)$. 
\end{proof}

\subsection{Usable Facts Definition}\label{app:def-usable}
We formalize usability of a fact with the following construction. 
The adornment of a predicate $p$ by $\fact$ is the (fresh) predicate $p^\fact$. The adornment of an atom $p(t_1,\ldots,t_n)$ by $\fact$ is $p^\fact(t_1,\ldots,t_n)$. An adornment of a rule $\phi(\vec{x},\vec{y} ) \rightarrow p(\vec{x})$ by $\fact$ is a rule of the shape $\phi^\fact(\vec{x},\vec{y} )\rightarrow p(\vec{x})$ or of the shape $\phi^\fact(\vec{x},\vec{y} )\rightarrow  p^\fact(\vec{x})$,
where $\phi^\fact(\vec{x},\vec{y} )$ is equal to $\phi(\vec{x},\vec{y} )$, except for one atom which has been replaced by its adornment by $\fact$. 
A fact $\fact$ is \emph{adornment-usable} to derive $\beta$ w.r.t. $\ontology$ and $\database$ if $\database^\fact,\ontology \cup \ontology^\fact \models \beta^\fact$ where 
$\database^\fact = \database \cup \{\fact^\fact\}$ 
and $\ontology^\fact$ is the set of adornment of rules from $\ontology$ by $\fact$. 

\begin{proposition}
	A fact $\fact$ is usable to derive $\beta$ w.r.t. $\ontology$ and $\database$ if and only if it is adornment-usable to derive $\beta$ w.r.t. $\ontology$ and $\database$.
\end{proposition}

\begin{proof}
	Let us assume that $\fact$ is usable to derive $\beta$ w.r.t. $\ontology$ and $\database$, and let $t$ be a derivation tree for $\beta$ having a leaf equal to $\fact$. We proof that $\fact$ is adornment-usable to derive $\beta$ w.r.t. $\ontology$ and $\database$ by induction on the depth of $t$.
	
	\begin{itemize}
		\item Depth $0$: the derivation tree is restricted to $\fact = \beta$. Hence $\beta^\fact = \fact^\fact$, and the derivation tree restricted to a single node $\fact^\fact$ witnesses that $D^\fact, \Sigma \cup \Sigma^\fact \models \beta^\fact$
		\item Depth $k \geq 1$: we assume the result to be true for any derivation tree of depth up to $k-1$. Let us consider a child $(\gamma,r_\gamma,h_\gamma)$ (or $\fact$) of $(\beta,r_\beta,h_\beta)$ in $t$ that has $\fact$ as a descendant. By induction assumption, there exists a derivation tree of $\gamma^\fact$ w.r.t. $\ontology$ and $\database$. Let $p(t_1,\ldots,t_n)$ be an antecedent of $\gamma$ by $h_\beta$, and let $r_\beta^\alpha$ be the adornment of $r_\beta$ replacing $p(t_1,\ldots,t_n)$ by $p^\fact(t_1,\ldots,t_n)$. The structure obtained from $t$ by modifying the root to $(\beta,r_\beta^\fact,h_\beta)$, leaving descendants unchanged, except for the subtree rooted in $\gamma$ that is replaced by the derivation tree of $\gamma^\fact$ is a derivation tree for $\beta^\fact$ w.r.t. $D^\fact$ and $\Sigma \cup \Sigma^\fact$, hence showing that $D^\fact, \Sigma \cup \Sigma^\fact \models \beta^\fact$. 
	\end{itemize}
	
	We now show that any adornment-usable fact is usable, thanks to the following two observations:
	\begin{itemize}
		\item removing adornments from facts and rules in a derivation tree of $\beta^\fact$ w.r.t. $D^\fact$ and $\Sigma \cup \Sigma^\fact$ results in a derivation tree of $\beta$ w.r.t. $D$ and~$\Sigma$;
		\item any derivation tree of $\beta^\fact$ contains $\fact^\fact$ as a leaf: indeed, a node can have an adorned atom only if one of its child has an adorned atom, or if it is a leaf and is adorned. As $\fact^\fact$ is the only adorned atom of $D^\fact$, this concludes the proof.\qedhere
	\end{itemize}
\end{proof}


\section{Proofs of Table \ref{fig:prop-def} Results, $\atann$, $\mtann$, $\mdtann$ and $\rmdtann$ Cases}\label{app:prooftable}
We prove here the positive results in the $\atann$, $\mtann$, $\mdtann$ and $\rmdtann$ columns in Table \ref{fig:prop-def}. Counter-examples are given in Section~\ref{sec:analysis} for the properties not satisfied by some of these provenance semantics. 
We go over all the properties and analyze each with respect to $\atann$, $\mtann$, $\mdtann$ and $\rmdtann$. 

\subsection{Algebra Consistency}
\propAlgebraConsistency

\begin{proposition}\label{prop:ac_at}
	$\atann$, $\mdtann$, $\rmdtann$ and $\mtann$ satisfy the Algebra Consistency Property.
\end{proposition}
\begin{proof}
	It is shown by~\cite{DBLP:conf/pods/GreenT17} that $\atann$ satisfies the Algebra Consistency Property.
	Note that all possible derivation trees of $H(\vec a)$ \wrt $\ontology$ and $\database$ are of depth one.  Hence they are of minimal depth,  hereditary minimal depth, and are non-recursive, which completes the proof for $\mdtann$, $\rmdtann$, and $\mtann$, respectively.
\end{proof}

\subsection{Boolean Compatibility}

\begin{proposition}\label{prop:bc_at}
	$\atann$ satisfies the Boolean Compatibility Property.
\end{proposition}
\begin{proof}
	It is shown in~\cite{DBLP:conf/pods/GreenT17} that $\atann$ satisfies the Boolean Compatibility Property.
\end{proof}

\begin{proposition}
	$\mtann$ satisfies the Boolean Compatibility Property.
\end{proposition}
\begin{proof}
	As $\mi{PosBool}(\semiringVars)$ is absorptive, $\mtann$ and $\atann$ coincide on $\mi{PosBool}(\semiringVars)$. Hence $\mtann$ satisfies the Boolean Compatibility Property.
\end{proof}

\subsection{Commutation with Homomorphisms}
If there is a semiring homomorphism $h$ from $\semiringshort_1$ to $\semiringshort_2$, then $h(\prov(\ontology, \database,\semiringshort_1,\annot,\fact))=\prov(\ontology, \database,\semiringshort_2, h\circ\annot, \fact)$.

\begin{proposition}\label{prop:comh_md}
	$  \mtann, \mdtann$ and $\rmdtann$ satisfy the Commutation with Homomorphisms Property.
\end{proposition}
\begin{proof}
	By definition, we have  
	\[\mtann(\ontology, \database,\semiringshort_1,\annot,\fact) = \sum_{ t\in  \drvtrsig}   \ann(t)
	\]
	By definition of $ \ann(t)$ we have
	\[\mtann(\ontology, \database,\semiringshort_1,\annot,\fact) = \sum_{ t\in  \drvtrsig}   \prod_{v\text{ is a leaf of $t$}} \annot(v)
	\]	
	Since $h$ is a homomorphism and since the sum is finite, we have 
	\[h(\mtann(\ontology, \database,\semiringshort_1,\annot,\fact)) = \sum_{ t\in  \drvtrsig}   \prod_{v\text{ is a leaf of $t$}} h(\annot(v))
	\]	
	which is, in turn, equal to $\prov(\ontology, \database,\semiringshort_2, h\circ\annot, \fact)$.
	In a similar way one can proof the same claim also for  $\mdtann, \rmdtann$. 
\end{proof}

\subsection{Commutation with $\omega$-Continuous Homomorphisms}

\begin{proposition}\label{prop:comconth_at}
	$\atann,  \mtann, \mdtann$ and $\rmdtann$ satisfy the Commutation with $\omega$-Continuous Homomorphisms.
\end{proposition}
\begin{proof}
	The case of $  \mtann, \mdtann$ and $\rmdtann$ is a direct consequence of Proposition~\ref{prop:comh_md}. 
	The case of $\atann$ has been proved in \cite{DBLP:conf/pods/GreenT17}.
\end{proof}

\subsection{Joint and Alternative Use}
\sloppy{
	$\prov$ satisfies the \emph{Joint and Alternative Use Property} if for all tuples of facts $(\fact^1_1,\cdots, \fact^1_{n_1})$, $\dots$, $(\fact^m_1,\cdots, \fact^m_{n_m})$, $\prov(\ontology', \database, \semiringshort, \annot,\goal)= \Sigma_{i=1}^m \Pi_{j=1}^{n_i}
	\prov(\ontology, \database, \semiringshort, \annot,\fact^i_j)$ where $\ontology'=\ontology\cup\{\bigwedge_{j=1}^{n_i}\fact^i_j\rightarrow\goal| 1\le i \le m \}$ and $\goal$ is a nullary predicate such that $\goal\notin\schema(\ontology)\cup\schema(\database)$.
}

\begin{proposition}\label{prop:janda_at}
	$\atann$ satisfies the Joint and Alternative Use Property.
\end{proposition}
\begin{proof}
	For $\alpha_j^i$ let us denote the set of its derivation trees \wrt $\ontology'$ and $\database$ by  $T(j,i)$.
	The derivation trees for $\goal$ \wrt $\ontology'$ and $\database$ are exactly those of the following form:
	
	
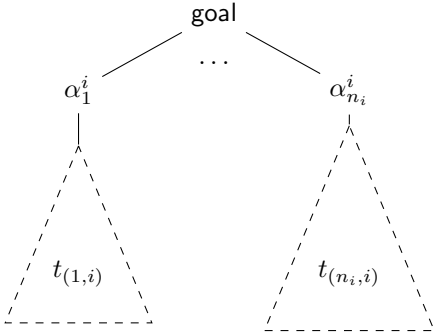
\begin{tikzpicture}
		[level distance=1.0cm,
		level 1/.style={sibling distance=1.8cm},
		level 2/.style={sibling distance=1.8cm}]
		\node {$\goal$}
		child {node {$\alpha_1^i$}
			child{node[itria] {$\small{t_{(1,i)}}$}}						
		}
		child{ edge from parent[draw=none] node[draw=none] (ellipsis) {$\ldots$} }
		child {node {$\alpha_{n_i}^i$}	
			child{node[itria] {$\tiny{t_{(n_i,i)}}$}}	
		};
	\end{tikzpicture}

	where $1 \le i \le m$, and $t_{(1,i)} \in T(1,i), \cdots, t_{(n_i,i)} \in T(n_i,i)$.
	Therefore, by definition we have: 
	$$\atann(\ontology', \database, \semiringshort, \annot,\goal) = \Sigma_{i=1}^{m} \Pi_{j=1}^{n_i}
	\sum_{ t\in  T(j,i)}   \ann(t). $$
	
	\noindent
	In addition, by definition of $\atann$ it holds that 
	$$\atann(\ontology', \database, \semiringshort, \annot,\fact^i_j) = 
	\sum_{ t\in  T(j,i)}   \ann(t) $$ for every $i$ and $j$.
	By definition of $\ontology'$, we can conclude that 
	$$\atann(\ontology, \database, \semiringshort, \annot,\fact^i_j) = 
	\sum_{ t\in  T(j,i)}   \ann(t) $$ for every $i$ and $j$.
	This concludes the proof.  
\end{proof}

\begin{proposition}\label{prop:janda_mt}
	$\mtann$ satisfies the Joint and Alternative Use Property.
\end{proposition}
\begin{proof}
	We use similar notation as those used in the proof of Proposition~\ref{prop:janda_at}.
	Let $T'(j,i) \subseteq T(j,i)$ be the set of non-recursive trees in $T(j,i)$ (those in which a fact is not a descendant of itself).
	Note that if each $t_{(j,i)}$ is non-recursive then so is the tree that is depicted in the proof of Proposition~\ref{prop:janda_at}, and therefore we have
	$$\mtann(\ontology', \database, \semiringshort, \annot,\goal) = \Sigma_{i=1}^{m} \Pi_{j=1}^{n_i}
	\sum_{ t\in  T'(j,i)}   \ann(t). $$
	
	\noindent
	In addition, by definition of $\mtann$ it holds that 
	$$\mtann(\ontology', \database, \semiringshort, \annot,\fact^i_j) = 
	\sum_{ t\in  T'(j,i)}   \ann(t) $$ for every $i$ and $j$.
	By definition of $\ontology'$, we can conclude that 
	$$\mtann(\ontology, \database, \semiringshort, \annot,\fact^i_j) = 
	\sum_{ t\in  T'(j,i)}   \ann(t) $$ for every $i$ and $j$.
	This concludes the proof.  
\end{proof}

\subsection{Joint Use}
$\prov$ satisfies the \emph{Joint Use Property} if for all facts $\fact_1,\cdots, \fact_n$, $\prov(\ontology', \database, \semiringshort, \annot,\goal)= \Pi_{j=1}^n \prov(\ontology, \database, \semiringshort, \annot,\fact_j)$ where  $\ontology'=\ontology\cup\{\bigwedge_{j=1}^n\fact_j \rightarrow\goal\}$ and $\goal$ is a nullary predicate such that $\goal\notin\schema(\ontology)\cup\schema(\database)$.

The next propositions are a consequence of the following straightforward observation:
\begin{proposition}
	A 
	provenance semantics that satisfies the Joint and Alternative Use Property satisfies also the Joint Use Property.
\end{proposition}

\begin{proposition}
	$\atann$ satisfies the Joint Use Property.
\end{proposition}
\begin{proof}
	This is a straightforward consequence of Proposition~\ref{prop:janda_at}.
\end{proof}

\begin{proposition}
	$\mtann$ satisfies the Joint Use Property.
\end{proposition}
\begin{proof}
	This is a straightforward consequence of Proposition~\ref{prop:janda_mt}.
\end{proof}

\begin{proposition}
	$\rmdtann$ satisfies the Joint Use Property.
\end{proposition}
\begin{proof}
	A key observation that is based on the definition of (hereditary) minimal depth trees is that all minimal depth derivation trees of a fact $\alpha_i$ are of the same depth; we denote this depth by $d_i$. 
	Note that the derivation trees of $\goal$ \wrt $\ontology'$ and $\database$ are of the form 
	
	
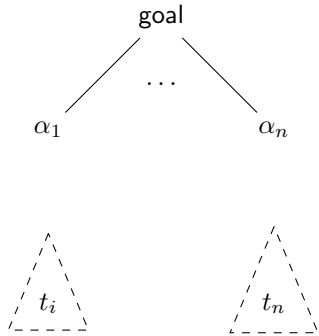
\begin{tikzpicture}
		\node {$\goal$}
		child {node {$\alpha_1$}
			child{edge from parent[draw=none] node[itria] {$\small{t_i}$}}						
		}
		child{ edge from parent[draw=none] node[draw=none] (ellipsis) {$\ldots$} }
		child {node {$\alpha_{n}$}	
			child{ edge from parent[draw=none] node[itria] {$\small{t_n}$}}	
		};
	\end{tikzpicture}
	
	where $t_i \in \drvtr{\alpha_i}$.
	Since we are interested in $\rmdtann(\ontology', \database, \semiringshort, \annot,\goal)$, we restrict the discussion only to those derivation trees of $\goal$ which are of hereditary minimal depth.
	Notice that these are exactly those trees for which each $t_i$ is a derivation tree of $\alpha_i$ of hereditary minimal depth. 
	Let us denote by $T_i$ the set of derivation trees for $\alpha_i$ of hereditary minimal depth. Then the derivation trees that we take into account in computing $\rmdtann(\ontology', \database, \semiringshort, \annot,\goal)$ are of those depicted above with $t_i\in T_i$.
	Thus, we get $\rmdtann(\ontology', \database, \semiringshort, \annot,\goal)= \Pi_{j=1}^n \rmdtann(\ontology, \database, \semiringshort, \annot,\fact_j)$, which completes the proof.
\end{proof}

\subsection{Alternative Use}
\sloppy{$\prov$ satisfies the \emph{Alternative Use Property} if for all facts $\fact_1,\cdots, \fact_m$, $\prov(\ontology', \database, \semiringshort, \annot,\goal)= \Sigma_{i=1}^m
	\prov(\ontology, \database, \semiringshort, \annot,\fact_i)$ where $\ontology'=\ontology\cup\{\fact_i\rightarrow\goal| 1\le i \le m \}$  and $\goal$ is a nullary predicate such that $\goal\notin\schema(\ontology)\cup\schema(\database)$. }

The next propositions are a consequence of the following straightforward observation:
\begin{proposition}
	A 
	provenance semantics that satisfies the Joint and Alternative Use Property satisfies also the Alternative Use Property.
\end{proposition}

\begin{proposition}\label{prop:al_at}
	$\atann$ satisfies the Alternative Use Property.
\end{proposition}
\begin{proof}
	This is a straightforward consequence of Proposition~\ref{prop:janda_at}.
\end{proof}

\begin{proposition}\label{prop:al_nrt}
	$\mtann$ satisfies the Alternative Use Property.
\end{proposition}
\begin{proof}
	This is a straightforward consequence of Proposition~\ref{prop:janda_mt}.
\end{proof}

\subsection{Self}
$\prov$ satisfies the \emph{Self  Property} if for every $\fact\in \database$, there exists $e\in K$ such that \(\prov(\ontology, \database, \semiringshort, \annot,\fact)= \annot(\fact)\plus e\).

\begin{proposition}
	$\atann , \mtann, \mdtann$ and $ \rmdtann$  satisfy the Self Property.
\end{proposition}
\begin{proof}
	Since $\alpha\in\database$, it holds that there exists a derivation tree of $\alpha$ \wrt $\ontology$ and $\database$ that consists of a single node $\alpha$. The claim follows directly.
\end{proof}

\subsection{Parsimony}
$\prov$ satisfies the \emph{Parsimony Property} if when $\fact$ belongs to $\database$ and does not occur in any rule head in the grounding $\ontology_\database$ of $\ontology$ \wrt $\database$, then $\prov(\ontology, \database, \semiringshort, \annot,\fact)= \annot(\fact)$.

\begin{proposition}\label{prop:tree-sat-parsimony}
	$\atann , \mtann, \mdtann$ and $ \rmdtann$  satisfy the Parsimony Property.
\end{proposition}
\begin{proof}
	It holds, by the definition of a derivation tree, that all derivation trees of $\alpha$ must consist of a single node  $\alpha$ (which is both root and leaf). 
	That is, there is a single derivation tree for  $\alpha$.
	Hence, the claim follows directly. 
\end{proof}

\subsection{Necessary Facts}
$\prov$ satisfies the \emph{Necessary Facts Property} if 
$\prov(\ontology, \database, \semiringshort, \annot,\fact) =\Pi_{\beta\in \mi{Nec}}\annot({\beta})\times e$ for some $e\in K$, where $\mi{Nec}$ is the set of facts necessary to $\ontology, \database \models \fact$.
\begin{proposition}\label{tree-sat-necfacts}
	$\atann , \mtann, \mdtann$ and $ \rmdtann$ satisfy the Necessary Facts Property.
\end{proposition}
\begin{proof}
	We start by showing that if $\beta$ is necessary to $\ontology, \database \models \fact$ then all derivation trees of $\fact$ has $\beta$ as a leaf. 
	Assume to the contrary that this is not the case, and let $t$ be a derivation tree whose leaves $\alpha_1,\cdots , \alpha_n $ are such that for every $i$ it holds that  $\alpha_i \ne \beta$. By definition, it holds that $\ontology, \{\alpha_1,\cdots , \alpha_n\} \models \fact$, which contradicts $\beta$ being a necessary fact.
	We can conclude that each Necessary Facts appear in every derivation tree of $\fact$. Therefore, by definition we have $\prov(\ontology, \database, \semiringshort, \annot,\fact) =\Pi_{\beta\in \mi{Nec}}\annot({\beta})\times e$ for some $e\in K$, for $\prov \in \{ \atann , \mtann, \mdtann, \rmdtann \}$.
\end{proof}

\subsection{Non-Usable Facts}
$\prov$ satisfies the \emph{Non-Usable Facts Property} if for every $\annot'$ that differs from $\annot$ only on facts that are not usable to $\ontology,\database\models\fact$, 
$\prov(\ontology, \database, \semiringshort, \annot,\fact)= \prov(\ontology, \database, \semiringshort, \annot',\fact)$.

\begin{proposition}
	$\atann , \mtann, \mdtann$ and $ \rmdtann$ satisfy the Non-Usable Facts Property.
\end{proposition}
\begin{proof}
	A fact is not usable if it does not occur in every derivation tree. To put it the other way around, for every derivation tree $t$ it holds that if $\beta$ occurs in $t$ it is usable. 
	By definition, for every derivation tree $t$ we have
	\[
	\ann(t) := \prod_{v\text{ is a leaf of $t$}} \annot(v)
	\]
	By the previous observation and the way $\annot'$ is defined, it holds that for every $t$,	
	\[
	\ann(t) := \prod_{v\text{ is a leaf of $t$}} \annot'(v)
	\]
	Therefore, we have the desired equivalence $\prov(\ontology, \database, \semiringshort, \annot,\fact)= \prov(\ontology, \database, \semiringshort, \annot',\fact)$ for $\prov \in \{ \atann , \mtann, \mdtann, \rmdtann \}$.
\end{proof}

\subsection{Insertion}
\sloppy{
	$\prov$ satisfies the \emph{Insertion Property} if for every $(\database',\semiringshort,\annot')$ such that $\database\cap\database'=\emptyset$, there exists $e\in K$ such that $\prov(\ontology, \database\cup\database',\semiringshort, \annot\cup\annot', \fact)=\prov(\ontology, \database,\semiringshort, \annot, \fact) \plus \prov(\ontology, \database',\semiringshort, \annot', \fact) \plus e$.
}

\begin{proposition}\label{prop:in_nr}
	$\atann$ satisfies the Insertion Property.
\end{proposition}
\begin{proof}
	Let us analyze the derivation trees of $\fact$ \wrt $(\database\cup\database',\semiringshort, \annot\cup\annot')$.
	We divide the derivation trees of $\fact$ \wrt $(\database\cup\database',\semiringshort, \annot\cup\annot')$ to three groups according to their leaves:
	\begin{itemize}
		\item[$(i)$] all leaves are elements in $D$,
		\item[$(ii)$] all leaves are elements in $D'$,
		\item[$(iii)$] at least one leaf is from $D$ and at least one is from $D'$.
	\end{itemize}
	We denote class $(i)$ by $T_D$, class $(ii)$ by $T_{D'}$, and class $(iii)$ by $T_{D,D'}$.
	We can change the order of summation to obtain
	\begin{equation}\label{eq1}
		\atann(\ontology,\database,\semiringshort,\annot,\fact)= \sum_{ t\in  \drvtrsig \cap T_D}   \ann(t)  +\sum_{ t\in  \drvtrsig\cap T_{D'}}   \ann(t)  +\sum_{ t\in  \drvtrsig\cap T_{D,D'}}   \ann(t)  
	\end{equation}	
	
	Notice that for $t\in T_D$ it holds that $\ann(t)$ \wrt $(\database\cup\database',\semiringshort, \annot\cup\annot')$ is the same as $\ann(t)$ \wrt $(\database,\semiringshort, \annot)$; for $t\in T_{D'}$ it holds that $\ann(t)$ \wrt $(\database\cup\database',\semiringshort, \annot\cup\annot')$ is the same as $\ann(t)$ \wrt $(\database',\semiringshort, \annot')$;
	Using the definition of $\atann$, we can replace the partial sums in equation~\ref{eq1} and obtain
	\begin{equation}\label{eq2}
		\atann(\ontology,\database\cup \database',\semiringshort,\annot\cup \annot',\fact)=  \atann(\ontology, \database,\semiringshort, \annot, \fact) \plus \atann(\ontology, \database',\semiringshort, \annot', \fact) \plus \sum_{ t\in  \drvtrsig\cap T_{D,D'}}   \ann(t)  
	\end{equation}
	which completes the proof.
\end{proof}

\begin{proposition}
	$\mtann$ satisfies the Insertion Property.
\end{proposition}
\begin{proof}
	We use here the same notations as used in the proof of Proposition~\ref{prop:in_nr}. 
	By similar arguments we have
	\begin{equation}
		\mtann(\ontology,\database,\semiringshort,\annot,\fact)= 
		\sum_{\substack{ t\in  \drvtrsig \cap T_D \\ \text{is non-recursive}} }  \ann(t)  +\sum_{\substack{ t\in  \drvtrsig\cap T_{D'}\\ \text{is non-recursive} }  }\ann(t)  
		+\sum_{\substack{ t\in  \drvtrsig\cap T_{D,D'}\\ \text{is non-recursive} }  } \ann(t)  
	\end{equation}	
	We note that this holds because of the non-recursiveness of trees that are taken into account and since $D\cap D' = \emptyset$.
	The rest of the proof is similar to that of Proposition~\ref{prop:in_nr}.
\end{proof}

\subsection{Deletion}
$\prov$ satisfies the \emph{Deletion Property} if for every provenance semiring $\mi{Prov}(\semiringVars)$ and $\database'\subseteq\database$, if $\annot'$ is the restriction of $\annot_X$ to $\database'$ and $\Delta=\database\setminus\database'$,  
then 
$\prov(\ontology, \database',\mi{Prov}(\semiringVars), \annot', \fact)$ is equal to the partial evaluation of $\prov(\ontology, \database, \mi{Prov}(\semiringVars), \annot_X, \fact)$ obtained by setting the annotations of facts in $\Delta$ to $\zero$: $\prov(\ontology, \database, \mi{Prov}(\semiringVars), \annot_X, \fact)[\{\annot_X(x) = \zero\}_{x \in \Delta}]$.

\begin{proposition}\label{dp_at}
	$\atann$ satisfies the Deletion Property.
\end{proposition}
\begin{proof}
	Let us denote the set of derivation trees of $\alpha$ \wrt $\ontology$ and $\database$ whose has at least one leaf from $\Delta$ by $T_{\Delta}$, and all other derivation trees of $\alpha$ (\ie those which have all their leaves in $D'$) by $T$.
	Then, by definition, we have 
	\begin{equation}
		\atann(\ontology,\database,\mi{Prov}(\semiringVars),\annot_X,\fact)= 
		\sum_{ t\in  \drvtrsig \cap T_{\Delta}}   \ann(t)  +\sum_{ t\in  \drvtrsig\cap T}   \ann(t)
	\end{equation}
	Therefore, the partial evaluation would result in 
	$
	\sum_{ t\in  \drvtrsig\cap T}   \ann(t)
	$
	which is, by definition, equals to 
	$	\atann(\ontology,\database',\mi{Prov}(\semiringVars),\annot',\fact)$.
\end{proof}

\begin{proposition}
	$\mtann$ satisfies the Deletion Property.
\end{proposition}
\begin{proof}
	The proof is obtained similarly to that of~\ref{dp_at}.
\end{proof}


\section{Proofs of Table \ref{fig:prop-def} Results, $\provmodtot$ and $\provmodmon$ Cases}\label{app:proofTableMod}

We prove here the positive results in the $\provmodtot$ and $\provmodmon$ columns in Table \ref{fig:prop-def} and give counter-examples to show that $\provmodtot$ does not satisfy the Necessary Facts Property and the Any $\omega$-Continuous Semiring Property. Counter-examples are given in Section \ref{sec:analysis} for the other properties not satisfied by $\provmodtot$ or $\provmodmon$.  
We will make use of the lemmas shown in Appendix \ref{app:model-based-relations}.

\subsection{$\provmodtot$ Case}\label{app:proofs-for-provmodtot-properties}

\begin{proposition}
	$\provmodtot$ satisfies the Boolean Compatibility Property.
\end{proposition}
\begin{proof}
	Since $\mi{PosBool}(\semiringVars)$ is a commutative $\omega$-continuous $\kplus$-idempotent semiring, by Proposition \ref{prop:modelBasedSameAsTreeBasedForIdempotent}, $\provmodtot(\ontology, \database, \mi{PosBool}(\semiringVars), \annot,\fact)= \atann(\ontology, \database,\mi{PosBool}(\semiringVars), \annot,\fact)$. 
	It follows by Proposition \ref{prop:bc_at} that 
	$\provmodtot(\ontology, \database,\mi{PosBool}(\semiringVars), \annot_X, \fact)=\bigvee_{D'\subseteq \database, \ontology, D' \models \fact}\bigwedge_{\beta\in D'}\annot_X(\beta)$. 
\end{proof}

\begin{proposition}\label{prop:modtot-sat-self}
	$\provmodtot$ satisfies the Self Property.
\end{proposition}
\begin{proof}
	If $\fact\in \database$, for every model $(\inter,\modannot^I)$ of $(\database, \semiringshort, \annot)$, $\annot(\fact)\sqsubseteq\modannot^I(\fact)$ so it follows straightforwardly from the definition of $\provmodtot$ that $\annot(\fact)\sqsubseteq\provmodtot(\ontology, \database, \semiringshort, \annot,\fact)$.
\end{proof}

\begin{proposition}
	$\provmodtot$ satisfies the Parsimony Property. 
\end{proposition}
\begin{proof}
	Assume that $\fact\in\database$ and does not occur in any rule head in the grounding $\ontology_\database$ of $\ontology$ \wrt $\database$. 
	By Proposition \ref{prop:modtot-sat-self}, $\annot(\fact)\sqsubseteq \provmodtot(\ontology, \database, \semiringshort, \annot,\fact)$.  
	Moreover, by Lemma \ref{provmodtot-upperbound}, 
	the interpretation $(\inter_0,\modannot^{I_0})$ defined by $\inter_0=\{\beta\mid \ontology,\database\models \beta\}$ and $\modannot^{I_0}(\beta)=\atann(\ontology, \database, \semiringshort, \annot,\beta)$ for every $\beta\in \inter_0$
	is a model of $\ontology$ and $(\database,\semiringshort, \annot)$, and by Proposition \ref{prop:tree-sat-parsimony} $\atann(\ontology, \database, \semiringshort, \annot,\fact)=\annot(\fact)$ so $\modannot^{I_0}(\fact)=\annot(\fact)$. 
	It follows that $\provmodtot(\ontology, \database, \semiringshort, \annot,\fact)= \annot(\fact)$.
\end{proof}

\begin{lemma}\label{lem:provmodtot-prov-is-model}
	The interpretation $(\inter_0,\modannot^{I_0})$ defined by $\inter_0=\{\fact\mid \ontology,\database\models \fact\}$ and $\modannot^{I_0}(\fact)=\provmodtot(\ontology, \database, \semiringshort, \annot,\fact)$ for every $\fact\in \inter_0$ is a model of $\ontology$ and $(\database, \semiringshort, \annot)$. 
\end{lemma}
\begin{proof}
	Point (1) of the definition is easy to check: $\database\subseteq \inter_0$ and for every $\fact\in\database$, $ \annot(\fact)\sqsubseteq \modannot^{I_0}(\fact)=\provmodtot(\ontology, \database, \semiringshort, \annot,\fact)$ by Proposition \ref{prop:modtot-sat-self}. 
	
	For point (2), let $\phi(\vec{x},\vec{y} ) \rightarrow H(\vec{x})$ be a rule in $\ontology$ and $h$ be a homomorphism from $\phi(\vec{x},\vec{y})$ to $\inter_0$. 
	By construction of $\inter_0$, for every $\beta\in h(\phi(\vec{x},\vec{y}))$, $\ontology,\database\models \beta$.
	Hence $\ontology,\database\models h(\phi(\vec{x},\vec{y}))$ so $\ontology,\database\models h(H(\vec{x}))$. 
	It follows that $h(H(\vec{x}))\in \inter_0$. 
	
	Let $h(H(\vec{x}))=\alpha$, and let $\Smc$ be the set of all homomorphisms $h':\phi(\vec{x},\vec{y})\mapsto \inter_0$ such that $h'(\vec{x})=h(\vec{x})$. For each such $h'\in \Smc$, let $\gamma^{h'}_1,\dots,\gamma^{h'}_{k_{h'}}$ be the set of facts from $\inter_0$ such that $h'(\phi(\vec{x},\vec{y}))=\gamma^{h'}_1\wedge\dots\wedge \gamma^{h'}_{k_{h'}}$. 
	Let $(\inter,\modannot^I)$ be a model of $\ontology$ and $(\database, \semiringshort, \annot)$. 
	\begin{enumerate}
		\item By definition of $\provmodtot$, for every $h'$ and $i$, $\provmodtot(\ontology, \database, \semiringshort, \annot,\gamma^{h'}_i)\sqsubseteq \modannot^I(\gamma^{h'}_i)$.
		\item Since $(\inter,\modannot^I)$ is a model of $\ontology$ and $(\database, \semiringshort, \annot)$, by definition of $\inter_0$ and Lemma \ref{ann-model-entailment}, $\inter_0\subseteq \inter$. It follows in particular that all homomorphisms from $\Smc$ are also homomorphisms from $\phi(\vec{x},\vec{y})$ to $\inter$. Thus $\Sigma_{h'\in\Smc} \Pi_{i=1}^{k_{h'}} \modannot^{\inter}(\gamma^{h'}_i) \sqsubseteq \modannot^{\inter}(\alpha)$.
		\item (1) and (2) imply that $\Sigma_{h'\in\Smc} \Pi_{i=1}^{k_{h'}} \provmodtot(\ontology, \database, \semiringshort, \annot,\gamma^{h'}_i) \sqsubseteq \modannot^{\inter}(\alpha)$. 
	\end{enumerate}
	Hence  $\Sigma_{h'\in\Smc} \Pi_{i=1}^{k_{h'}} \provmodtot(\ontology, \database, \semiringshort, \annot,\gamma^{h'}_i) \sqsubseteq \provmodtot(\ontology, \database, \semiringshort, \annot,\alpha)$, \ie  $\Sigma_{h':\phi(\vec{x},\vec{y})\mapsto \inter_0, h'(\vec{x})=h(\vec{x})} \Pi_{\beta\in h'(\phi(\vec{x},\vec{y}))} \modannot^{\inter_0}(\beta) \sqsubseteq \modannot^{\inter_0}(h(H(\vec{x}))$.  
\end{proof}

\begin{proposition}
	$\provmodtot$ satisfies the Joint Use Property.
\end{proposition}
\begin{proof}
	Let $(\fact_1,\dots,\fact_n)$ be a tuple of facts and $\ontology'=\ontology\cup\{\bigwedge_{j=1}^n \fact_j\rightarrow\mn{goal}\}$. 
	If $\ontology,\database\not\models\fact_j$ for some $\fact_j$, then $\ontology,\database\not\models\mn{goal}$, and by Lemma \ref{ann-model-null-prov} $\provmodtot(\ontology, \database, \semiringshort, \annot,\fact_j)=0$ and $\provmodtot(\ontology', \database, \semiringshort, \annot,\mn{goal})=0$. We next assume that $\ontology,\database\models\fact_j$ for every $\fact_j$. 
	
	Let $(\inter,\modannot^I)$ be a model of $\ontology'$ and $(\database, \semiringshort, \annot)$. 
	Since $(\inter,\modannot^I)$ is also a model of $\ontology$, then by Lemma \ref{ann-model-entailment}, $\{\fact_1,\dots,\fact_n\}\subseteq \inter$, and 
	by definition of $\provmodtot$, for every $\fact_j$, it holds that $\provmodtot(\ontology, \database, \semiringshort, \annot,\fact_j)\sqsubseteq \modannot^I(\fact_j)$. 
	Moreover, since there is a homomorphism from the body of $\bigwedge_{j=1}^n \fact_j\rightarrow\mn{goal}$ to $\inter$, then $\mn{goal}\in\inter$, and $\Pi_{j=1}^n \modannot^I(\fact_j)\sqsubseteq  \modannot^I(\mn{goal})$. 
	Hence $\Pi_{j=1}^n \provmodtot(\ontology, \database, \semiringshort, \annot,\fact_j)\sqsubseteq  \modannot^I(\mn{goal})$. 
	
	Let $(\inter_0,\modannot^{I_0})$ be the model of $\ontology$ and $(\database,\semiringshort,\annot)$ defined in Lemma \ref{lem:provmodtot-prov-is-model} and $(\inter_1,\modannot^{I_1})$ be such that $\inter_1=\inter_0\cup\{\mn{goal}\}$ and $\modannot^{I_1}(\mn{goal})=\Pi_{j=1}^n \provmodtot(\ontology, \database, \semiringshort, \annot,\fact_j)$. 
	It is easy to check that $(\inter_1,\modannot^{I_1})$ is a model of $\ontology'$ and $(\database, \semiringshort, \annot)$ as it satisfies point (2) of the definition for rule $\bigwedge_{j=1}^n \fact_j\rightarrow\mn{goal}$ by construction. 
	It follows that $\provmodtot(\ontology', \database, \semiringshort, \annot,\mn{goal})=\Pi_{j=1}^n\provmodtot(\ontology, \database, \semiringshort, \annot,\fact_j)$.
\end{proof}

\begin{proposition}
	$\provmodtot$ satisfies the Non-Usable Facts Property. 
\end{proposition}
\begin{proof}
	By Lemma \ref{lem:provmodtot-prov-is-model}, 
	the interpretation $(\inter_0,\modannot^{I_0})$ defined by $\inter_0=\{\fact\mid \ontology,\database\models \fact\}$ and $\modannot^{I_0}(\fact)=\provmodtot(\ontology, \database, \semiringshort, \annot,\fact)$ for every $\fact\in \inter_0$ is a model of $\ontology$ and $(\database, \semiringshort, \annot)$. 
	We show that for every $\fact$ such that $\ontology,\database\models\fact$, $\modannot^{I_0}(\fact)$ does not depend on non-usable facts, in the sense that none of the constraints of the form `$X \sqsubseteq \modannot^{I_0}(\fact)$' that $\modannot^{I_0}(\fact)$ has to fulfill according to the definition of annotated models involves non-usable facts. 
	This will show that $\provmodtot(\ontology, \database, \semiringshort, \annot,\fact)$ is not impacted by the annotations of non-usable facts. 
	
	Let $\fact$ be a fact such that $\ontology,\database\models\fact$.
	Assume for a contradiction that there exists a constraint $\Cmc$ on $\modannot^{I_0}(\fact)$ such that $\Cmc$ involves some non-usable fact. 
	\begin{enumerate}
		\item $\Cmc$ cannot be of the form `$\annot(\fact)\sqsubseteq \modannot^{I_0}(\fact)$', because it would mean that $\fact\in\database$ and hence is usable to $\ontology,\database\models\fact$, so that $\Cmc$ does not involved any non-usable fact. 
		\item Hence $\Cmc$ is of the form `$\sum_{h':\phi(\vec{x},\vec{y})\mapsto {I_0}, h'(\vec{x})=h(\vec{x})} \prod_{\beta\in h'(\phi(\vec{x},\vec{y}))} \modannot^{I_0}(\beta) \sqsubseteq \modannot^{I_0}(\fact)$' with $\datrule:=\phi(\vec{x},\vec{y} ) \rightarrow H(\vec{x})$ a rule in $\ontology$ and $h$ a homomorphism from $\phi(\vec{x},\vec{y})$ to $\inter_0$. 
		
		\item By assumption, there exists $h':\phi(\vec{x},\vec{y})\mapsto {I_0}, \ h'(\vec{x})=h(\vec{x})$ and $\beta_0\in h'(\phi(\vec{x},\vec{y}))$ such that $\modannot^{I_0}(\beta_0)$ depends on non-usable facts. 
		
		\item We can choose such $\beta_0\neq\fact$. Otherwise, if the only $\beta\in h'(\phi(\vec{x},\vec{y}))$ such that $\modannot^{I_0}(\beta)$ depends on non-usable facts is $\fact$ itself, the constraint $\Cmc$ does not add any dependance on non-usable facts to $\modannot^{I_0}(\fact)$. 
		
		\item For every $\beta\in h'(\phi(\vec{x},\vec{y}))$, since $\beta\in\inter_0$, then $\ontology,\database\models\beta$ by construction of $(\inter_0,\modannot^{I_0})$. 
		Hence there exists a derivation tree $t_0$ of $\fact$ with root $(\fact,\datrule,h')$ having children of the form $(\beta,\rho,g)$ with $\beta\in h'(\phi(\vec{x},\vec{y}))$ (including $\beta_0$), $\rho\in\ontology$, and $g$ homomorphism. 
		
		\item  By (3), there is a constraint $\Cmc_0$ on $\modannot^{I_0}(\beta_0)$ involving non-usable facts. 
		
		\item $\Cmc_0$ cannot be of the form `$\annot(\beta_0)\sqsubseteq \modannot^{I_0}(\beta_0)$' because in this case $\beta_0$ would be in $\database$ and hence usable as a leaf of $t_0\in\drvtr{\fact}$ by (5). 
		
		\item Thus $\Cmc_0$ is of the form `$\sum_{h_0':\phi_0(\vec{x},\vec{y})\mapsto {I_0}, h_0'(\vec{x})=h_0(\vec{x})} \prod_{\beta\in h_0'(\phi(\vec{x},\vec{y}))} \modannot^{I_0}(\beta) \sqsubseteq \modannot^{I_0}(\beta_0)$' with $\datrule_0:=\phi_0(\vec{x},\vec{y} ) \rightarrow H_0(\vec{x})$ a rule in $\ontology$ and $h_0$ a homomorphism from $\phi_0(\vec{x},\vec{y})$ to $\inter_0$.

		\item It follows that there exists $h_0':\phi_0(\vec{x},\vec{y})\mapsto {I_0}, \ h'_0(\vec{x})=h_0(\vec{x})$ and $\beta_1\in h_0'(\phi_0(\vec{x},\vec{y}))$ such that $\modannot^{I_0}(\beta_1)$ depends on non-usable facts. 
		
		\item We can choose such $\beta_1$ such that $\beta_1\neq\beta_0$ and $\beta_1\neq\fact$. Indeed, if the only $\beta\in h_0'(\phi_0(\vec{x},\vec{y}))$ such that $\modannot^{I_0}(\beta)$ depends on non-usable facts is equal to $\beta_0$ or $\fact$, $\Cmc_0$ does not add any dependance on non-usable facts to $\modannot^{I_0}(\beta_0)$.
		
		\item For every $\beta\in h_0'(\phi_0(\vec{x},\vec{y}))$, since $\beta\in\inter_0$, then $\ontology,\database\models\beta$ by construction of $(\inter_0,\modannot^{I_0})$. 
		Hence there exists a derivation tree $t_1$ of $\fact$ with root $(\fact,\datrule,h')$ having children of the form $(\beta,\rho,g)$ with $\beta\in h'(\phi(\vec{x},\vec{y}))$, among which $(\beta_0,\datrule_0,h'_0)$ has children of the form $(\beta,\rho,g)$ with $\beta\in h_0'(\phi_0(\vec{x},\vec{y}))$ (including $\beta_1$). 
		
		\item  By (9), there is a constraint $\Cmc_1$ on $\modannot^{I_0}(\beta_1)$ involving non-usable facts. 
		
		\item By repeating this process, we can build an infinite sequence of distinct facts $\fact,\beta_0,\beta_1,\beta_2,\dots$ that are all in $\inter_0$. This is a contradiction because $\inter_0$ is finite.
	\end{enumerate}
	We conclude that there is no constraint $\Cmc$ on $\modannot^{I_0}(\fact)$ such that $\Cmc$ involves some non-usable fact. 
	
	It follows that if $\annot'$ differs from $\annot$ only on facts that are not usable to $\ontology,\database\models\fact$, 
	$\provmodtot(\ontology, \database, \semiringshort, \annot,\fact)= \provmodtot(\ontology, \database, \semiringshort, \annot',\fact)$. 
\end{proof}

\begin{proposition}
	$\provmodtot$ satisfies the Deletion Property.
\end{proposition}
\begin{proof}
	Let $\mi{Prov}(\semiringVars)$ be a provenance semiring, $\database'\subseteq\database$, $\annot_X'$ be the restriction of $\annot_X$ to $\database'$ and $\Delta=\database\setminus\database'$. 
	
	If $\mi{Prov}(\semiringVars)$ is $\plus\,$-idempotent, by Proposition \ref{prop:modelBasedSameAsTreeBasedForIdempotent}, 
	$\provmodtot(\ontology, \database, \mi{Prov}(\semiringVars), \annot_X,\fact)=\atann(\ontology, \database, \mi{Prov}(\semiringVars), \annot_X,\fact)$ and $\provmodtot(\ontology, \database, \mi{Prov}(\semiringVars), \annot'_X,\fact)=\atann(\ontology, \database, \mi{Prov}(\semiringVars), \annot'_X,\fact)$ so by Proposition \ref{dp_at}, $\provmodtot(\ontology, \database',\mi{Prov}(\semiringVars), \annot_X', \fact)$ is equal to the partial evaluation of $\provmodtot(\ontology, \database, \mi{Prov}(\semiringVars), \annot_X, \fact)$ obtained by setting the annotations of facts in $\Delta$ to $\zero$.
	
	If $\mi{Prov}(\semiringVars)$ is not $\plus\,$-idempotent, then $\provmodtot(\ontology, \database, \mi{Prov}(\semiringVars), \annot_X, \fact)$ is a sum of monomials $m$ such that there exists $t\in \drvtr{\fact}$ such that $m=\ann(t)$. Indeed, 
	by Lemma~\ref{provmodtot-upperbound}, $\atann(\ontology, \database,  \mi{Prov}(\semiringVars), \annot_\semiringVars,\fact)=\provmodtot(\ontology, \database,  \mi{Prov}(\semiringVars), \annot_\semiringVars,\fact)+S$ for some $S\in \mi{Prov}(\semiringVars)$, and by Lemma \ref{provmodtot-lowerbound}, for every monomial $m$ in $S$, $m\sqsubseteq \provmodtot(\ontology, \database,  \mi{Prov}(\semiringVars), \annot_\semiringVars,\fact)$ because $m$ occurs in $\atann(\ontology, \database,  \mi{Prov}(\semiringVars), \annot_\semiringVars,\fact)$ so there exists $t\in \drvtr{\fact}$ such that $m=\ann(t)$. 
	It follows that $\provmodtot(\ontology, \database, \mi{Prov}(\semiringVars), \annot_X, \fact)=\provmodtot(\ontology, \database', \mi{Prov}(\semiringVars), \annot'_X, \fact)+ e$ where $e$ is a sum of products of the form $\Pi_{\beta\text{ is a leaf of $t$}} \annot_\semiringVars(\beta)$ for some $t\in\drvtrsig\setminus T^\ontology_{\database'}(\fact)$. 
	Moreover, by definition of $\Delta$ and the sets of derivation trees $\drvtrsig$ and $T^\ontology_{\database'}(\fact)$, it holds that $\drvtrsig\setminus T^\ontology_{\database'}(\fact)= \{t\in \drvtrsig \mid \exists \beta\text{ leaf of }t, \beta\in \Delta\} $. 
	Hence $\provmodtot(\ontology, \database',\mi{Prov}(\semiringVars), \annot_X', \fact)$ is equal to the partial evaluation of $\provmodtot(\ontology, \database, \mi{Prov}(\semiringVars), \annot_X, \fact)$ obtained by setting the annotations of facts in $\Delta$ to $\zero$, which makes $e$ evaluate to $0$.
\end{proof}

The following example shows that $\provmodtot$ does not satisfy the Necessary Facts Property.
\begin{example}\label{app:example-semiring-necfacts}
	Let $\semiringshort=\semiring$ be defined as follows. We show in Lemma \ref{app:lemma-example-semiring-necfacts} that $\semiringshort$ is a commutative $\omega$-continuous semiring such that the greatest lower bound of every pair of elements is well-defined.
	\begin{itemize}
		\item $K=\{0,1,\infty, a, b, c, d, e,f\}$;
		\item $\kzero=0$, $\kone=1$;
		\item $\kplus$ is defined by
		\begin{itemize}
			\item for every $x\in K$, $0+x=x+0=x$;
			\item $b+1=1+b=a$;
			\item $c+1=1+c=a$;
			\item in every other cases, $x+y=\infty$.
		\end{itemize}
		\item $\times$ is defined by
		\begin{itemize}
			\item for every $x\in K$, $1\times x=x\times 1=x$;
			\item for every $x\in K$, $0\times x=x\times 0=0$;
			\item $d\times e=e\times d=b$;
			\item $d\times f=f\times d=c$;
			\item in every other cases, $x\times y=\infty$.
		\end{itemize}
	\end{itemize}
	Let $\ontology=\{A(x)\wedge B(x)\rightarrow \mn{goal}, A(x)\wedge C(x)\rightarrow \mn{goal}\}$ and $\database=\{A(a), B(a),C(a)\}$ with $\annot(A(a))=d$, $\annot(B(a))=e$, $\annot(C(a))=f$. $A(a)$ is the only necessary fact. We can show that $\provmodtot(\ontology,\database,\mathbb{N}^\infty\llbracket\semiringVars\rrbracket,\annot_\semiringVars,\mn{goal})=a$. 
	Indeed, the models $(\inter,\modannot^I)$ of $\ontology$ and $(\database,\semiringshort,\annot)$ are such that $d\times e=b\sqsubseteq \modannot^I(\mn{goal})$ and $d\times f=c\sqsubseteq \modannot^I(\mn{goal})$, so that the possible values for $\modannot^I(\mn{goal})$ are $a=b+1=c+1$, and $\infty=b+c$, and $a\sqsubseteq \infty$. 
	Hence $\provmodtot(\ontology,\database,\mathbb{N}^\infty\llbracket\semiringVars\rrbracket,\annot_\semiringVars,\mn{goal})\neq d\times x$ for every $x\in K$ and $\provmodtot$ does not satisfy the Necessary Facts Property.
\end{example}

\begin{lemma}\label{app:lemma-example-semiring-necfacts}
	$\semiringshort=\semiring$ defined in Example \ref{app:example-semiring-necfacts} is a commutative $\omega$-continuous semiring such that the greatest lower bound of every pair of elements is well-defined.
\end{lemma}
\begin{proof}
	First, $\semiringshort$ is a commutative semiring:
	\begin{itemize}
		\item $\kplus$ is associative: 
		Let $x,y,z\in K$ and consider $x+(y+z)$. 
		If $x=0$, $x+(y+z)=y+z = (x+y)+z$ no matter the value of $y$ and $z$. 
		If $y=0$,  $x+(y+z)=x+z=(x+y)+z$ no matter the value of $x$ and $z$, and similarly if  $z=0$. 
		If $x$, $y$ and $z$ are all distinct from $0$, $(y+z)$ is equal to $a$ or $\infty$ and in all cases $x+(y+z)=\infty$. 
		The same holds for $(x+y)$ and $(x+y)+z$, so $x+(y+z)=\infty=(x+y)+z$.
		\item It is clear from the construction that $\kplus$ is commutative and has identity $\kzero$.
		\item $\ktimes$ is associative: Let $x,y,z\in K$ and consider $x\times(y\times z)$. 
		If $x=0$, $x\times(y\times z)=0=(x\times y)\times z$. 
		If $y=0$, $x\times(y\times z)=x\times 0 = 0 = 0\times z = (x\times y)\times z$, and similarly for $z=0$. 
		If $x=1$, $x\times(y\times z)=y\times z=(x\times y)\times z$. 
		If $y=1$, $x\times(y\times z)=x\times z=(x\times y)\times z$, and similarly for $z=1$. 
		If $x$, $y$, $z$ are all different from $0$ and $1$, $y\times z$ and $x\times y$ can both be equal to $b$, $c$ or $\infty$ and in all cases, 
		$x\times(y\times z)=\infty=(x\times y)\times z$. 
		\item  It is clear from the construction that $\ktimes$ is commutative and has identity $\kone$.
		\item $\ktimes$ distributes over $\kplus$:  Let $x,y,z\in K$ and consider $x\times(y+ z)$. 
		If $x=0$, $x\times(y+ z)=0=(x\times y)+(x\times z)$. 
		If $y=0$, $x\times(y+ z)=x\times z=(x\times y)+(x\times z)$ and similarly for $z=0$. 
		If $x=1$, $x\times(y+ z)=y+z=(x\times y)+(x\times z)$. 
		If $x,y,z$ are all different from $0$ and $x\neq 1$, then $y+z$ is equal to $a$ or $\infty$ and in both cases $x\times(y+ z)=\infty$. Moreover, $x\times y$ and $x\times z$ can be equal to $\infty, a, b, c, d, e$ or $f$ and in all cases $(x\times y)+(x\times z)=\infty$
		\item It is clear from the construction that $\kzero$ is annihilating for $\ktimes$.
	\end{itemize}
	Second, $\semiringshort$ is $\omega$-continuous. The $\sqsubseteq$ relation defined by $x\sqsubseteq y$ if and only if there exists $z$ such that $x+z=y$ is as follows:
	\begin{itemize}
		\item $x\sqsubseteq x$ for every $x\in K$;
		\item $0\sqsubseteq x$ for every $x\in K$;
		\item $x\sqsubseteq \infty$ for every $x\in K$;
		\item $b\sqsubseteq a$, $c\sqsubseteq a$, $1\sqsubseteq a$.
	\end{itemize}
	It is easy to check that $\sqsubseteq$ is a partial order and every $\omega$-chain $x_0\sqsubseteq x_1\sqsubseteq \dots$ has a least upper bound $\sup((x_i)_{i\in \mathbb{N}})$. 
	Moreover, for every $x\in K$, we show that $x\plus \sup((x_i)_{i\in \mathbb{N}})=\sup((x\plus x_i)_{i\in \mathbb{N}})$ and $x\times \sup((x_i)_{i\in \mathbb{N}})=\sup((x\times x_i)_{i\in \mathbb{N}})$.
	\begin{itemize}
		\item If $x=0$, $x\plus \sup((x_i)_{i\in \mathbb{N}})=\sup((x_i)_{i\in \mathbb{N}})=\sup((x\plus x_i)_{i\in \mathbb{N}})$ and $x\times \sup((x_i)_{i\in \mathbb{N}})=0=\sup((x\times x_i)_{i\in \mathbb{N}})$.
		\item If $x=1$, $x\times \sup((x_i)_{i\in \mathbb{N}})=\sup((x_i)_{i\in \mathbb{N}})=\sup((x\times x_i)_{i\in \mathbb{N}})$. Moreover, 
		\begin{itemize}
			\item if $x_i=0$ for every $i\in \mathbb{N}$, then $x\plus \sup((x_i)_{i\in \mathbb{N}})=1\plus 0=1=\sup((1)_{i\in \mathbb{N}})=\sup((x\plus x_i)_{i\in \mathbb{N}})$, 
			\item if there exists $x_{i_0}$ such that $x_{i}=b$ (resp. $c$) for every $i\geq i_0$, then (i) $x\plus x_{i}=a$ for every $i\geq i_0$ so $\sup((x\plus x_i)_{i\in \mathbb{N}})=a$ and (ii) $\sup((x_i)_{i\in \mathbb{N}})=b$ (resp. $c$) so $x\plus \sup((x_i)_{i\in \mathbb{N}})=a=\sup((x\plus x_i)_{i\in \mathbb{N}})$,
			\item if there exists $x_{i_0}$ such that $x_{i}\in\{1,\infty,a,d,e,f\}$ for every $i\geq i_0$, then (i) $x\plus x_{i}=\infty$ for every $i\geq i_0$ so $\sup((x\plus x_i)_{i\in \mathbb{N}})=\infty$ and (ii) $\sup((x_i)_{i\in \mathbb{N}})\in\{1,\infty,a,d,e,f\}$ so $x\plus \sup((x_i)_{i\in \mathbb{N}})=\infty=\sup((x\plus x_i)_{i\in \mathbb{N}})$.
		\end{itemize}
		\item If $x\neq 0$ and $x\neq 1$: consider first the case $x\times  \sup((x_i)_{i\in \mathbb{N}})$:
		\begin{itemize}
			\item if $x_i=0$ for every $i\in \mathbb{N}$, then $x\times  \sup((x_i)_{i\in \mathbb{N}})= x\times 0=0=\sup((x\times x_i)_{i\in \mathbb{N}})$;
			\item if there exists $x_{i_0}$ such that $x_i=1$ for every $i\geq i_0$, then $x\times  \sup((x_i)_{i\in \mathbb{N}})= x\times 1 =x=\sup((x\times x_i)_{i\in \mathbb{N}})$;
			\item if there exists $x_{i_0}$ such that $x_i=d$ for every $i\geq i_0$ and $x\neq e$, $x\neq f$, then $x\times  \sup((x_i)_{i\in \mathbb{N}})= x\times d =\infty=\sup((x\times x_i)_{i\in \mathbb{N}})$;
			\item if there exists $x_{i_0}$ such that $x_i=d$ for every $i\geq i_0$ and $x= e$, then $x\times  \sup((x_i)_{i\in \mathbb{N}})= e\times d =b=\sup((x\times x_i)_{i\in \mathbb{N}})$;
			\item if there exists $x_{i_0}$ such that $x_i=d$ for every $i\geq i_0$ and $x= f$, then $x\times  \sup((x_i)_{i\in \mathbb{N}})= f\times d =c=\sup((x\times x_i)_{i\in \mathbb{N}})$;
			\item if there exists $x_{i_0}$ such that $x_i=e$ for every $i\geq i_0$ and $x\neq d$, then $x\times  \sup((x_i)_{i\in \mathbb{N}})= x\times e =\infty=\sup((x\times x_i)_{i\in \mathbb{N}})$;
			\item if there exists $x_{i_0}$ such that $x_i=e$ for every $i\geq i_0$ and $x= d$, then $x\times  \sup((x_i)_{i\in \mathbb{N}})= d\times e =b=\sup((x\times x_i)_{i\in \mathbb{N}})$;
			\item if there exists $x_{i_0}$ such that $x_i=f$ for every $i\geq i_0$ and $x\neq d$, then $x\times  \sup((x_i)_{i\in \mathbb{N}})= x\times f =\infty=\sup((x\times x_i)_{i\in \mathbb{N}})$;
			\item if there exists $x_{i_0}$ such that $x_i=f$ for every $i\geq i_0$ and $x= d$, then $x\times  \sup((x_i)_{i\in \mathbb{N}})= d\times f =c=\sup((x\times x_i)_{i\in \mathbb{N}})$;
			\item if there exists $x_{i_1}$ different from $0,1,d,e,f$, then $\sup((x_i)_{i\in \mathbb{N}})$ is different from $0,1,d,e,f$ and $x\times  \sup((x_i)_{i\in \mathbb{N}})=\infty=\sup((x\times x_i)_{i\in \mathbb{N}})$.
		\end{itemize}
		\item If $x\neq 0$ and $x\neq 1$: consider now the case $x\plus  \sup((x_i)_{i\in \mathbb{N}})$:
		\begin{itemize}
			\item if $x_i=0$ for every $i\in \mathbb{N}$, then $x\plus  \sup((x_i)_{i\in \mathbb{N}})= x\plus 0=x=\sup((x\plus x_i)_{i\in \mathbb{N}})$;
			\item if there exists $x_{i_0}\neq 0$: 
			\begin{itemize}
				\item if there exists $x_{i_1}$ different from $1,b,c$, $\sup((x_i)_{i\in \mathbb{N}})$ can be equal to  $\infty,a,d,e,f$ and in all cases, $x\plus \sup((x_i)_{i\in \mathbb{N}})=\infty$ and $x\plus x_{i_1}=\infty$ so $\sup((x\plus x_i)_{i\in \mathbb{N}})=\infty$.
				\item else, all $x_i$ are either equal to $0$ or to $x_{i_1}\in\{1,b,c\}$ (since $1,b,c$ are not comparable they cannot occur in the same $\omega$-chains) and $\sup((x_i)_{i\in \mathbb{N}})=x_{i_1}$. 
				\begin{itemize}
					\item Assume $x_{i_1}=1$. Then if $x\notin\{ b,c\}$, $x\plus x_{i_1}=\infty$ so $x\plus \sup((x_i)_{i\in \mathbb{N}})=\infty =\sup((x\plus x_i)_{i\in \mathbb{N}})$. If $x=b$ or $x=c$, then $x\plus x_{i_1}=a$ so $x\plus \sup((x_i)_{i\in \mathbb{N}})=a =\sup((x\plus x_i)_{i\in \mathbb{N}})$.
					\item Assume $x_{i_1}=b$. Then if $x\neq 1$, $x\plus x_{i_1}=\infty$ so $x\plus \sup((x_i)_{i\in \mathbb{N}})=\infty =\sup((x\plus x_i)_{i\in \mathbb{N}})$. If $x=1$, then $x\plus x_{i_1}=a$ so $x\plus \sup((x_i)_{i\in \mathbb{N}})=a =\sup((x\plus x_i)_{i\in \mathbb{N}})$.
					\item The case $x_{i_1}=c$ is similar.
				\end{itemize}
			\end{itemize}
		\end{itemize}
	\end{itemize}
	Hence $\semiringshort$ is $\omega$ continuous.\\
	Finally, for every $x,y\in K$, the greatest lower bound of $x,y$ exists.
	\begin{itemize}
		\item If $x=y$, the greatest lower bound of $x,y$ is $x$.
		\item If $x=\infty$, the greatest lower bound of $x,y$ is $y$. 
		\item If $x=a$ and $y=b$, $y=c$ or $y=1$, the greatest lower bound of $x,y$ is $y$. 
		\item Otherwise, the greatest lower bound of $x,y$ is $0$. \qedhere
	\end{itemize}
\end{proof}

The following example shows that the greatest lower bound of a pair of elements is not guaranteed to exists, even for $\omega$-continuous semirings, so that $\provmodtot$ does not satisfy the Any $\omega$-Continuous Semiring Property. 

\begin{example}\label{app:example-model-based}
	Let $\semiringshort=\semiring$ be defined as follows. We show in Lemma \ref{app:lemma-example-semiring} that $\semiringshort$ is a commutative $\omega$-continuous semiring.
	\begin{itemize}
		\item $K=\{0,1,\infty, a, b, c, d, e\}$;
		\item $\kzero=0$, $\kone=1$;
		\item $\kplus$ is defined by
		\begin{itemize}
			\item for every $x\in K$, $0+x=x+0=x$;
			\item $c+d=d+c=a$;
			\item $d+e=e+d=a$;
			\item $c+e=e+c=b$;
			\item in every other cases, $x+y=\infty$.
		\end{itemize}
		\item $\times$ is defined by
		\begin{itemize}
			\item for every $x\in K$, $1\times x=x\times 1=x$;
			\item for every $x\in K$, $0\times x=x\times 0=0$;
			\item in every other cases, $x\times y=\infty$.
		\end{itemize}
	\end{itemize}
	
	According to the $\sqsubseteq$ relation defined by $x\sqsubseteq y$ if and only if there exists $z$ such that $x+z=y$, 
	$a$ and $b$ have two lower bounds $e$ and $c$ (since $c\sqsubseteq a$, $e\sqsubseteq a$, $c\sqsubseteq b$, $e\sqsubseteq b$) which are not comparable (since $c\not\sqsubseteq e$ and $e\not\sqsubseteq c$).
\end{example}

\begin{lemma}\label{app:lemma-example-semiring}
	The semiring $\semiringshort$ of Example \ref{app:example-model-based} is a commutative $\omega$-continuous semiring. 
\end{lemma}
\begin{proof}
	First, $\semiringshort$ is a commutative semiring:
	\begin{itemize}
		\item $\kplus$ is associative: 
		Let $x,y,z\in K$ and consider $x+(y+z)$. 
		If $x=0$, $x+(y+z)=y+z = (x+y)+z$ no matter the value of $y$ and $z$. 
		If $y=0$,  $x+(y+z)=x+z=(x+y)+z$ no matter the value of $x$ and $z$, and similarly if  $z=0$. 
		If $x$, $y$ and $z$ are all distinct from $0$, $(y+z)$ is equal to $a$, $b$ or $\infty$ and in all cases $x+(y+z)=\infty$. 
		The same holds for $(x+y)$ and $(x+y)+z$, so $x+(y+z)=\infty=(x+y)+z$.
		\item It is clear from the construction that $\kplus$ is commutative and has identity $\kzero$.
		\item $\ktimes$ is associative: Let $x,y,z\in K$ and consider $x\times(y\times z)$. 
		If $x=0$, $x\times(y\times z)=0=(x\times y)\times z$. 
		If $y=0$, $x\times(y\times z)=x\times 0 = 0 = 0\times z = (x\times y)\times z$, and similarly for $z=0$. 
		If $x=1$, $x\times(y\times z)=y\times z=(x\times y)\times z$. 
		If $y=1$, $x\times(y\times z)=x\times z=(x\times y)\times z$, and similarly for $z=1$. 
		If $x$, $y$, $z$ are all different from $0$ and $1$, $x\times(y\times z)=x\times\infty=\infty=\infty\times z=(x\times y)\times z$.
		\item  It is clear from the construction that $\ktimes$ is commutative and has identity $\kone$.
		\item $\ktimes$ distributes over $\kplus$:  Let $x,y,z\in K$ and consider $x\times(y+ z)$. 
		If $x=0$, $x\times(y+ z)=0=(x\times y)+(x\times z)$. 
		If $y=z=0$, $x\times(y+ z)=0=(x\times y)+(x\times z)$. Note that this is the only case where $y+z=0$. 
		If $x=1$, $x\times(y+ z)=y+z=(x\times y)+(x\times z)$. 
		If $y=0$ and $z=1$, $x\times(y+ z)=x\times 1=x=(x\times y)+(x\times z)$, and similarly in the case where $y=1$ and $z=0$. Note that these two cases cover the case $y+z=1$. 
		If $x$ is different from $0$ and $1$, $y+z\neq 1$, and $y+z\neq 0$, then $x\times(y+ z)=\infty=\infty + \infty=(x\times y)+(x\times z)$. 
		\item It is clear from the construction that $\kzero$ is annihilating for $\ktimes$.
	\end{itemize}
	Second, $\semiringshort$ is $\omega$-continuous. The $\sqsubseteq$ relation defined by $x\sqsubseteq y$ if and only if there exists $z$ such that $x+z=y$ is as follows:
	\begin{itemize}
		\item $x\sqsubseteq x$ for every $x\in K$;
		\item $0\sqsubseteq x$ for every $x\in K$;
		\item $c\sqsubseteq a$; $d\sqsubseteq a$; $e\sqsubseteq a$;
		\item $c\sqsubseteq b$; $e\sqsubseteq b$;
		\item $x\sqsubseteq \infty$ for every $x\in K$.
	\end{itemize}
	It is easy to check that $\sqsubseteq$ is a partial order and every $\omega$-chain $x_0\sqsubseteq x_1\sqsubseteq \dots$ has a least upper bound $\sup((x_i)_{i\in \mathbb{N}})$. 
	Moreover, for every $x\in K$, we show that $x\plus \sup((x_i)_{i\in \mathbb{N}})=\sup((x\plus x_i)_{i\in \mathbb{N}})$ and $x\times \sup((x_i)_{i\in \mathbb{N}})=\sup((x\times x_i)_{i\in \mathbb{N}})$.
	\begin{itemize}
		\item If $x=0$, $x\plus \sup((x_i)_{i\in \mathbb{N}})=\sup((x_i)_{i\in \mathbb{N}})=\sup((x\plus x_i)_{i\in \mathbb{N}})$ and $x\times \sup((x_i)_{i\in \mathbb{N}})=0=\sup((x\times x_i)_{i\in \mathbb{N}})$.
		\item If $x=1$, $x\times \sup((x_i)_{i\in \mathbb{N}})=\sup((x_i)_{i\in \mathbb{N}})=\sup((x\times x_i)_{i\in \mathbb{N}})$. Moreover, 
		\begin{itemize}
			\item if $x_i=0$ for every $i\in \mathbb{N}$, then $x\plus \sup((x_i)_{i\in \mathbb{N}})=1\plus 0=1=\sup((1)_{i\in \mathbb{N}})=\sup((x\plus x_i)_{i\in \mathbb{N}})$, 
			\item if there exists $x_{i_0}$ such that $x_{i_0}\neq 0$, then (i) $x\plus x_{i_0}=\infty$ so $\sup((x\plus x_i)_{i\in \mathbb{N}})=\infty$ and (ii) $\sup((x_i)_{i\in \mathbb{N}})\neq 0$ so $x\plus \sup((x_i)_{i\in \mathbb{N}})=\infty=\sup((x\plus x_i)_{i\in \mathbb{N}})$.
		\end{itemize}
		\item If $x\neq 0$ and $x\neq 1$: consider first the case $x\times  \sup((x_i)_{i\in \mathbb{N}})$:
		\begin{itemize}
			\item if $x_i=0$ for every $i\in \mathbb{N}$, then $x\times  \sup((x_i)_{i\in \mathbb{N}})= x\times 0=0=\sup((x\times x_i)_{i\in \mathbb{N}})$;
			\item if there exists $x_{i_0}\neq 0$ and $x_i=1$ for every $i\geq i_0$, then $x\times  \sup((x_i)_{i\in \mathbb{N}})= x\times 1 =x=\sup((x\times x_i)_{i\in \mathbb{N}})$;
			\item if there exists $x_{i_1}$ different from $0$ and $1$, then $\sup((x_i)_{i\in \mathbb{N}})$ is different from $0$ and from $1$ and $x\times  \sup((x_i)_{i\in \mathbb{N}})=\infty=\sup((x\times x_i)_{i\in \mathbb{N}})$.
		\end{itemize}
		\item If $x\neq 0$ and $x\neq 1$: consider now the case $x\plus  \sup((x_i)_{i\in \mathbb{N}})$:
		\begin{itemize}
			\item if $x_i=0$ for every $i\in \mathbb{N}$, then $x\plus  \sup((x_i)_{i\in \mathbb{N}})= x\plus 0=x=\sup((x\plus x_i)_{i\in \mathbb{N}})$;
			\item if there exists $x_{i_0}\neq 0$: 
			\begin{itemize}
				\item if there exists $x_{i_1}$ different from $c,d,e$ (\ie $x_{i_1}$ is equal to $1$, $a$, $b$, or $\infty$), $\sup((x_i)_{i\in \mathbb{N}})$ can be equal to  $1$, $a$, $b$, or $\infty$ and in all cases, $x\plus \sup((x_i)_{i\in \mathbb{N}})=\infty$ and $x\plus x_{i_1}=\infty$ so $\sup((x\plus x_i)_{i\in \mathbb{N}})=\infty$.
				\item else, all $x_i$ are either equal to $0$ or to $x_{i_1}\in\{c, d, e\}$ (since $c,d,e$ are not comparable they cannot occur in the same $\omega$-chains) and $\sup((x_i)_{i\in \mathbb{N}})=x_{i_1}$. 
				\begin{itemize}
					\item Assume $x_{i_1}=c$. Then if $x\notin\{ d,e\}$, $x\plus x_{i_1}=\infty$ so $x\plus \sup((x_i)_{i\in \mathbb{N}})=\infty =\sup((x\plus x_i)_{i\in \mathbb{N}})$. If $x=d$, then $x\plus x_{i_1}=a$ so $x\plus \sup((x_i)_{i\in \mathbb{N}})=a =\sup((x\plus x_i)_{i\in \mathbb{N}})$. Similarly if $x=e$, $x\plus \sup((x_i)_{i\in \mathbb{N}})= b =\sup((x\plus x_i)_{i\in \mathbb{N}})$.
					\item The cases $x_{i_1}=d$ and $x_{i_1}=e$ are similar.
				\end{itemize}
			\end{itemize}
		\end{itemize}
	\end{itemize}
	Hence $\semiringshort$ is $\omega$ continuous.
\end{proof}

\subsection{$\provmodmon$ Case}

\begin{proposition}
	$\provmodmon$ satisfies the Boolean Compatibility Property.
\end{proposition}
\begin{proof}
	Since $\mi{PosBool}(\semiringVars)$ is a commutative $\omega$-continuous $\kplus$-idempotent semiring, by Proposition \ref{prop:modelBasedSameAsTreeBasedForIdempotent}, $\provmodmon(\ontology, \database, \mi{PosBool}(\semiringVars), \annot,\fact)= \atann(\ontology, \database,\mi{PosBool}(\semiringVars), \annot,\fact)$. 
	It follows by Proposition \ref{prop:bc_at} that 
	$\provmodmon(\ontology, \database,\mi{PosBool}(\semiringVars), \annot_X, \fact)=\bigvee_{D'\subseteq \database, \ontology, D' \models \fact}\bigwedge_{\beta\in D'}\annot_X(\beta)$. 
\end{proof}

\begin{proposition}
	$\provmodmon$ satisfies the Self Property.
\end{proposition}
\begin{proof}
	If $\fact\in \database$, for every model $(\inter,\modannot^I)$ of $(\database, \semiringshort, \annot)$, $\annot(\fact)\in\modannot^I(\fact)$ so it follows straightforwardly from the definition of $\provmodmon$ that $\annot(\fact)\sqsubseteq\provmodmon(\ontology, \database, \semiringshort, \annot,\fact)$.
\end{proof}

\begin{proposition}
	$\provmodmon$ satisfies the Parsimony Property. 
\end{proposition}
\begin{proof}
	Assume that $\fact\in\database$ and does not occur in any rule head in the grounding $\ontology_\database$ of $\ontology$ \wrt $\database$. 
	By Lemma \ref{provmodmon-lemma}
	$\provmodmon(\ontology, \database, \semiringshort, \annot,\fact)= \sum_{\{\ann(t) \mid t\in \drvtr{\fact}\}} \ann(t)$, where $\ann(t)=\Pi_{\beta\text{ is a leaf of $t$}} \annot(\beta)$ and it follows from the assumptions on $\fact$ that $\drvtr{\fact}$ contains a single derivation tree which consists of a single root node. 
	Hence $\provmodmon(\ontology, \database, \semiringshort, \annot,\fact)= \annot(\fact)$.
\end{proof}

\begin{proposition}
	$\provmodmon$ satisfies the Necessary Facts Property.
\end{proposition}
\begin{proof}
	Let $\mi{Nec}$ be the set of facts necessary to $\ontology, \database \models \fact$. 
	By Lemma \ref{provmodmon-lemma}
	$\provmodmon(\ontology, \database, \semiringshort, \annot,\fact)= \sum_{\{\ann(t) \mid t\in \drvtr{\fact}\}} \ann(t)$, where $\ann(t)=\Pi_{\beta\text{ is a leaf of $t$}} \annot(\beta)$. 
	Hence, $\provmodmon(\ontology, \database, \semiringshort, \annot,\fact) =\Pi_{\beta\in \mi{Nec}}\annot({\beta})\times e$ for some $e\in K$.
\end{proof}

\begin{proposition}
	$\provmodmon$ satisfies the Non-Usable Facts Property. 
\end{proposition}
\begin{proof}
	Let $\annot'$ that differs from $\annot$ only on facts that are not usable to $\ontology,\database\models\fact$. 
	By Lemma \ref{provmodmon-lemma}
	$\provmodmon(\ontology, \database, \semiringshort, \annot,\fact)= \sum_{\{\ann(t) \mid t\in \drvtr{\fact}\}} \ann(t)$, where $\ann(t)=\Pi_{\beta\text{ is a leaf of $t$}} \annot(\beta)$. 
	Since $\annot$ and $\annot'$ coincide on all facts usable to $\ontology,\database\models\fact$ and only such facts occur in leaves of trees from $\drvtrsig$, it follows that $\Pi_{\beta\text{ is a leaf of $t$}} \annot(\beta)=\Pi_{\beta\text{ is a leaf of $t$}} \annot'(\beta)$. 
	Hence 
	$\provmodmon(\ontology, \database, \semiringshort, \annot,\fact)= \provmodmon(\ontology, \database, \semiringshort, \annot',\fact)$.
\end{proof}

\begin{proposition}
	$\provmodmon$ satisfies the Deletion Property.
\end{proposition}
\begin{proof}
	Let $\mi{Prov}(\semiringVars)$ be a provenance semiring, $\database'\subseteq\database$, $\annot_X'$ be the restriction of $\annot_X$ to $\database'$ and $\Delta=\database\setminus\database'$. 
	By Lemma \ref{provmodmon-lemma} $\provmodmon(\ontology, \database',\mi{Prov}(\semiringVars), \annot'_X,\fact)= \sum_{\{\ann(t) \mid t\in T^\ontology_{\database'}(\fact)\}}\ann'(t)$ with $\ann'(t)=\Pi_{\beta\text{ is a leaf of $t$}} \annot'_X(\beta)=\Pi_{\beta\text{ is a leaf of $t$}} \annot_X(\beta)$. 
	
	Let us denote the set of derivation trees of $\alpha$ \wrt $\ontology$ and $\database$ whose has at least one leaf from $\Delta$ by $T_{\Delta}$, and all other derivation trees of $\alpha$ (\ie those which have all their leaves in $D'$) by $T$. 
	Note that $T=T^\ontology_{\database'}(\fact)$. 
	Then, by Lemma \ref{provmodmon-lemma}, we have 
	\begin{align*}
		\provmodmon(\ontology, \database,\mi{Prov}(\semiringVars), \annot_X,\fact)= &\sum_{\{\ann(t) \mid t\in \drvtr{\fact}\}} \ann(t) \text{ with }\Pi_{\beta\text{ is a leaf of $t$}}\ann(t)= \annot_X(\beta)\\
		= &\sum_{\{\ann(t) \mid t\in T\}} \ann(t) + \sum_{\{\ann(t) \mid t\in T_{\Delta}\}} \ann(t) \\
		=&\provmodmon(\ontology, \database',\mi{Prov}(\semiringVars), \annot'_X,\fact) +\sum_{\{\ann(t) \mid t\in T_{\Delta}\}} \ann(t) 
	\end{align*}
	Hence 
	$\provmodmon(\ontology, \database',\mi{Prov}(\semiringVars), \annot_X', \fact)$ is equal to the partial evaluation of $\provmodmon(\ontology, \database, \mi{Prov}(\semiringVars), \annot_X, \fact)$ obtained by setting the annotations of facts in $\Delta$ to $\zero$, which makes $\sum_{\{\ann(t) \mid t\in T_{\Delta}\}} \ann(t)$ evaluate to $0$.
\end{proof}

\section{Other Proofs for Section \ref{sec:analysis}}

\subsection{ $\provmodtot$ and $\provmodmon$ on $\mathbb{B}\llbracket\semiringVars\rrbracket$}\label{app:booleanseriesformodlbased}

\begin{proposition}
	The semiring $\mathbb{B}\llbracket\semiringVars\rrbracket$ of formal power series with Boolean coefficients is such that for every $\omega$-continuous $\plus\, $-idempotent semiring $\semiringshort$  and $\ontology$, $(\database, \semiringshort, \annot)$ and $\alpha$, for $\prov\in\{\provmodtot,\provmodmon\}$, $$\prov(\ontology, \database,\semiringshort, \annot, \alpha)=h(\prov(\ontology, \database,\mathbb{B}\llbracket\semiringVars\rrbracket,\annot_\semiringVars,\alpha))$$ 
	where $\annot_\semiringVars$ associates a distinct variable from $\semiringVars$ to each fact of $\database$ and $h$ is the unique semiring homomorphism that extends $\nu:\semiringVars\rightarrow K$ where $\nu(x)=\annot(\annot_\semiringVars^-(x))$ for every $x\in \semiringVars$. 
\end{proposition}
\begin{proof}
	Let $\semiringshort$ be a $\omega$-continuous $\plus\, $-idempotent semiring.
	Let $\ontology$ be a Datalog program, $(\database, \semiringshort, \annot)$ be an annotated database and $\alpha$ be a fact. 
	\begin{itemize}
		\item By Proposition \ref{prop:modelBasedSameAsTreeBasedForIdempotent}, since $\semiringshort$ is $\omega$-continuous and $\plus\, $-idempotent, $\provmodtot(\ontology, \database, \semiringshort, \annot,\alpha)=\provmodmon(\ontology, \database, \semiringshort, \annot,\alpha)= \atann(\ontology, \database, \semiringshort, \annot,\alpha)$. 
		
		\item Since $\atann$ satisfies the Commutation with $\omega$-Continuous Property and $\mathbb{N}^\infty\llbracket\semiringVars\rrbracket$ is universal for $\omega$-continuous semirings, 
		$\atann(\ontology, \database, \semiringshort, \annot,\alpha)=\mn{Eval}_\nu(\atann(\ontology, \database, \mathbb{N}^\infty\llbracket\semiringVars\rrbracket, \annot_\semiringVars,\alpha))$ where $\mn{Eval}_\nu$ is the unique semiring homomorphism from $\mathbb{N}^\infty\llbracket\semiringVars\rrbracket$ to $\semiringshort$ that extends $\nu:\semiringVars\rightarrow K$ where $\nu(x)=\annot(\annot_\semiringVars^-(x))$ for every $x\in \semiringVars$. 
		
		\item Let $f:\mathbb{N}^\infty\llbracket\semiringVars\rrbracket\mapsto \mathbb{B}\llbracket\semiringVars\rrbracket$ be the function that replaces all coefficients different from $0$ by $1$. 
		Since $\semiringshort$ is $\plus\, $-idempotent, for every $s\in \mathbb{N}^\infty\llbracket\semiringVars\rrbracket$, $\mn{Eval}_\nu(s)=\mn{Eval}_\nu(f(s))$. 
		Moreover, it is easy to check that $f$ is actually the unique $\omega$-continuous homomorphism of semirings such that for the one-variable monomials we have $f(x)=x$. 
		Hence since $\atann$ satisfies the Commutation with $\omega$-Continuous Property and $\mathbb{N}^\infty\llbracket\semiringVars\rrbracket$ is universal for $\omega$-continuous semirings, $\atann(\ontology, \database,\mathbb{B}\llbracket\semiringVars\rrbracket, \annot_\semiringVars,\alpha)=f(\atann(\ontology, \database,\mathbb{N}^\infty\llbracket\semiringVars\rrbracket, \annot_\semiringVars,\alpha))$.
		
		\item By Proposition \ref{prop:modelBasedSameAsTreeBasedForIdempotent}, since $\mathbb{B}\llbracket\semiringVars\rrbracket$ is $\omega$-continuous and $\plus\, $-idempotent, 
		$\provmodtot(\ontology, \database, \mathbb{B}\llbracket\semiringVars\rrbracket, \annot_\semiringVars,\alpha)= \provmodmon(\ontology, \database, \mathbb{B}\llbracket\semiringVars\rrbracket, \annot_\semiringVars,\alpha)= \atann(\ontology, \database,\mathbb{B}\llbracket\semiringVars\rrbracket, \annot_\semiringVars,\alpha)$.
		
		\item Finally, note that for every $s\in  \mathbb{B}\llbracket\semiringVars\rrbracket$, $\mn{Eval}_\nu(s)=h(s)$. 
	\end{itemize}
	To sum up, for $\prov\in\{\provmodtot,\provmodmon\}$ we have
	\begin{align*}
		\prov(\ontology, \database, \semiringshort, \annot,\alpha)=& \atann(\ontology, \database, \semiringshort, \annot,\alpha) \\
		=& \mn{Eval}_\nu(\atann(\ontology, \database, \mathbb{N}^\infty\llbracket\semiringVars\rrbracket, \annot_\semiringVars,\alpha)) \\
		=& \mn{Eval}_\nu(f(\atann(\ontology, \database, \mathbb{N}^\infty\llbracket\semiringVars\rrbracket, \annot_\semiringVars,\alpha))) \\
		=& \mn{Eval}_\nu(\atann(\ontology, \database,\mathbb{B}\llbracket\semiringVars\rrbracket, \annot_\semiringVars,\alpha))\\\
		=& \mn{Eval}_\nu(\prov(\ontology, \database, \mathbb{B}\llbracket\semiringVars\rrbracket, \annot_\semiringVars,\alpha))\\
		=& h(\prov(\ontology, \database, \mathbb{B}\llbracket\semiringVars\rrbracket, \annot_\semiringVars,\alpha))
		\qedhere
	\end{align*}
\end{proof}

\section{Provenance Computation}
\label{appendix:computation}

The problem of computing the provenance of a fact or query answer has been studied in different manners for $\atann$. 
As presented before, in this case the provenance expressions can be infinite. 
\citeauthor{DeutchMRT14} \shortcite{DeutchMRT14} study different semirings for which the provenance is finite and can be computed in a finite time. They also show that it is possible to represent $\atann$ with these semirings through a polynomial structure: circuits. 
We show how to adapt the classical semi-naive evaluation algorithm for the different provenance semantics studied in this paper.

\subsection{$\mdtann$ and $\rmdtann$ Cases}

We show that both $\mdtann$ and $\rmdtann$ can be represented by polynomial size circuits 
regardless of the semiring. 
We use a generalization of the algorithm presented by \citeauthor{DeutchMRT14} \shortcite{DeutchMRT14} for which we will use arithmetic circuits and not Boolean circuits to represent the provenance expressions.

\begin{definition} Let $\semiringVars$ be a set of annotations.
	An \emph{arithmetic circuit} $C$ is a pair  of a directed acyclic graph  $G$ and a labeling function $\gamma$ from the nodes of $G$ to $\semiringVars \cup \{+,*\} \cup \{0,1\}$.
	The nodes without outgoing edges are called leaves and the other nodes internal nodes. The labeling function associates each internal node to $+$ or $*$ and the leaves to variables in $\semiringVars$ or $1$ or $0$.
	The root of circuit is the only node without incoming edges.
\end{definition}

In the next proposition, we consider the semiring $\mathbb{N}^\infty\llbracket\semiringVars\rrbracket$. Thanks to the property of commutation with $\omega$-continuous homomorphisms, it can be extended to another semiring $\semiringshort$ by applying the homormophism from $\mathbb{N}^\infty\llbracket\semiringVars\rrbracket$ to $\semiringshort$ to the obtained circuit.

\begin{proposition}
	\label{lem:acircuit}
	For every fact $\fact$ and every $i\geq 0$, the annotation $\naiveannot^i(\fact)$ of $\fact$ in $(\naive^{i},\mathbb{N}^\infty\llbracket\semiringVars\rrbracket,\naiveannot^i) =\naive^{i}(\ontology,\database,\mathbb{N}^\infty\llbracket\semiringVars\rrbracket,\annot)$ can be represented by a circuit of size polynomial in the size of $\database$ and $i$. 
	Moreover, $\mdtann(\ontology,\database,\mathbb{N}^\infty\llbracket\semiringVars\rrbracket,\annot,\alpha)$ and $\rmdtann(\ontology,\database,\mathbb{N}^\infty\llbracket\semiringVars\rrbracket,\annot,\alpha)$ can be computed in a polynomial time in the size of $\database$.
\end{proposition}
\begin{proof}
	We generalize the algorithm proposed by \citeauthor{DeutchMRT14} \shortcite{DeutchMRT14} to construct Boolean circuits for $\mi{PosBool}(\semiringVars)$ of size polynomial in the database by using arithmetic circuits instead of Boolean circuits to represent the provenance. 
	
	We inductively describe an algorithm that constructs polynomial circuits $\mathcal{C}(i,\alpha,\ontology,\database,\mathbb{N}^\infty\llbracket\semiringVars\rrbracket,\annot)$ representing $\naiveannot^i(\fact)$ for every $\fact$. 
	
	\begin{itemize}
		\item Base case: $i=0$. In this case, for every $\fact$, $\naiveannot^i(\fact)=\annot(\fact)$ if $\fact\in\database$ and $\naiveannot^i(\fact)=0$ otherwise. Hence $\mathcal{C}(i,\alpha,\ontology,\database,\mathbb{N}^\infty\llbracket\semiringVars\rrbracket,\annot)$ consists of a single node labeled by $\annot(\fact)$ or $0$. 
		
		\item Induction step: Assume that for every $\fact$, we have built $\mathcal{C}(i,\alpha,\ontology,\database,\mathbb{N}^\infty\llbracket\semiringVars\rrbracket,\annot)$ polynomial in the size of $\database$ and $i$ that represents $\naiveannot^i(\fact)$. 
		
		We apply $\consop$ over the database $(\naive^{i},\mathbb{N}^\infty\llbracket\semiringVars\rrbracket,\annot')$ where $\annot'$ is a function associating to each fact in $\naive^{i}$ a new variable in a new set of variables $\semiringVars'$, and get a new annotated database denoted by $(\inter_{\consop}, \mathbb{N}^\infty\llbracket\semiringVars'\rrbracket, \annot_{\consop})$.
		
		To compute $\mathcal{C}(i+1,\alpha,\ontology,\database,\mathbb{N}^\infty\llbracket\semiringVars\rrbracket,\annot)$,  we use $\annot_{\consop}(\alpha)$.  It is known that there is a circuit representation of $\annot_{\consop}(\alpha)$ polynomial in $\inter_{\consop}$. 
		We replace each variable in  $\annot_{\consop}(\alpha)$ that represents some fact $\beta$ by the root of the circuit $\mathcal{C}(i,\beta,\ontology,\database,\mathbb{N}^\infty\llbracket\semiringVars\rrbracket,\annot)$. 
		Note that we add only a polynomial number of nodes and edges to the previous circuit. This conclude our induction.
	\end{itemize}
	
	Our proof can be easily adapted in the context of $\semi^{k}(\ontology,\database,\mathbb{N}^\infty\llbracket\semiringVars\rrbracket,\annot) $ and $\opti^{k}(\ontology,\database,\mathbb{N}^\infty\llbracket\semiringVars\rrbracket,\annot)$.
	Moreover, the annotation of a fact in $\semi^{k}(\ontology,\database,\mathbb{N}^\infty\llbracket\semiringVars\rrbracket,\annot) $  is not modified after the its creation, and the annotation of the goal fact in  $\opti^{k}(\ontology,\database,\mathbb{N}^\infty\llbracket\semiringVars\rrbracket,\annot)$ is not modified after its creation. Therefore,$\semi^{k}(\ontology,\database,\mathbb{N}^\infty\llbracket\semiringVars\rrbracket,\annot) $ and $\opti^{k}(\ontology,\database,\mathbb{N}^\infty\llbracket\semiringVars\rrbracket,\annot)$ terminates in a number of steps polynomial in the size of the database and we can conclude that $\mdtann(\ontology,\database,\mathbb{N}^\infty\llbracket\semiringVars\rrbracket,\annot,\alpha)$ and $\rmdtann(\ontology,\database,\mathbb{N}^\infty\llbracket\semiringVars\rrbracket,\annot,\alpha)$ can be computed in polynomial time in the size of the database.
\end{proof}

\subsection{$\mtann$ Case}
We show that there is no
polynomially computable
circuit 
that computes $\mtann$ on $\mathbb{N}^\infty[\semiringVars]$, by a reduction from the problem of counting the number of simle paths for the RPQ $a^*$, which is $\# P$ - hard \cite{DBLP:conf/www/ArenasCP12}. We hence start by defining these notions.

RPQ queries are binary queries over labeled-edge graphs. They are based on a regular language that can be defined by a deterministic automaton. 
A deterministic automaton is a tuple $\mathbb{A} = (s_0,S,F,A,\delta)$, where $S$ is a set of states, $F\subseteq S$ a set of final states, $A$ is a finite set called alphabet, and $\delta$ is a complete function from $S \times A$ to $S$. We extend $\delta$ to a complete function from $S \times A^*$ to $S$.

A path predicate is a binary predicate $\Lambda$ given by a regular expression over binary predicates. A path atom is an atom of the shape $\Lambda(t_1,t_2)$, where $\Lambda$ is a path predicate and $t_1, t_2$ are terms. A regular path query is a query $q(t_1,t_2) := \Lambda(t_1,t_2)$, where $\Lambda(t_1,t_2)$ is a path atom. 
Let $\mathcal{I}$ be an interpretation. We call \emph{path} (from $e_0$ to $e_n$) in $\mathcal{I}$ a (finite) sequence $p=e_0r_1e_1\ldots r_ne_n$ with $n\geq 0$ such that $e_0 \in \Delta_\mathcal{I}$, and for any $i \geq 1$, $e_i \in \Delta_\mathcal{I}$, $r_i$ is a binary predicate and $(e_{i-1},e_i) \in r^\mathcal{I}$, and denote by $w(p)$ the word $r_1\ldots r_n$ . 
We extend interpretations by interpreting path predicates as follows:
\begin{align*}
	&\Lambda^\mathcal{I} = \{(e_0,e_n) \mid \textrm{there exists a path $p$ from $e_0$ to $e_n$} \textrm{in $\mathcal{I}$ such that $w(p) \in \mathcal{L}(\Lambda)$}\}
\end{align*}
A path is simple if for any $i \not = j$, $e_i \not = e_j$. 
Given a deterministic automaton $\mathbb{A}$, a path is $\mathbb{A}$-simple if for any $i \not = j$, $(e_i,\delta(s_0,r_1\ldots r_{i-1})) \not = (e_j,\delta(s_0,r_1\ldots r_{j-1}))$.

There is a straightforward translation of an RPQ in a Datalog program, in particular for the RPQ $a*$.
\begin{definition}
	Let $\mathbb{A}$ be a deterministic automaton. The Datalog program $\Sigma_\mathbb{A}$ canonically associated with $\mathbb{A}$ contains, for each transition $\delta(s,a) = s'$, the following Datalog rule:
	\[s(x) \wedge a(x,y) \rightarrow s'(y). \]
	
	Moreover, for each final state $s_f \in F$, one Datalog rule is added:
	\[s_f(x) \rightarrow \mathsf{accept}(x)\]
\end{definition}

Given a databse $\database$, there is a strong connection between the non recursive proof trees of $\mathsf{accept}(a_2)$ \wrt $\Sigma_\mathbb{A}$ and $\database$ and the simple paths between $a_1$ and $a_2$ in $\database$.
\begin{proposition}\label{prop:link-path-nrt}
	There is a bijection between:
	\begin{enumerate}
		\item the set of all derivation trees of $\mathsf{accept}(a_2)$ w.r.t. $D \cup \{s_0(a_1)\}$ and $\Sigma_\mathbb{A}$, 
		\item the set of $\mathcal{L}(A)$ paths from $a_1$ to $a_2$ in $D$.
	\end{enumerate}
	
	There is also a bijection between:
	\begin{enumerate}
		\item[(3)] the set of non recursive derivation trees of $\mathsf{accept}(a_2)$ w.r.t. $D \cup \{s_0(a_1)\}$ and $\Sigma_\mathbb{A}$, 
		\item[(4)] the set of $\mathbb{A}$-simple $\mathcal{L}(A)$-paths from $a_1$ to $a_2$ in $D$.
	\end{enumerate}
	
	Therefore, $\mtann(\Sigma_\mathbb{A},D \cup \{s_0(a_1)\},\mathbb{N},\annot, \mathsf{accept}(a_2))$ where $\annot(\fact)=1$ for every $\fact$ gives the exact number of simple paths from $a_1$ to $a_2$ in $D$ satisfying $a^*$.
\end{proposition}

We are now ready to proof our result.
\begin{proposition}
	\label{th:datalognrt}
	Under the assumption that $P \neq NP$, there exist $\ontology$, $\database$ and $\fact$ such that $\mtann(\ontology,\database,\mathbb{N}^\infty\llbracket\semiringVars\rrbracket,\annot_X,\fact)$ cannot be represented by a circuit computable in polynomial time and of polynomial size in $\database$.
\end{proposition}
\begin{proof}
	Suppose by contradiction that for every $\ontology$, $\database$ and $\fact$ we can compute in a polynomial time an arithmetic circuit representing $\mtann(\ontology,\database,\mathbb{N}^\infty\llbracket\semiringVars\rrbracket,\annot_X,\fact)$. It implies that the circuit has a polynomial size. 
	By applying this assumption to compute $\mtann(\Sigma_\mathbb{A},\database\cup \{s_0(a_1)\},\mathbb{N}^\infty\llbracket\semiringVars\rrbracket,\annot_X, \mathsf{accept}(a_2))$ and using the homomorphism from $\mathbb{N}^\infty\llbracket\semiringVars\rrbracket$ into $(\mathbb{N}^\infty,+,\times,0,1)$ that associate each variable to 1 over this circuit, since $\mtann$ satisfies the commutation with $\omega$-continuous homomorphisms, it is possible to compute the number of simple paths from  $a_1$ to $a_2$ in $\database$ satisfying $a^*$ by Proposition \ref{prop:link-path-nrt}. 
	This contradicts the result by \citeauthor{DBLP:conf/www/ArenasCP12} \shortcite{DBLP:conf/www/ArenasCP12}, which concludes our proof.
\end{proof}

\end{document}